\documentclass[a4paper,runningheads]{llncs}

\title{Urgency Annotations for Alternating Choices}

\author{Eren Keskin \and Roland Meyer \and Sören van der Wall}

\institute{TU Braunschweig, \email{\{e.keskin, roland.meyer, s.van-der-wall\}@tu-bs.de}}

\usepackage{comment}
\usepackage{tabu}
\usepackage{float}
\usepackage{paralist}
\usepackage[inline]{enumitem}
\usepackage{listings}
\usepackage{environ}
\usepackage{tabularx}

\usepackage{scalefnt}
\usepackage{mathtools}
\usepackage{amsfonts}

\usepackage{amsthm}
\usepackage{amsmath}
\usepackage{amssymb}

\usepackage{scalerel}
\usepackage{ebproof}
\usepackage[format=hang]{subcaption}
\usepackage[T1]{fontenc}
\usepackage{lmodern}
\usepackage{stmaryrd}
\usepackage[hidelinks]{hyperref}
\hypersetup{
  hidelinks,
  colorlinks=true,
  allcolors=black,
  pdfstartview=Fit,
  breaklinks=true
}

\usepackage{cleveref}
\crefname{axiom}{axiom}{axioms}

\usepackage{tikz}
\usetikzlibrary{shapes,automata,positioning,scopes,math}
\tikzstyle{cuttingset}=[fill = gray]
\tikzstyle{pfirst}=[draw, trapezium]
\tikzstyle{psecond}=[draw, trapezium, shape border rotate = 180]
\tikzstyle{peve}=[draw, ellipse]
\tikzstyle{padam}=[draw, rectangle]
\tikzstyle{punk}=[draw, rounded corners]

\usepackage{xcolor} 
\usepackage{comments}
\usepackage{style}
\usepackage{style_testing}

\begin{document}


\maketitle


\vspace{-0.5cm}
\begin{abstract}
We propose urgency programs, a new programming model with support for alternation, imperfect information, and recursion.
The novelty are urgency annotations that decorate the (angelic and demonic) choice operators and control the order in which alternation is resolved.
We study standard notions of contextual equivalence for urgency programs.
Our first main result are fully abstract characterizations of these relations based on sound and complete axiomatizations.
%
%
Our second main result settles their computability via a normal form construction.
Notably, we show that the contextual preorder is \kexptime{($2\maxurg-1$)}-complete for programs of maximal urgency~$\maxurg$ when the regular observable is given as an input resp.\ \ptime-complete when the regular observable is fixed. 
We designed urgency programs as a framework in which it is convenient to formulate and study verification and synthesis problems. 
We demonstrate this on a number of examples including the verification of concurrent and recursive programs and hyper model checking. 

\end{abstract}

\vspace{-0.5cm}
\section{Introduction}
Algorithmic program verification may seem like a zoo of approaches that do not have much in common. 
Refinement is checked by establishing a relation between the program and a reference implementation~\cite{Milner71}, the model checking of linear-time properties~\cite{Pnueli77} is formulated as language inclusion~\cite{VW86}, and branching-time properties are reduced to games~\cite{EJ91}. 
This variety makes it difficult to choose a verification approach when one is confronted with a new class of programs. 
Fortunately, there are master problems that provide guidance. 
When the new class of programs is concurrent, there is a good chance that the verification problem can be cast as coverability in well-structured transition systems~\cite{Finkel87,AJ93,FS01}. 
When it is recursive, one would try to reduce it to higher-order model checking~\cite{Ong06,KO09}. 
The master problem then informs us about how to implement the verification algorithm~\cite{KSU11}, from the symbolic representation of program states~\cite{FG09} to the search strategy~\cite{IIC13}. 
But the master problems are no silver bullet. 
Well-structured transition systems are not good at modeling recursion, and the verification of branching-time properties is undecidable for them. 
Higher-order models, in turn, are not good at modeling concurrency. 

We propose a new programming construct to capture verification tasks that are not handled well by the existing master problems.   
The key insight is that by \emph{combining alternation~\cite{CKS81} with imperfect information~\cite{Reif84}} one can model concurrency even in a sequential programming model. 
%
%
%
%
%
%
%
We combine the two by adding urgency annotations to the (angelic and demonic) choice operators.  
%
%
Choices are classically resolved in program order~\cite{CKS81}. 
Urgency annotations are natural numbers that define when a choice has to be made: the higher the urgency, the sooner. 
The lower urgency choices remain unresolved in the program term until all higher urgency choices have been made. 
Choices of different urgency can thus be imagined as belonging to different components in the verification task, the program threads or the specification that has to judge the program behavior. 
%
%
%
%
%
%
To illustrate the order in which choices with urgency are made, consider the transition sequence 
\begin{align*}
(\aletter_l \achoicen{1} \aletter_r)\appl\makeleading{(\aletterp \echoicen{2} \aletterpp)}
\quad\gamemove\quad 
\makeleading{(\aletter_l \achoicen{1} \aletter_r)}\appl \aletterp\quad
\gamemove\quad  \aletter_l\appl \aletterp\ .
\end{align*}
The demonic choice $\achoicen{1}$ of urgency $1$ is resolved only after the angelic choice $\echoicen{2}$ of urgency $2$, although the demonic choice is written earlier in the program.

Urgency programs not only serve as a convenient backend to which to reduce verification tasks.
By comparing the reductions, we also understand how the tasks are related. 
As a running example, we study the formulation of simulation and language inclusion as urgency programs.
Let the systems of interest be $\adfa$ and $\adfap$ from \Cref{Figure:InclusionVsSimulation}.
We examine whether \begin{enumerate*}[label=(\emph{\alph*})]
\item $\adfa$ is simulated by $\adfap$ and
\item the language of~$\adfa$ is included in the language of~$\adfap$.
\end{enumerate*}
The transitions $\aletter_l$ and $\aletter_r$ carry the observation $\aletter$. 
%
%
The programs $\termsim$ and~$\termincl$ from \Cref{Figure:InclusionVsSimulation} encode the simulation resp.\@ the inclusion problem.
For now, just observe that the difference in these programs is merely a shift in urgency.
%
%
%
%
%
\begin{figure}
    \centering
    \begin{tikzpicture}[scale=0.7]
        \tikzmath{
            \xoffset=6.5;
            \ytxt=-0.8; \ytxtp=\ytxt-0.7;
            \distdiag=0.7;
            \distvert=0.5;
        }
        \tikzset{every state/.style={shape = circle,fill = black,minimum size = 4,inner sep=0}}
        {
            \node[state] (lin) at (\xoffset,0) {};
            \node[state] (laleft) [below left=\distdiag of lin] {};
            \node[state] (laright) [below right=\distdiag of lin] {};
            \node[state] (lableft) [below=\distvert of laleft] {};
            \node[state] (lacright) [below=\distvert of laright] {};
            \node (laut) [right=0.6cm of lin] {$\adfap$};
            \node (ltxt) [left=0.6cm of lin] {$\objective_{\adfap} = \set{ \aletter_l \aletterp, \aletter_r \aletterpp}$};
            \path[->]   (lin)    edge node[left]{$\aletter_l$} (laleft)
                                edge node[right]{$\aletter_r$} (laright)
                        (laleft) edge node[right]{$\aletterp$} (lableft)
                        (laright)edge node[left]{$\aletterpp$} (lacright);    
        }
        {
            \node[state] (rin) at (-\xoffset,0) {};
            \node[state] (ra) [below=\distvert of rin] {};
            \node[state] (rableft) [below left=\distdiag of ra] {};
            \node[state] (racright) [below right=\distdiag of ra] {};
            \node (raut) [left=0.6cm of rin]{$\adfa$};
            \node (rtxt) [right=0.6cm of rin] {$\objective_{\adfa} = \set{ \aletter \aletterp, \aletter \aletterpp}$};
            \path[->]   (rin)    edge node[left]{$\aletter$}   (ra)
                        (ra)     edge node[left]{$\aletterp$}  (rableft)
                                 edge node[right]{$\aletterpp$}(racright);
        }
        \node (txt) at (0,\ytxt) {$\termsim = \aletter \appl \makeleading{(\aletter_l \achoicen{1} \aletter_r)}
                \appl (\aletterp \echoicen{1} \aletterpp) \appl (\aletterp \achoicen{1} \aletterpp)$};
        \node (txt2) at (0,\ytxtp) {%
            $\termincl = \aletter \appl (\aletter_l \achoicen{1} \aletter_r) 
                \appl \makeleading{(\aletterp \echoicen{2} \aletterpp)} \appl (\aletterp \achoicen{1} \aletterpp)$};
    \end{tikzpicture}
    \caption{\label{Figure:InclusionVsSimulation}
        Automata $\adfa$, $\adfap$ and programs $\termsim, \termincl$ for simulation resp. inclusion. 
    }
\end{figure}
%

Urgency programs can only provide guidance for solving verification tasks if their verification problem is decidable and the decision procedure is easy to adapt to the new setting.  
To develop decision procedures for urgency programs, we follow a common approach: understand the contextual equivalence~\cite{Milner77} of the programming model, then one understands how to summarize programs, and every algorithm will be a fixed point over these summaries. 
%
%
%
%
The notion of contextual equivalence depends on the level of detail at which one intends to observe the program behavior.
We consider two standard definitions~\cite{JM21}: 
\begin{alignat*}{5}
\aterm\ &\congeq\ \atermp\quad &\text{if} &\quad \forall\objective.\forall \acontext{\contextvar}. &\quad \winsof{\acontext{\aterm}}{\objective}\quad&\text{iff}\quad \winsof{\acontext{\atermp}}{\objective}\\
\aterm\ &\speccongeq{\objective}\ \atermp\quad &\text{if} &\quad \phantom{\forall\objective.}\forall \acontext{\contextvar}.&\quad\winsof{\acontext{\aterm}}{\objective}\quad&\text{iff}\quad \winsof{\acontext{\atermp}}{\objective}\ . 
\end{alignat*}
The former definition quantifies over observables $\objective$ and contexts $\acontext{\contextvar}$ and is our notion of \emph{contextual equivalence}. 
The latter fixes an observable, like termination or reaching an error, and we refer to it as \emph{$\objective$-specialized contextual equivalence}. 
Due to the alternating choices in urgency programs, 
making an observation is defined in a game-theoretic way: 
$\winsof{\acontext{\atermp}}{\objective}$ means 
Eve has a winning strategy in the game arena $\semof{\acontext{\atermp}}$ 
when $\objective\subseteq\analph^*$ is the objective. 

Our first contribution are full abstraction results: contextual equivalence and its specialized variant coincide with congruence relations that neither quantify over contexts nor observables. 
The congruences are defined axiomatically, and one may say we axiomatize the contextual equivalences.
An important insight is that imperfect information distributes over perfect information. 
In the example,  
\begin{align*}
(\aletter_l \achoicen{1} \aletter_r)\appl
(\aletterp \echoicen{2} \aletterpp)\ \congeq\ 
(\aletter_l \achoicen{1} \aletter_r)\appl \aletterp \echoicen{2} (\aletter_l \achoicen{1} \aletter_r)\appl \aletterpp\ \congeq\ 
(\aletter_l\appl\aletterp \achoicen{1} \aletter_r\appl\aletterp)  \echoicen{2} (\aletter_l\appl \aletterpp \achoicen{1} \aletter_r\appl \aletterpp)\ .
\end{align*}

Our second contribution is to settle the complexity of the two contextual equivalences and their preorder variants. 
%
The main finding is that the specialized contextual preorder is 
\kexptime{($2\maxurg-1$)}-complete for programs of maximal urgency~$\maxurg$ when the regular observable is given as an input resp.\@ \ptime-complete when it is fixed.
A consequence is indeed that the verification problem can be solved with the same complexity. 
%
%
%
To get the upper bound right, an important idea is to factorize the set of contexts. 
We equate contexts $\acontext{\contextvar}$ that have the same solution space: the same set of programs $\aterm$ so that  $\winsof{\acontext{\aterm}}{\objective}$.  
The challenge is to handle the factorization algorithmically. 
We show how to represent solution spaces using the novel concept of characteristic terms. 
%

For the lower bound, we reduce context-bounded multi-pushdown games~\cite{Seth09} to the verification of urgency programs. 
This justifies our claim that urgency programs can capture concurrency in a rather natural way. 
To further demonstrate the usefulness of urgency programs, we show how to encode other popular verification problems, notably the recent hyper model checking~\cite{CS10,CFKMRS14}.

\paragraph{Structure.} After an introduction to urgency programs, we state the full abstraction results in \Cref{Section:FullAbstraction}, followed by the proofs of soundness, normalization, and completeness.  
In \Cref{Section:Verification}, we give the decidability and complexity results, followed by the upper and lower bound proofs in \Cref{Section:Applications,Section:UpperBound}. 
%

\vspace{-0.1cm}
\section{Programming Model}
Throughout the development, we fix a natural number $\maxurg>0$ for the maximal urgency in programs, 
a finite alphabet $\analph$ of terminal symbols with typical elements $\aletter, \aletterp, \aletterpp$, 
and a finite or infinite set of non-terminals~$\nonterminals$ with elements $\nonterminal, \nonterminalp, \nonterminalpp$. 
Each non-terminal has a so-called defining program term given by $\eqmap:\nonterminals\to \terms$.  

The set $\terms$ of \emph{program terms} (of urgency up to $\maxurg$ over~$\analph, \nonterminals$) is defined as 
\begin{align*}
    \aterm\quad ::=\quad \aletter \bnf \tskip \bnf \terr \bnf\nonterminal \bnf \aterm\appl\aterm
    \bnf  \bigEchoiceOf{\anurg}{}{\atermset} \bnf \bigAchoiceOf{\anurg}{}{\atermset}\ .
\end{align*}
Terminal symbols represent program commands with visible behavior, $\tskip$ is a command without visible effect, and $\terr$  aborts the computation unsuccessfully. 
Non-terminals model recursive functions, we have concatenation, and 
angelic~($\echoicen{\anurg}$) as well as demonic ($\achoicen{\anurg}$) choice of urgency~$0<\anurg\leq \maxurg$. 
We use $\pchoicen{\anurg}$ to mean $\echoicen{\anurg}$ or $\achoicen{\anurg}$, infix notation for binary choices, and $\bigPchoiceOf{\anurg}{}{\aterm}$ for $\bigPchoiceOf{\anurg}{}{\set{\aterm}}$. 
An action is a term
$\nonterminal$ or $\bigPchoiceOf{\anurg}{}{\atermset}$.
Terms that contain actions are called active.
Terms that do not contain actions are called passive. 
Passive terms are 
words over $\Sigma\cup\set{\tskip, \terr}$, and we also call them word terms $\aword \in\wordterms$. 
%
To avoid brackets, we let concatenation bind stronger than choices. 
Choices range over non-empty but possibly infinite sets of terms.
Infinitary syntax requires care to make sure set and game-theoretic concepts remain sound. 
We moved the corresponding lemmas to \Cref{Appendix:Infinitary}~and~\ref{Appendix:TreeBounds} to keep the presentation light.
The motivation for infinitary syntax will be given at the end of the section. 
We lift the notion of urgency from choices to program terms and define $\urgency: \terms\to\nat$ by
\begin{alignat*}{5}
    \urgencyof{\tskip}&= \urgencyof{\terr} = \urgencyof{\aletter}=0&\qquad
    \urgencyof{\nonterminal}&=\maxurg\\
\urgencyof{\aterm\appl\atermp}&=\max\set{\urgencyof{\aterm}, \urgencyof{\atermp}}&\qquad     \urgencyof{\bigPchoiceOf{\anurg}\atermset}&=\anurg\ .
\end{alignat*}

A \emph{context} $\acontext{\contextvar}$ (of maximal urgency $\maxurg$ over $\analph$ and $\nonterminals$) is a term that contains at most one occurrence of the fresh non-terminal $\contextvar$.
The set of all contexts is denoted by $\contexts$.
The expression $\acontext{\aterm}$ refers to the term 
obtained from $\acontext{\contextvar}$ by replacing $\contextvar$ with the term $\aterm$.
A term $\atermp$ is a subterm of $\aterm$ if there is a context 
$\acontext{\contextvar}$ with $\aterm=\acontext{\atermp}$.
A subterm is called outermost if it is not enclosed by a choice.
For example, $\aterm$ is outermost in $\aterm\appl\atermp$ 
but not in $(\aterm\echoicen{\anurg}\atermp)\appl\atermpp$.
%
%
%
The \emph{leading subterm} $\leadingof{\aterm}$ of a term 
$\aterm$ is defined as the outermost action with the highest urgency. 
If several outermost actions have this urgency, then the leftmost of them is 
leading.
Passive terms do not have leading subterms. 
Where helpful, we will underline a subterm that contains the leading subterm.
Note that $\makeleading{\aterm}\appl\atermp$ implies $\urgencyof{\aterm}\geq\urgencyof{\atermp}$
and 
$\aterm\appl\makeleading{\atermp}$ implies $\urgencyof{\atermp}>\urgencyof{\aterm}$. 
We denote the unique context enclosing the leading subterm
$\leadingof{\aterm}$
in $\aterm$ by $\enclosingctx{\aterm}{\contextvar}\in\contexts$. 
We have $\aterm = \enclosingctx{\aterm}{\leadingof{\aterm}}$. 
\begin{example}
	Consider $\termincl$ from \Cref{Figure:InclusionVsSimulation}. 
	Its urgency is $\urgencyof{\termincl} = 2$, 
	the leading subterm is $\aletterp \echoicen{2} \aletterpp$, 
	and the enclosing context is $\enclosingctx{\termincl}{\contextvar} = 
	\aletter \appl (\aletter_l \achoicen{1} \aletter_r) 
                \appl \contextvar \appl (\aletterp \achoicen{1} \aletterpp)$.
	Further, $\aletter_l \achoicen{1} \aletter_r$ is an outermost subterm of $\termincl$ and $\aletter_l$ is not.
\end{example}
\vspace{-0.2cm}
\subsection{Semantics}
Given that our programming model has alternation, the operational semantics of a term is not a plain transition system but a \emph{game arena} $\semof{\aterm}=(\terms, \aterm, \own, \gamemove)$ in which positions are owned by player $\eve$ or player $\adam$.  
The set of positions is the set of all terms.
The initial position is the given term. 
The ownership assignment $\own:\terms\to \set{\eve, \adam}$ returns the owner of the leading subterm, $\ownof{\aterm}=\ownof{\leadingof{\aterm}}$.
Adam owns the demonic choices, $\ownof{\bigAchoiceOf{\anurg}\atermset}=\adam$, and Eve owns the angelic choices, $\ownof{\bigEchoiceOf{\anurg}\atermset}=\eve$. 
We also give $\tskip$, $\terr$, terminals, and non-terminals to Eve.
This has no influence on the semantics as there will be at most one move from these positions.
The set of moves is defined as the smallest relation satisfying the following rules: 
\begin{align*}
    \axdef{}{\aterm\in\atermset}{\bigPchoiceOf{\anurg}\atermset\gamemove\aterm}
    \hspace{1.3cm}
    \axdef{}{}{\nonterminal\gamemove\eqmapof{\nonterminal}}
    \hspace{1.3cm}
    \axdef{}{\leadingof{\aterm}\gamemove\atermp}
    {\aterm
    \gamemove\enclosingctx{\aterm}{\atermp}}\ .
\end{align*}

\noindent A move always rewrites the leading subterm as illustrated by \Cref{fig:opsemantics}.
For a choice, it selects one alternative.
For a non-terminal, it inserts the defining term.
We define $\successorsof{\aterm}=\setcond{\atermp}{\aterm\gamemove\atermp}$.
It is worth noting that the game arena has perfect information. 
Imperfect information is modeled through choices, and they are eventually resolved.
%

The operational semantics is intensional in that it gives precise information about the program state at runtime. 
From a programming perspective, what matters is the result of a computation or, more generally, the \emph{observable behavior} of the program. 
Due to Adam's influence, the observable behavior will rarely be a single word but rather a language $\objective\subseteq\analph^*$.
We write $\winsof{\aterm}{\objective}$ to mean that Eve can enforce termination and the result will be a word from~$\objective$, no matter how Adam plays. 
We make this precise. 

Our notion of observable behavior is based on concepts from game theory. 
We refer to a language $\objective\subseteq\analph^{*}$ as a reachability objective for the game arena~$\semof{\aterm}$.
A play in this arena is a maximal (finite or infinite) sequence of positions 
$\aplay = \aterm_{0}, \aterm_{1}, \ldots$ that starts in the given term, $\aterm_0=\aterm$, and respects the moves of the game arena, $\aterm_{i}\gamemove\aterm_{i+1}$ for all~$i$. 
If the play ends, the result is a word term $\aword \in \wordterms$.
We interpret it as an element of the monoid with zero $(\analph^* \cup \set{\terr}, \appl\;, \tskip, \terr)$. 
Here, $\tskip$ is the unit, often denoted by $\varepsilon$, and $\terr$ is the zero.
We use $\sigeq$ to denote the monoid equality.  
%
%
Eve wins the play when $\aword$ belongs to $\objective$, meaning there is $\awordp \in \objective$ so that $\aword \sigeq \awordp$ (write $\aword \in \objective$).
Otherwise, Adam wins the play.
In particular, Adam wins all infinite plays and all plays exhibiting $\terr\notin\analph$.

%
A positional strategy for Eve is a function $\strat: \terms\to\terms$
so that $\atermp\gamemove\strat(\atermp)$ for all terms $\atermp$ owned by Eve that admit further rewriting.
Since we are interested in reachability objectives, we can use positional strategies without loss of generality~\cite{Martin75}. 
A play $\aplay$ is conform to $\strat$ if for all $i$ with $\ownof{\aterm_{i}}=\eve$ and $\successorsof{\aterm_i}\neq\emptyset$ we have $\strat(\aterm_{i})=\aterm_{i+1}$.
Eve wins objective~$\objective$, if she has a strategy~$\strat$ so as to win all plays that are conform to this strategy. 
This is what we denote by $\winsof{\aterm}{\objective}$. 
\begin{example}\label{Example:SimulationInclusion}
	We demonstrate how to encode simulation and inclusion between $\adfa$ and $\adfap$ from \Cref{Figure:InclusionVsSimulation}.
	In $\termsim$, Eve is tasked to find a violation of the simulation property, while Adam tries to prove it:
	whenever Eve takes a transition in $\adfa$, Adam tries to select a simulating transition right away in $\adfap$.
	The term $\termincl$ models inclusion:
	Eve selects a run in $\adfa$ and Adam tries to come up with a run in $\adfap$ with matching observations. 
	We define an objective $\objective$ 
	so that $\winsof{\termsim}{\objective}$ holds precisely 
	when $\adfa$ is simulated by $\adfap$ 
	and $\winsof{\termincl}{\objective}$ holds precisely 
	when $\adfa$'s language is included in $\adfap$'s.
	The objective has to make sure that the choices of transitions form paths in the automata and that the letters actually match.
	The latter is easy to check. 
	The former can be guaranteed by interleaving $\objective_\adfa$ and $\objective_{\adfap}$ in an alternating fashion, so $a_1.a_2$ and $b_1.b_2$ yield $a_1.b_1.a_2.b_2$.
	Details are in Section~\ref{Section:Applications}.
\end{example}
\begin{figure}
	\begin{subfigure}[b]{0.5\linewidth}
		\centering
		\scalebox{0.85}{
		\begin{tikzpicture}[auto,
			level 1/.style={sibling distance=25mm, level distance=10mm},
			level 2/.style={sibling distance=13mm, level distance=10mm}]
		\node[padam]{\vphantom{$f$}
		$(\makeleading{\aletter_l\achoicen{1} \aletter_r}).(\aletterp\echoicen{1}\aletterpp)$}
		child{
			node[peve]{\vphantom{$f$}$\aletter_l\appl(\makeleading{\aletterp\echoicen{1}\aletterpp})$} 
			child{node[punk]{\vphantom{$f$}$\aletter_l\appl \aletterp$}} 
			child{node[punk,cuttingset]{\vphantom{$f$}$\aletter_l\appl\aletterpp$}}}
		child{
			node[peve]{\vphantom{$f$}$\aletter_{r}.(\makeleading{\aletterp \echoicen{1}\aletterpp})$} 
			child{node[punk,cuttingset]{\vphantom{$f$}$\aletter_{r}\appl \aletterp$}}
			child{node[punk]{\vphantom{$f$}$\aletter_{r}\appl c$}}};
		\end{tikzpicture}
		}
	\end{subfigure}%
	\begin{subfigure}[b]{0.5\linewidth}
		\centering
		\scalebox{0.85}{
		\begin{tikzpicture}[auto,
			level 1/.style={sibling distance=27mm, level distance=10mm},
			level 2/.style={sibling distance=13mm, level distance=10mm}]
		\node[peve]{\vphantom{$f$}$(\aletter_{l}\achoicen{1} \aletter_{r}).(\makeleading{b \echoicen{2} c})$}
			child{
				node[padam]{\vphantom{$f$}$(\makeleading{\aletter_{l}\achoicen{1} \aletter_{r}})\appl b$} 
					child{node[punk]{\vphantom{$f$}$\aletter_{l}\appl b$}} 
					child{node[punk,cuttingset]{\vphantom{$f$}$\aletter_{r}\appl b$}}}
			child{
				node[padam]{\vphantom{$f$}$(\makeleading{\aletter_{l}\achoicen{1} \aletter_{r}})\appl c$} 
					child{node[punk,cuttingset]{\vphantom{$f$}$\aletter_{l}\appl c$}} 
					child{node[punk]{\vphantom{$f$}$\aletter_{r}\appl c$}}};
		\end{tikzpicture}
		}
	\end{subfigure}
	\caption{\label{fig:opsemantics}The game arenas 
	$\semof{(\aletter_{l}\echoicen{1} \aletter_{r})\appl(b \achoicen{1} c)}$ and 
	$\semof{(\aletter_{l}\echoicen{1} \aletter_{r})\appl(b \achoicen{2} c)}$. $\objective$ is gray.}
\end{figure}%
\begin{example}\label{Example:ReducedTerms}
	The encodings for simulation and inclusion only differ in the terms 
	$(\aletter_l \achoicen{1} \aletter_r)\appl(\aletterp \echoicen{2} \aletterpp)$ and 
	$(\aletter_l \achoicen{1} \aletter_r)\appl(\aletterp \echoicen{1} \aletterpp)$.  
	Their semantics is given in \Cref{fig:opsemantics}.
	Rectangular nodes are owned by Adam, circular ones by Eve, and for rectangular nodes with rounded corners the ownership does not matter.
	The objective is $\objective = \set{\aletter_l\appl\aletterpp, \aletter_r\appl\aletterp}$.
	Indeed, $\adfa$'s language is included in $\adfap$'s, 
	$\notwinsof{(\aletter_{l}\achoicen{1} \aletter_{r})\appl(b \echoicen{2} c)}{\objective}$,
	but $\adfa$ is not simulated by $\adfap$,
	$\winsof{(\aletter_{l}\achoicen{1} \aletter_{r})\appl(b\echoicen{1} c)}{\objective}$.
	%
	%
	%
\end{example}
%
\subsection{Contextual Preorder}
The notion of observable behavior is not compositional: we may have $\winsof{\aterm}{\objective}$ if and only if $\winsof{\atermp}{\objective}$ for all objectives~$\objective$, 
yet the two terms behave differently when placed into a context.
%
%
In our example, $\winsof{(b\echoicen{1} c)}{\objectivep}$
if and only if $\winsof{(b\echoicen{2} c)}{\objectivep}$
for all objectives $\objectivep\subseteq\analph^{*}$.
When inserting the terms into the context 
$\acontext{\contextvar}=(\aletter_{l}\achoicen{1} \aletter_{r})\appl\contextvar$, however, we have the difference discussed above. 
%
This is a common problem, and the way out is to consider the largest 
congruence that lives inside observational equivalence.  
It is more elegant to work with a precongruence and define the congruence of interest 
as a derived notion.  
\begin{definition}
The \emph{contextual preorder} $\congleq\ \subseteq \terms\times\terms$ is defined by $\aterm\congleq\atermp$, if 
\begin{align*}
\forall \objective\subseteq\analph^*.\ \forall \acontext{\contextvar}\in \contexts.\  \winsof{\acontext{\aterm}}{\objective} \text{ implies }\winsof{\acontext{\atermp}}{\objective}\ . 
\end{align*}
The \emph{$\objective$-specialized contextual preorder} $\speccongleq{\objective}\ \subseteq \terms\times\terms$ is defined by fixing $\objective$ and dropping the leading universal quantifier. 
%
The \emph{contextual equivalence} is then 
$\congeq\ =\ \congleq\cap\conggeq$, and the \emph{\hbox{$\objective$-specialized contextual equivalence}} is 
$\speccongeq{\objective}\ =\ \speccongleq{\objective}\cap\specconggeq{\objective}$.
\end{definition}

In the example, $(b\echoicen{2} c)\congleq(b\echoicen{1} c)$
and so by congruence $(\aletter_{l}\achoicen{1} \aletter_{r})\appl(b\echoicen{2} c)\congleq (\aletter_{l}\achoicen{1} \aletter_{r})\appl(b\echoicen{1} c)$. 
The reverse does not hold, consider context $(\aletter_{l}\achoicen{1} \aletter_{r})\appl\contextvar$ and objective $\set{\aletter_{l}\appl c, \aletter_{r}\appl b}$ from above. 
Note that $\aterm\congleq\atermp$ implies $\aterm\speccongleq{\objective} \atermp$ for all $\objective$. 

With $\terminateobj = \analph^*$ as the objective, we can use the specialized contextual equivalence to study the termination behavior of programs. 
We can also introduce a letter $\reachletter$ so that
$\winsof{\aterm}{\reachobj}$ with $\reachobj = \analph^*\appl\reachletter\appl\analph^*$ observes visits to a specific location. 
(Non-specialized) Contextual equivalence is more precise and takes into account all objectives. 
Both notions are also motivated by verification, where contextual equivalence gives information about which information can be abstracted away from an urgency
term without an influence on the objective, similar to how bisimilarity preserves $\text{CTL}^*$ properties~\cite{HM85}.

%
%

The motivation for an infinite set of non-terminals and infinitary terms is to model parameterized functions in a simple yet general way. 
The idea is to introduce a non-terminal for each instantiation of the function's formal parameters by actual values, inspired by value passing in process algebra~\cite{BPS2001}. 
Another more technical argument is that our normal form relies on infinitary syntax.  
%
%
\section{Full Abstraction}\label{Section:FullAbstraction}
We define a precongruence $\axleq\ \subseteq\terms\times\terms$ on program terms that neither quantifies over contexts nor objectives but relates terms solely based on their syntactic structure. 
The relation is defined through a set of axioms that should be understood as explaining the interplay between the operators in our programming model. 
The main finding is that this axiomatic precongruence coincides with the contextual preorder, and we say that we axiomatize (in a sound and complete~way) the contextual preorder. 
This is our main theorem.
\begin{theorem}[Full Abstraction 1]\label{Theorem:FullAbstractionContextual}
$\aterm\axleq\atermp$ if and only if $\aterm\congleq \atermp$. 
\end{theorem}

We have a corresponding result for the $\objective$-specialized contextual preorder.
In this setting, a complete axiomatization is considerably more difficult to obtain because, intuitively, we have to understand the concatenation behavior of language~$\objective$. 
Our solution is partial in that we impose a side condition on the objective to obtain completeness: it should be right-separating, a notion we will define in a moment.
Luckily, the objectives of interest $\terminateobj$ and $\reachobj$
are right-separating. 
It is always sound to reason with the $\objective$-specialized axiomatic precongruence.
\begin{theorem}[Full Abstraction 2]\label{Theorem:FullAbstractionSpecialized}
$\aterm\axleq_{\objective}\atermp$ implies $\aterm\speccongleq{\objective}\atermp$.
If $\objective$ is right-separating, then also $\aterm\speccongleq{\objective} \atermp$ implies $\aterm\axleq_{\objective}\atermp$.
\end{theorem}


We understand the concatenation behavior of an objective with the help of the syntactic precongruence over the monoid $\analph^*\cup\set{\terr}$.
It may relate terminal words to $\terr$ in case they cannot be extended to a word from the objective. 
\begin{definition}
	The \emph{syntactic precongruence} induced by  $\objective$ on~$\analph^*\cup\set{\terr}$ is defined by $\aword\lrleq{\objective} \awordp$, if for all $\awordpp, \awordppp\in\analph^{*}$we have $\awordpp\appl\aword\appl\awordppp\in\objective$ implies $\awordpp\appl\awordp\appl\awordppp\in\objective$. 
\end{definition}
An objective is then right-separating, if the concatenation from left in the above definition is not needed to distinguish words. 
We define $\rleq{\objective}$ on $\analph^*\cup\set{\terr}$ by $\aword\rleq{\objective} \awordp$, if for all $\awordppp\in\analph^{*}$ we have $\aword\appl\awordppp\in\objective$ implies $\awordp\appl\awordppp\in\objective$. 
\begin{definition}[and Lemma]
Objective $\objective$ is \emph{right-separating}, if $\lrleq{\objective}\ =\ \rleq{\objective}$. 
The objectives $\terminateobj$ and $\reachobj$ are right-separating.
\end{definition} 

\begin{example}
	The syntactic congruence 
	${\lreq{\objective}} = {\lrleq{\objective} \cap \lrgeq{\objective}}$ 
	induced by the objective $\objective=\set{\aletter_{l}.c, \aletter_{r}.b}$ 
	has classes  $\factorize{\analph^*\cup\set{\terr}}{\lreq{\objective}}=\set{\classof{\tskip}, \classof{\aletter_{l}}, \classof{\aletter_{r}}, \classof{b}, \classof{c}, \classof{\aletter_{l}\appl c, \aletter_{r}\appl b}}$ plus a class for the remaining words.
\end{example}

Intuitively, right-separating objectives allow us to evaluate the $\objective$-specialized contextual preorder by using contexts $\contextvar\appl\atermpp$ that only append to the right. 
For arbitrary objectives, we have to consider contexts $\atermppp\appl\contextvar\appl\atermpp$ and it is difficult to understand the interplay between high urgencies in $\atermpp$ and low urgencies in $\atermppp$.

We now give the two axiomatizations and explain them on an intuitive level. 
Recall that a precongruence is a reflexive and transitive relation that is preserved when inserting related terms into the same context. 

\begin{definition}
The \emph{axiomatic precongruence} $\axleq\ \subseteq\terms\times\terms$
is the least precongruence satisfying the axioms in \Cref{fig:axioms} except \labelcref{axiom:spec}. 
The \emph{$\objective$-specialized axiomatic precongruence} $\specaxleq{\objective}$ on terms is the least precongruence satisfying all axioms in \Cref{fig:axioms}. 
We use $\axeq$ for $\axleq\cap\axgeq$ and $\specaxeq{\objective}$ for $\specaxleq{\objective}\cap\specaxgeq{\objective}$.
\end{definition}

\begin{figure}
	\begin{subfigure}[t]{0.58\columnwidth}
		\begin{subfigure}[t]{\columnwidth}
			\caption{Lattice}\label{axioms:lattice}
			\begin{axdefenv}{\rlatmntname}
				\axdefleft{\rlatmntname}{\forall i\in I.\;\aterm_{i}\axleq \atermp_{i}}
				{\bigPchoiceOf{\anurg}{}{\set{\aterm_{i}\mid i\in I}}\axleq \bigPchoiceOf{\anurg}{}{\setcond{\atermp_{i}}{i\in I}}}%
				\label{axiom:lattice-mono}
			\end{axdefenv}
\vspace{0.8em}
			{
			\begin{axdefenv}{\rlatdistname}
				\axdefleft{\rlatdistname}
				{\vphantom{\urgencyof{\aterm}}}
				{\bigEchoiceOf{\anurg}_{i\in I}{\bigAchoiceOf{\anurg}{}{\atermset_{i}}}
				\axeq\bigAchoiceOf{\anurg}_{f: I\to \atermset_{I}}{\bigEchoiceOf{\anurg}{}{\set{f (i)\mid i\in I}}}}%
				\label{axiom:lattice-dist}
			\end{axdefenv}}
\vspace{0.8em}
			{
			\begin{axdefenv}{\rlatabsname}
				{\axdefleft{\rlatabsname}{\urgencyof{\aterm}\leq\anurg}
				{\aterm\achoicen{\anurg} (\aterm\echoicen{\anurg}\atermp)\axeq\aterm
				\qquad
				\aterm\echoicen{\anurg} (\aterm\achoicen{\anurg}\atermp)\axeq\aterm}}%
				\label{axiom:lattice-abs}
			\end{axdefenv}}
\vspace{0.8em}			
			{
			\begin{axdefenv}{\rlatassocname}
				\axdefleft{\rlatassocname}{}
				{\bigPchoiceOf{\anurg}_{i\in I}{\bigPchoiceOf{\anurg}{}{\atermset_{i}}}
				\axeq\bigPchoiceOf{\anurg}{}{\bigcup_{i\in I}\atermset_{i}}}%
				\label{axiom:lattice-assoc}
			\end{axdefenv}}
			\hspace{-1.5em}
			{
			\begin{axdefenv}{\rlatordname}
				\axdefleft{\rlatordname}{\urgencyof{\aterm}\leq\anurg}{\aterm\axleq\aterm\echoicen{\anurg}\atermp}%
				\label{axiom:lattice-ord}
			\end{axdefenv}}
		\end{subfigure}%
	\end{subfigure}%
	\hfill%
	\begin{subfigure}[t]{0.42\columnwidth}
		\begin{subfigure}[t]{\columnwidth}
			\caption{Distributivity}\label{axioms:distributivity}
			{
			\begin{axdefenv}{\rdistlname}
				{\axdefleft{\rdistlname}{\urgencyof{\aterm}<\anurg}
				{\aterm\appl(\bigPchoiceOf{\anurg}{}{\atermsetp})\axeq\bigPchoiceOf{\anurg}{}{\setcond{\aterm\appl\atermp}{\atermp\in\atermsetp}}}}%
				\label{axiom:dist-left}
			\end{axdefenv}}
			\vspace{0.5em}			
			{
			\begin{axdefenv}{\rdistrname}
				\axdefleft{\rdistrname}
				{\urgencyof{\aterm}\leq\anurg}
				{(\bigPchoiceOf{\anurg}{}{\atermsetp})\appl\aterm
				\axeq\bigPchoiceOf{\anurg}{}{\setcond{\atermp\appl\aterm}{\atermp\in\atermsetp}}}%
				\label{axiom:dist-right}
			\end{axdefenv}}
		\end{subfigure}%
		
		\begin{subfigure}[t]{\columnwidth}
			\caption{Normalization}\label{axioms:normalization}
			\begin{axdefenv}{\rnormalformname}
				\axdefleft{\rnormalformname}
				{\anurgp<\anurg}
				{\bigEchoiceOf{\anurgp}{}{ \bigPchoiceOf{\anurg}{}{\atermsetp}}
				\axeq\bigEchoiceOf{\anurgp}{}{\bigPchoiceOf{\anurgp}{}{\atermsetp}}}%
				\label{axiom:norm}%
			\end{axdefenv}
		\end{subfigure}%

		\begin{subfigure}[t]{0.4\columnwidth}
			\caption{Err}
			\begin{axdefenv}{\rordbotname}
				{\axdefleft{\rordbotname}{}
				{\terr\axleq\aterm}}%
				\label{axiom:least}
			\end{axdefenv}
		\end{subfigure}%
		\hfill%
		\begin{subfigure}[t]{0.6\columnwidth}
			\caption{Monoid}
			\begin{axdefenv}{\rmonoidname}
				\axdefleft{\rmonoidname}{{\aword} \sigeq {\awordp}}
				{\aword\axeq\awordp}%
				\label{axiom:monoid}%
			\end{axdefenv}
		\end{subfigure}
	\end{subfigure}

	\begin{subfigure}[t]{0.68\columnwidth}
		\caption{Fixed point}%
		\label{axioms:lfp}%
		{
		\begin{axdefenv}{\rappleqname}
			\axdefleft{\rappleqname}{}{\nonterminal\axeq\eqmapof{\nonterminal}}%
			\label{axiom:eqmap}%
		\end{axdefenv}}
		{
		\begin{axdefenv}{\rfpname}
			\axdefleft{\rfpname}
			{\forall \nonterminal \in \nonterminals .\; \eqmapof{\nonterminal}\replace{\nonterminals}{\atermn{\nonterminals}}\axleq \atermn{\nonterminal}}
			{\nonterminalp\axleq\atermn{\nonterminalp}}%
			\label{axiom:lfp}%
		\end{axdefenv}}
	\end{subfigure}%
	\hfill%
	\begin{subfigure}[t]{0.25\columnwidth}
		\caption{Specialization}%
		\label{axioms:semigroup}
		{
		\begin{axdefenv}{\rtequivname}
			\axdefleft{\rtequivname}
			{\aword\lrleq{\objective}\awordp}
			{\aword\specaxleq{\objective}\awordp}%
			\label{axiom:spec}%
		\end{axdefenv}}
	\end{subfigure}%
	\caption{\label{fig:axioms}Axioms defining $\axleq$ and $\specaxleq{\objective}$.}
\end{figure}

With the axioms given in \Cref{axioms:lattice}, the choice operators span a 
completely distributive lattice on each urgency.
The monotonicity \cref{axiom:lattice-mono} is not covered by the precongruence but implements an infinite replacement. 
The axiom has a side condition that can be found \Cref{Appendix:Infinitary}. 
Due to this axiom, nodes in our proof trees may have an infinite degree. 
Yet, every path is guaranteed to be finite. 
To see the premise in axiom~\labelcref{axiom:lattice-ord}, 
consider $\aterm = (\aletter_{l}\achoicen{2} \aletter_{r})$ and context $\bullet\appl(b\echoicen{2} c)$. 
Then in $\aterm.(b\echoicen{2} c)$ Eve wins while in $(\aterm\echoicen{1}\aterm).(b\echoicen{2} c)$ she loses, similar to \Cref{fig:opsemantics}. 
%
%
%
%
For \labelcref{axiom:lattice-abs}, the reasoning is similar. 
As a consequence of the lattice axioms,
one can derive the dual rules of \labelcref{axiom:lattice-ord,axiom:lattice-dist}.
Distributivity \labelcref{axiom:lattice-dist} states that 
the order of choices can be changed by considering all
choice functions $f:I\to\bigcup_{i\in I}\atermset_{i}$ with $f(i)\in \atermset_{i}$
for all $i\in I$, denoted by $f:I\to \atermset_{I}$. 

The distributivity in \labelcref{axiom:dist-left} captures the essence of imperfect information: concatenation to the left distributes over choice, provided the internals of the term are invisible as the choice has a higher urgency.
The distributivity from the right in~\labelcref{axiom:dist-right} is similar but takes into account that the leading subterm for equal urgencies is leftmost.
This clean interplay between imperfect information and choice came as a surprise and we consider these laws an important contribution.

A string with $\terr\not\in\analph$ is the most disadvantageous term for Eve,
because it belongs to no objective $\objective\subseteq\analph^{*}$.
%

The monoid \cref{axiom:monoid} 
refers to word terms $\aword,\awordp \in \wordterms$. 
We interpret them in the monoid~$\analph^*\cup\set{\terr}$ and inherit the equality there, denoted by $\sigeq$ above. 
The equality strips brackets and $\tskip$, and interprets $\terr$ as zero.
%

The normalization \cref{axiom:norm} reflects the fact that only the outermost choice operator determines the urgency of a term. 
Towards soundess, note that once the outer choice with urgency~$\anurg$ is resolved, 
we are sure that the context to the left has urgency strictly smaller than $\anurg$ and the context to the right has urgency at most $\anurg$.  
Hence, the inner choice is the next to be resolved, independent of whether its
urgency is $\anurg$ or $\anurgp \geq \anurg$. 
%

The fixed-point \cref{axiom:eqmap} allows us to rewrite non-terminals to their defining terms. 
The \cref{axiom:lfp} allows us to rewrite non-terminals to a prefixed point,
using Knaster and Tarski's characterization of least fixed points~\cite{BirkhoffLatticeTheory}.
Here, we let $\aterm_{\nonterminals}$ denote
a vector of terms with one entry $\aterm_{\nonterminal}$ 
per non-terminal $\nonterminal\in\nonterminals$, 
and use $\{\nonterminals/\aterm_{\nonterminals}\}$ 
for the substitution of all non-terminals by these terms. 

Recall that \Cref{axiom:spec} only plays a role in the definition of the specialized axiomatic precongruence. 
The axiom depends on the objective $\objective$ of interest, meaning it actually is a family of axioms. 
The axiom relates word terms $\aword, \awordp \in \wordterms$ as prescribed by the syntactic precongruence. 

\begin{example} 
We show that simulation implies inclusion in \Cref{Example:SimulationInclusion}.
Remember that Eve tries to refute the relation.
We prove axiomatically
\[
	b\echoicen{2} c 
	\;\axleqper{\labelcref{axiom:lattice-ord}}\; (b\echoicen{1} c)\echoicen{2}(b\echoicen{1} c) 
	\;=\; \bigEchoiceOf{2}{}{(b\echoicen{1} c)}
	\;\axeqper{\labelcref{axiom:lattice-abs},\;\labelcref{axiom:lattice-assoc}}\;
	\bigAchoiceOf{2}{}\bigEchoiceOf{2}{}{(b\echoicen{1} c)}
	\;\axeqper{\labelcref{axiom:lattice-abs}}\;
	b\echoicen{1} c\ .
\]
By~\labelcref{axiom:lattice-ord}, we have $b\axleq b\echoicen{1} c$ and $c\axleq b\echoicen{1} c$, which we can apply to subterms by congruence. 
The equality holds because choices range over sets. 
We apply \labelcref{axiom:lattice-abs} with $\aterm=\atermp=\mathord{\echoicen{2}}(b\echoicen{1}c)$, apply \labelcref{axiom:lattice-assoc} to flatten the choices, apply \labelcref{axiom:lattice-abs} once more to remove the choices, and finally apply congruence. 
This yields $\termincl \axleq \termsim$.
With soundess of the axiomatization (\Cref{Section:Soundness}), we obtain the desired $\termincl \congleq \termsim$.
With a normal form result (\Cref{Section:Normalization}), we also show completeness (\Cref{Section:Completeness}).
\end{example}

\section{Soundness}\label{Section:Soundness}
\begin{proposition}[Soundness]\label{Proposition:Soundness}
$\aterm\axleq\atermp$ implies $\aterm\congleq \atermp$, $\aterm\axleq_{\objective}\atermp$ implies $\aterm\speccongleq{\objective}\atermp$.
\end{proposition}

This section is devoted to proving~\Cref{Proposition:Soundness}. 
Proving soundness is difficult because it requires us to reason over all contexts.
If $\aterm\specaxleq{\objective}\atermp$ is an axiom, then we need to show that
$\winsof{\acontext{\aterm}}{\objective}$ 
implies
$\winsof{\acontext{\atermp}}{\objective}$
for all $\acontext{\contextvar}\in\contexts$.

We first develop a proof technique for soundness that allows us to reduce the set of contexts we have to consider, and in a second step prove the axioms sound.
To define the contexts that have to be considered, we introduce some terminology. 
We say that term $\aterm$ is \emph{immediate} for context $\acontext{\contextvar}\in\contexts$, 
if $\acontext{\aterm}\in\wordterms$ or $\acontext{\aterm}$ 
is active and $\aterm$ contains the leading subterm 
in $\acontext{\aterm}$, denoted by~$\acontext{\makeleading{\aterm}}$.
If this is not the case, we call $\aterm$ \emph{paused} for~$\acontext{\contextvar}$. 
As the names suggest, immediate terms get rewritten in the next move while paused terms do not.
For example, term $\bigEchoiceOf{2}\atermset$ is immediate 
for $\contextvar\appl\bigAchoiceOf{1}\atermsetp$ 
but paused for $\contextvar\appl\bigAchoiceOf{3}\atermsetp$
and $\bigEchoiceOf{2}\set{\contextvar, \atermp}$. 
In the last context,~$\contextvar$ is enclosed by 
a choice. No term is immediate for such a context.

\begin{lemma}[Proof Technique]\label{Lemma:ContextLemma}
If  $\winsof{\acontext{\aterm}}{\objective}$ implies $\winsof{\acontext{\atermp}}{\objective}$ 
for all contexts $\acontext{\contextvar}\in\contexts$ where at least one of $\aterm$ or $\atermp$ is immediate, then 
	$\aterm\speccongleq{\objective}\atermp$.
\end{lemma}

\Cref{Lemma:ContextLemma} defines precisely the contexts we need to consider when proving the axioms sound.
Its proof relies on the following observation: 
a paused term does not change the outcome of a move in the context (cf.\ \Cref{Appendix:ProofTechnique}).
\begin{lemma}\label{Lemma:ContextLeads}
Consider terms $\aterm$ and $\atermp$ that are paused for context $\acontext{\contextvar}$.
Then, $\ownof{\acontext{\aterm}}=\ownof{\acontext{\atermp}}$ and 
there is a set of contexts $\acontextset\subseteq\contexts$
so that 
	$\successorsof{\acontext{\aterm}}=\set{\acontextd{\aterm}\mid 
	\acontextd{\contextvar}\in D}$ and 
	$\successorsof{\acontext{\atermp}}=
	\set{\acontextd{\atermp}\mid \acontextd{\contextvar}\in D}$.
\end{lemma}

We are now prepared to prove each axiom sound.
Using \Cref{Lemma:ContextLemma}, all proofs share a common approach:
we fix an objective and pick a context that 
is immediate for at least one term in the axiom.
Then, we unroll the game arena for a few moves until it reveals the winning implication we are after. 
We restrict our presentation to the soundness of \Cref{axiom:dist-left} and defer details to \Cref{Appendix:MissingAxioms}.
The proofs make use of two properties of immediate terms.
\begin{lemma}\label{Lemma:TermLeads}
Let term $\aterm$ be immediate for context $\acontext{\contextvar}$. 
Then we have $\ownof{\acontext{\aterm}}=\ownof{\aterm}$.
If $\urgencyof{\aterm}\leq\urgencyof{\atermp}$, then also term $\atermp$ is immediate for
	$\acontext{\contextvar}$.
\end{lemma}


\begin{proof}[Proof of \Cref{Proposition:Soundness} (Soundness)]
\textbf{\Cref{axiom:dist-left}:} We consider binary choices owned by Eve, the generalization to choices over arbitrary sets and also Adam's case are similar. 
For $\urgencyof{\aterm}<\anurg$, the axiom says that  
$\aterm\appl(\atermp\echoicen{\anurg}\atermpp)\axleq\aterm\appl\atermp\echoicen{\anurg}\aterm\appl\atermpp$ and vice versa. 
The goal is thus to show 
$
	\aterm\appl(\atermp\echoicen{\anurg}\atermpp)
	\wineq\aterm\appl\atermp\echoicen{\anurg}\aterm\appl\atermpp
$.

Consider an objective $\objective$ and let $\acontext{\contextvar}$ 
be a context for which at least one of $\aterm\appl(\atermp\echoicen{\anurg}\atermpp)$ or $\aterm\appl\atermp\echoicen{\anurg}\aterm\appl\atermpp$ is immediate.
The urgencies of both terms are $\anurg$.
This means that not only one but actually both terms are immediate for $\acontext{\contextvar}$
and, moreover, the owner of $\acontext{\makeleading{\aterm\appl(\atermp\echoicen{\anurg}\atermpp)}}$ 
and $\acontext{\makeleading{\aterm\appl\atermp\echoicen{\anurg}\aterm\appl\atermpp}}$ is Eve.
The first moves in the game arenas are thus done by the same player and have the same result: 
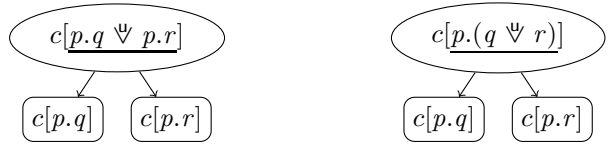
\begin{gametrees}
	\begin{tikzpicture}[scale=0.9]
		\tikzstyle{level 1}=[sibling distance=4.5em, level distance=12mm]
		\node[peve] {$\acontext{\makeleading{\aterm\appl\atermp\echoicen{\anurg}\aterm\appl\atermpp}}$}
		child {node[punk] {$\acontext{\aterm\appl\atermp}$} edge from parent[->]}
		child {node[punk] {$\acontext{\aterm\appl\atermpp}$} edge from parent[->]};
	\end{tikzpicture}
	\hfill
	\begin{tikzpicture}[scale=0.9]
		\tikzstyle{level 1}=[sibling distance=4.5em, level distance=12mm]
		\node[peve] {$\acontext{\makeleading{\aterm\appl(\atermp\echoicen{\anurg}\atermpp)}}$}
		child {node[punk] {$\acontext{\aterm\appl\atermp}$} edge from parent[->]}
		child {node[punk] {$\acontext{\aterm\appl\atermpp}$} edge from parent[->]};
	\end{tikzpicture}
\end{gametrees}
As a consequence, translating winning strategies becomes straightforward
and the equivalence ${\winsof{\acontext{\aterm\appl(\atermp\echoicen{\anurg}\atermpp)}}{\objective}}$ 
iff $\winsof{\acontext{\aterm\appl\atermp\echoicen{\anurg}\aterm\appl\atermpp}}{\objective}$ holds.
\end{proof}
\section{Normalization}\label{Section:Normalization}
As a first step towards completeness, we show that each term can be brought 
into a normal form using our axioms. 
Normalization is a standard approach in completeness proofs.
The treatment of alternation and urgency is new. 

The normal form eliminates non-terminals and orders the interplay between 
concatenation and choice: a normal form term is a tree of height $2\maxurg$ (with $\maxurg$ the maximal urgency) 
that repeatedly alternates between Eve's and Adam's choices while decreasing the urgency.  
The leaves of the tree are terminal words, $\tskip$, or $\terr$.
Inductively, we define $\gnfn{0} = \analph^{*}\cup\set{\terr}$ and for $\anurg>0$,
\begin{align*}
    \agnfn{\anurg}\ = \
    \setcond{\bigAchoiceOf{\anurg}\atermset}{\emptyset\neq\atermset\subseteq\gnfn{\anurg-1}}
    \hspace{1cm}
    \gnfn{\anurg}\ = \
    \setcond{\bigEchoiceOf{\anurg}\atermset}{\emptyset\neq\atermset\subseteq\agnfn{\anurg}}\,.
\end{align*}
The base case terms are all owned by Eve. 
In an $\agnfn{\anurg}$ term, Adam chooses over $\gnfn{\anurg-1}$ terms. 
In an $\gnfn{\anurg}$ term, Eve chooses over such $\agnfn{\anurg}$ terms owned by Adam.
The main result of this section is the following.
\begin{proposition}\label{Proposition:Normalization}
    There is a function $\norm:\terms\to\gnfn{\maxurg}$ so that  $\normof{\aterm}\axeq\aterm$ for all $\aterm$. 
\end{proposition}

\begin{example}
We illustrate the normal form computation on our running example.
Let $\maxurg = 2$ and consider $(\aletter_{l}\achoicen{1} \aletter_{r})\appl(b\echoicen{2} c)$. 
\begin{align*}
(\aletter_{l}\achoicen{1} \aletter_{r})\appl(b\echoicen{2} c)\ 
\axeqper{\hspace{0.9em}\labelcref{axiom:dist-left}\hspace{0.9em}}&\quad (\aletter_{l}\achoicen{1} \aletter_{r})\appl b \echoicen{2} (\aletter_{l}\achoicen{1} \aletter_{r})\appl c\\
\axeqper{\hspace{0.9em}\labelcref{axiom:dist-right}\hspace{0.9em}}&\quad (\aletter_{l}\appl b\achoicen{1} \aletter_{r}\appl b)\echoicen{2} (\aletter_{l}\appl c \achoicen{1} \aletter_{r}\appl c)\\
\axeqper{\labelcref{axiom:lattice-abs},\; \labelcref{axiom:lattice-assoc}}&\quad [\achoicen{2}\echoicen{1}(\aletter_{l}\appl b\achoicen{1} \aletter_{r}\appl b)]\echoicen{2} [\achoicen{2}\echoicen{1}(\aletter_{l}\appl c \achoicen{1} \aletter_{r}\appl c)]\ .
\end{align*}
\end{example}

%
Formally, we prove \Cref{Proposition:Normalization} in two steps.
First, we compute a normal form term $\normof{\nonterminal}\axeq\nonterminal$ for every non-terminal $\nonterminal\in\nonterminals$ in a fixed-point iteration.
Then, we replace each non-terminal in $\aterm$ by its normal form.
Finally, we normalize the remaining term to obtain $\normof{\aterm}\axeq\aterm$. 
The two steps are formalized in the next lemma, details of the normalization can be found in \Cref{Appendix:Normalization}.
\begin{lemma}\label{Lemma:NTNormalization}
For all $\aterm\in\terms$
    without non-terminals there is
    $\normof{\aterm}\in\gnfn{\maxurg}$
    with $\aterm\axeq\normof{\aterm}$.
For all non-terminals $\nonterminal\in\nonterminals$,
    there is a $\normof{\nonterminal}\in\gnfn{\maxurg}$
    with $\normof{\nonterminal}\axeq\nonterminal$.
\end{lemma}

\section{Completeness}\label{Section:Completeness}
\begin{proposition}[Completeness]\label{Proposition:Completeness}
$\aterm\winleq\atermp$ implies $\aterm\axleq\atermp$. For a right-separating objective $\objective\subseteq\analph^{*}$, $\aterm\speccongleq{\objective}\atermp$ implies $\aterm\specaxleq{\objective}\atermp$.
\end{proposition}

Interestingly, the completeness for the axiomatic precongruence follows from the completeness for the specialized version. 
To prove this implication, we consider objectives for which \Cref{axiom:spec} does not add relations. 
Formally, we call an objective $\objective$ \emph{domain shattering}, if $\axleq\ =\ \specaxleq{\objective}$. 
\begin{lemma}\label{Lemma:DomainShattering}
    The language $\aword.\aword^{\mathit{reverse}}$ is domain-shattering and right-separating.
\end{lemma}
%
%
%
%

The argument for completeness is this. 
Let $\objective\subseteq\analph^{*}$ be the domain-shattering and right-separating 
objective from \Cref{Lemma:DomainShattering}.
We have
$$\congleq\ \overset{\text{Definition}}{\subseteq}\ \speccongleq{\objective}
    \ \overset{\text{Prop.~\ref{Proposition:Completeness}}}{\subseteq}\ \specaxleq{\objective}\ \overset{\text{Shattering}}{\vphantom{\subseteq}=}\ \axleq\ 
    \overset{\text{Soundness, Prop.~\ref{Proposition:Soundness}}}{\subseteq}\ \congleq.$$

\subsection{Completeness Proof, Specialized Case}
To show \Cref{Proposition:Completeness}, fix a right-separating $\objective\subseteq\analph^{*}$. 
With the normalization in \Cref{Proposition:Normalization} and soundness of the axiomatic precongruence in \Cref{Proposition:Soundness}, it is sufficient to show completeness for terms in normal form. 
For $\aterm,\atermp\in\gnfn{\maxurg}$, we want to show that 
$\aterm\not\specaxleq{\objective}\atermp$ implies $\aterm\not\speccongleq{\objective}\atermp$ 
by giving a context that tells them apart. 
This, however, is difficult as it requires us to understand precisely when the axiomatic congruence fails.

Our way out is to define a less flexible preorder that is easier to handle.
The \emph{domination preorder} $\discleq$ only relates normal form terms of the same urgency and owned by the same player.
It is defined by induction on the urgency. 
For $\aword, \awordp\in\gnfn{0}$, we have $\aword\discleq\awordp$ if $\aword\sgleq{\objective}\awordp$. 
For $\anurg >0$:
    \[
        \bigAchoiceOf{\anurg}\atermset
        \discleq\bigAchoiceOf{\anurg}\atermsetp
        \;\text{ if }\forall\atermppp\in\atermsetp.
        \exists\atermpp\in\atermset.\atermpp\discleq\atermppp
        \quad
        \bigEchoiceOf{\anurg}\atermset
        \discleq\bigEchoiceOf{\anurg}\atermsetp
        \;\text{ if }\forall\atermpp\in\atermset.
        \exists\atermppp\in\atermsetp.\atermpp\discleq\atermppp
    \]
\Cref{Proposition:Completeness} holds with the following lemma, which we prove in \Cref{Appendix:CompletenessProof}.
\begin{lemma}\label{Lemma:DominationPreorder}
For a right-separating objective $\objective$,
$\aterm\speccongleq{\objective}\atermp$ implies $\aterm\discleq\atermp$ and  $\aterm\discleq\atermp$ implies  $\aterm\specaxleq{\objective}\atermp$. 
\end{lemma}


\newcommand{\decwhichover}[2]{{\normalfont\textsf{$#2$-DEC-}$#1$}}
\newcommand{\decwinsof}[1]{\textsf{$#1$-DEC-}$\winsof{}{}$}
\section{Decidability and Complexity}\label{Section:Verification}
We study the decidability and complexity of checking the contextual preorder and its specialized variant.
To this end, we leave the setting of infinitary syntax and call $\aterm$ \emph{finitary}, if it refers to a finite set of non-terminals $(\nonterminals, \eqmap)$ and all defining terms~$\eqmapof{\nonterminal}$ and $\aterm$ itself are finite. 
We make the decision problem parametric in the relation~$\mathit{R}$ to be checked, and  instantiate~$\mathit{R}$ with $\congeq$, $\speccongeq{\objective}$, and $\speccongleq{\objective}$: 
\begin{quote}
\underline{\decwhichover{\mathit{R}}{\maxurg}}\\
{\bfseries Given:} Finitary $\aterm, \atermp$ over $\analph$,~$(\nonterminals, \eqmap)$ of urgency $\maxurg$.  \\
{\bfseries Problem:} Does $\aterm\ \mathit{R}\ \atermp$ hold?
\end{quote}

The first finding is that already the contextual equivalence is undecidable. 
The proof is by a reduction from the equivalence problem for context-free languages and the result continues to hold if we fix an alphabet with at least two letters. 
A consequence is that the specialized contextual equivalence for domain-shattering objectives is also undecidable.
\begin{proposition}\label{Proposition:Undec}
\decwhichover{\congeq}{\maxurg} and \decwhichover{\speccongeq{\objective}}{\maxurg} with $\objective$ domain-shattering are undecidable for every $\maxurg$. 
\end{proposition}
Recall that the language $\aword\appl\aword^{\mathit{reverse}}$ is domain-shattering, so already context-free objectives lead to undecidability. 

Our main result in this section is that for regular objectives the specialized contextual preorder is decidable. 
We can also give the precise complexity, for which we measure the size of the input in the expected way as 
$\sizeof{\aterm}+\sizeof{\atermp}+\sizeof{\analph}+\sizeof{\eqmap}$. 
The size of the defining equations is $\sizeof{\eqmap}=\sum_{\nonterminal\in\nonterminals}1+\sizeof{\eqmapof{\nonterminal}}$. 
The size of a term is $\sizeof{\tskip}=\sizeof{\terr}=\sizeof{\aletter}=\sizeof{\nonterminal}=1$, 
$\sizeof{\aterm\appl\atermp}=1+\sizeof{\aterm}+\sizeof{\atermp}$, and $\sizeof{\bigovoid^{\anurg}\atermset}=1+\sum_{\aterm\in\atermset}\sizeof{\aterm}$. 
\begin{theorem}\label{Theorem:UpperLower}
Let $\maxurg$ be an urgency. 
For every regular objective~$\objective\neq \emptyset$, the problem \decwhichover{\speccongleq{\objective}}{\maxurg} is \ptime-complete.  
\end{theorem}
It is worth noting that the result does not expect the objective to be right-separating.
We can indeed decide the specialized contextual preorder $\speccongleq{\objective}$ for all regular objectives $\objective$. 
Moreover, the lower bound holds no matter the objective.

It is natural to define a variant \decwhichover{\speccongleq{*}}{\maxurg} of the problem in which also the objective is part of the input and given as a deterministic finite automaton 
$(\analph, \states, \astateinit, \transitions, \finalstates)$. 
In this case, we use $\sizeof{\objective}$ to refer to $\sizeof{\analph}+\sizeof{\states}$. 
The following result shows the dramatic influence that the objective has on the complexity. 
\begin{theorem}\label{Theorem:UpperLowerInput}
\decwhichover{\speccongleq{*}}{\maxurg} is $\kexptime{(2\maxurg - 1)}$-complete.  
\end{theorem}

A consequence is that we can also solve the problem of making an observation, denoted as \decwinsof{\maxurg} and defined with almost the same input as \decwhichover{\speccongleq{*}}{\maxurg}. 
\begin{corollary}\label{Corollary:DecWinner}
\decwinsof{\maxurg} is $\kexptime{(2\maxurg - 1)}$-complete.
\end{corollary}





\section{Upper Bound}\label{Section:UpperBound}
We prove the upper bounds claimed in \Cref{Theorem:UpperLower,Theorem:UpperLowerInput}
as follows. 
\begin{proposition}\label{Proposition:DecideUpperBoundEff}
Given finitary terms $\aterm$ and $\atermp$ and a regular objective $\objective$ as a DFA, deciding $\aterm\speccongleq{\objective}\atermp$ can be done in time 
$(\sizeof{\aterm}+\sizeof{\atermp}+\sizeof{\eqmap}\sizeof{\nonterminals})\cdot 
    \expof{2\maxurg-1}{\bigoof{\sizeof{\objective}^2}}$.
\end{proposition}

The undecidability result in \Cref{Proposition:Undec} shows that the normal form for the axiomatic congruence from \Cref{Section:Normalization} is insufficient as a basis for algorithms. 
The problem is that the normal form terms are typically infinite, and therefore difficult to handle computationally. 
The source of infinity can be found in the base case: already $\gnfn{0}=\analph^*\cup\set{\terr}$ is infinite, and this propagates upwards. 
We realize that the $\objective$-specialized axiomatic congruence admits a more refined normal form that is guaranteed to yield finite terms (and finitely many of~them).  
The key idea is to factorize $\gnfn{0}$ using \Cref{axiom:spec}. 

We define the set of \emph{$\objective$-specialized normal form terms} by induction on the urgency. 
The base case $\specgnfn{0} = \factorize{\analphbot^{*}}{\lreq{\objective}}$ are classes of words in the syntactic congruence $\lreq{\objective}\ =\ \lrleq{\objective}\cap\lrgeq{\objective}$ induced by $\objective$. 
For $\anurg>0$, the definition is
\begin{align*}
    \specagnfn{\anurg}\ = \
    \setcond{\bigAchoiceOf{\anurg}\atermset}{\emptyset\neq\atermset\subseteq\specgnfn{\anurg-1}}\quad
    \specgnfn{\anurg}\ = \
    \setcond{\bigEchoiceOf{\anurg}\atermset}{\emptyset\neq\atermset\subseteq\specagnfn{\anurg}}\ .
\end{align*}
Note that since $\objective$ is regular, the set $\specgnfn{0}$ and so all~$\specagnfn{\anurg}$ and $\specgnfn{\anurg}$ are guaranteed to be finite~\cite{RS59}. 
Another aspect is that we change the alphabet to having $\lreq{\objective}$-congruence classes as letters. 
This can be fixed by working with a representative system: we represent every $\lreq{\objective}$-class by one of its elements.

We adapt the normalization process from \Cref{Section:Normalization} 
to compute a term in the $\objective$-specialized normal form. 
Only the base case changes, for the inductive cases we merely study the complexity. 
Interestingly, the overall normalization takes time $2\maxurg$-fold exponential only in the size of the objective, 
using the common definition $\repexpof{0}{x}=x$ and $\repexpof{\anurg+1}{x}= 2^{\repexpof{\anurg}{x}}$. 
\begin{lemma}\label{Lemma:EffNormalForm}
Given a finitary term $\aterm$ and a regular objective $\objective\subseteq\analph^*$ as a DFA, we can compute $\specnormof{\aterm}\in\specgnfn{\maxurg}$
    with $\specnormof{\aterm}\specaxeq{\objective}\aterm$
    in time $(\sizeof{\aterm}+\sizeof{\eqmap}\sizeof{\nonterminals})\cdot\repexpof{2\maxurg-1}{\bigoof{\sizeof{\specgnfn{0}}}}$. 
We have $\sizeof{\specgnfn{\maxurg}}=\repexpof{2\maxurg}{\bigoof{\sizeof{\specgnfn{0}}}}$. 
\end{lemma}

The result already allows us to decide the $\objective$-specialized contextual preorder as follows. 
Since we cannot assume the objective to be right-separating, the algorithm cannot rely on a full abstraction result. 
Instead, we have to evaluate $\aterm\speccongleq{\objective}\atermp$ directly, by iterating over contexts.  
What makes this possible is the combination of our proof technique for soundness in \Cref{Lemma:ContextLemma} and the $\objective$-specialized normal form just introduced. 
With \Cref{Lemma:ContextLemma}, we do not have to iterate over all contexts to show $\aterm\speccongleq{\objective}\atermp$, but only over contexts of the form $\atermpp\appl\contextvar\appl\atermppp$. 
With  \Cref{Lemma:EffNormalForm}, the terms $\atermpp$ and $\atermppp$ can be normalized. 
\begin{corollary}\label{Corollary:Algorithm}
Let $\objective$ be regular and $\aterm, \atermp$ finitary.  
Then $\aterm\speccongleq{\objective}\atermp$ iff for all $\acontext{\contextvar}=\atermpp\appl\contextvar\appl\atermppp$ with $\atermpp\in \specgnfnof{\maxurg-1}{\objective}$, $\atermppp\in\specgnfn{\maxurg}$ we have 
$\winsof{\specnormof{\acontext{\aterm}}}{\objective}$ implies $\winsof{\specnormof{\acontext{\atermp}}}{\objective}$. 
\end{corollary}

The algorithm formulated in the corollary is slower than the promised upper bound by two exponents because $\specgnfn{\maxurg}$ contains $\repexpof{2\maxurg+1}{\bigoof{\sizeof{\objective}^2}}$ many terms.
To overcome the problem, the first step is to reduce the number of contexts that have to be considered. 
The idea is to factorize the contexts along their solution spaces.
The solution space of a context is the set of terms $\aterm$ for which $\winsof{\acontext{\aterm}}{\objective}$ holds.
When checking for $\aterm \speccongleq{\objective} \atermp$, the job of a context $\acontext{\contextvar}$ is to disprove $\aterm \speccongleq{\objective} \atermp$ by showing $\winsof{\acontext{\aterm}}{\objective}$ and $\notwinsof{\acontext{\atermp}}{\objective}$. 
Hence, when two contexts have the same solution space, it suffices to consider one of them.  

What makes the solution space equivalence algorithmically interesting 
is that (i) it is coarse enough to save an exponent and (ii) we can directly compute with equivalence classes of contexts.  
The key insight behind both statements is that the solution space of a context can be represented in a convenient way: it is the $\specaxleq{\objective}$-upward closure of a so-called \emph{characteristic term}. 
%
%
\begin{definition}
    Term $\aterm$ is \emph{characteristic for~$\acontext{\contextvar}$ wrt. $\objective\subseteq\analph^{*}$}, if for all 
    $\atermp \in \terms$ we have $\winsof{\acontext{\atermp}}{\objective}$ 
    if and only if $\aterm \specaxleq{\objective} \atermp$.
\end{definition}

We will show that there are only $\repexpof{2\maxurg}{\bigoof{\sizeof{\objective}^2}}$ many characteristic terms for contexts of the form $\atermpp\appl\contextvar\appl\atermppp$. 
Moreover, we can compute the characteristic terms directly, without building up the corresponding contexts. 
This saves an exponent in the complexity: we modify the algorithm in Corollary~\ref{Corollary:Algorithm} to iterate through characteristic terms rather than contexts.

We save another exponent in the runtime of our algorithm by a more compact representation of the terms in $\objective$-specialized normal form.  
One exponent in the size of $\specgnfn{\maxurg}$ is inherited from the base case, where already $\sizeof{\specgnfn{0}}=\repexpof{1}{\bigoof{\sizeof{\objective}^2}}$.  
We use the fact that each class of the syntactic congruence can be represented by a function $\states \to \states$ between the states of the objective DFA~\cite{RS59}. 
The key idea is to see these functions as sets of state-pairs.
We simulate a function, say for letter $\aletter$, by letting Eve choose a state change $(\astate, \astatep)$ with $\transitions(\astate,\aletter) = \astatep$.
We thus represent the congruence class $\classof{\aletter} \in \analphbot^*/\lreq{\objective}$ by an angelic choice of urgency $1$ over letters from the alphabet $\states\times\states$.
%
%
%
The modification propagates to the normal form terms, and requires a simple modification of the objective.
The new objective~$\arrof{\objective}$ has a syntactic congruence with an exponentially smaller index, 
namely $\sizeof{\specgnfnof{0}{\arrof{\objective}}}= \bigoof{\sizeof{\objective}^2}$. 

Details on the compact term representation and on how to utilize characteristic terms are given in \Cref{Appendix:Upperbound}.

\renewcommand{\tracesof}[1]{\alang(#1)}
\newcommand{\hp}{\mathsf{HP}}
\newcommand{\langhp}{\alang_{\mathsf{HP}}}
\section{Lower Bound and Applications}\label{Section:Applications}
We show how to reduce well-known program verification tasks to the problem of making an observation. 
The first reduction justifies our complexity lower bound.
\paragraph{Multi-pushdown systems (MPDS).}
MPDS are a model for concurrent programs with recursive threads that synchronize through a shared, finite memory~\cite{QadeerR05,MW20}. 
MPDS are Turing complete~\cite{R00}. 
For bug hunting, it is therefore common to under-approximate their behavior by limiting the number of context switches~\cite{QadeerR05}. 
A context is a sequence of transitions that operate only on one stack, and a context switch is the moment when this stack changes. 
We consider multi-pushdown games (MPDG), the alternating variant of MPDS.\@
A \emph{$\ctxtbound$-context-bounded 2-stack pushdown game~($\ctxtbound$-2PDG)} is an MPDG with two stacks. 
%
%
%
%
\begin{theorem}{\cite{Seth09,MW20}}
    $\ctxtbound$-2PDG are $\kexptime{(\ctxtbound - 2)}$-complete.
\end{theorem}
%
%
%
Our lower bound is due to the following result that we prove in \Cref{Appendix:LowerBound}.  
%

\begin{proposition}\label{Proposition:LowerBound}
    Given a $(2\maxurg+1)$-2PDG $\pdg$, 
    we can compute in polynomial time an alphabet $\analph$, 
    a term $\aterm$ and non-terminals $(\nonterminals, \eqmap)$ over $\analph$ all of maximal urgency~$\maxurg$, and an objective DFA $\objective$ 
    so that Eve wins $\pdg$ if and only if $\winsof{\aterm}{\objective}$.
\end{proposition}
\paragraph{Hyper model checking (HMC).}
Hyperproperties~\cite{CS10} are program properties that refer to sets of computations, as opposed to the classical linear-time properties~\cite{Pnueli77} that refer to single computations. 
The corresponding logics~\cite{CFKMRS14} express hyperproperties by quantifying over the computations of the program $\akripke$, as in  
\begin{align*}
\hp\quad =\quad    \forall \aword_1\in \tracesof{\akripke}.\ \exists \aword_2\in\tracesof{\akripke}\ \ldots\ \mathsf{Q} \aword_n\in \tracesof{\akripke}.\ \aword_1\oplus\ldots \oplus \aword_n\in \langhp \ .
\end{align*}
%
%
%
We use $\oplus$ to denote the convolution of words. 
The convolution expects words of the same length (that is then chosen by the first quantifier) and forms their letter-by-letter product, meaning $\langhp$ is a language over an alphabet ${\analph^n}$. 
%
%
The hyper model checking problem (HMC) asks whether a given program satisfies a given hyperproperty, denoted by $\akripke \vDash \hp$.
The standard way of solving HMC is to compose the program with $\langhp$, and then repeatedly with itself~\cite{BDR04,FRS15}. 
We show how to encode HMC into the problem of making an observation for urgency terms.
To be precise, HMC is usually defined over infinite words, but we consider a finite-word version.
The details are in \Cref{Appendix:Hyper}. 
\begin{theorem}
    Given $\akripke$ and $\hp$, we can compute in polynomial time a term $\aterm_{\akripke, \hp}$ and an objective $\objective$ so that $\akripke \models \hp$ if and only if $\winsof{\aterm_{\akripke, \hp}}{\objective}$.
\end{theorem}
Interestingly,  our reduction to urgency terms even yields decidability results for HMC over recursive programs.
As HMC generalizes inclusion, which is undecidable for context-free languages, the recursive case needs a restriction. 
We can work with a restriction similarly liberal as in the recent positive results~\cite{GMO22arxiv,BZCV23}, which form the decidability frontier for HMC.\@
%
%
%
\newcommand{\impset}{H}
\newcommand{\imp}{\mathsf{hd}}
\newcommand{\aimp}{h}
\newcommand{\impof}[1]{\imp(#1)}
\newcommand{\agame}{\mathit{G}}
\newcommand{\runsof}[1]{\mathsf{runs} (#1)}

\paragraph{Simulation, Inclusion, and Games with imperfect information.}
As explained in the introduction, urgency terms can express inclusion and simulation problems.
We now give the details of the encoding, and also add imperfect information games to the picture. 
Let $\adfa = (\states_{\adfa}, \transitions_{\adfa}, \astateinit_{\adfa}, \finalstates_{\adfa})$
and $\adfap = (\states_{\adfap}, \transitions_{\adfap}, \astateinit_{\adfap}, \finalstates_{\adfap})$ be non-deterministic finite automata. 
Let $\astate, \astatep$ range over $\states_{\adfa}$ 
and $\astatepp, \astateppp, \astatepppp$ over $\states_{\adfap}$ with~$\astateppp \neq \astatepppp$. 
Imperfect information games $(\adfa, \imp)$ (for a definition, refer to~\cite{Reif84}) 
are played on a finite automaton $\adfa$ 
and hide the concrete states $\astate \in \states_{\adfa}$ from Adam 
and only present him an abstraction $\aimp$, 
using a function $\imp : \states \to \impset$. 
Such a game is played in rounds in which Adam sees the abstraction $\aimp \in \impset$ of the current state~$\astate$, 
chooses an action $\aletter \in \analph$, 
and then Eve selects a transition from $\astate$ for that action.
Eve wins $\agame$ when she can force the play from the initial state $\astateinit_{\adfa}$ into a final state in $\finalstates_{\adfa}$.

We give the defining equations for the non-terminals: 
\begin{align*}
    \nonterminal_{\astate}^{\subseteq} &= 
        \Echoicen{2}
        \setcond{(\Achoicen{1}_{(\astatepp, \aletter, \astateppp)\in \transitions_{\adfap}} 
        (\astatepp, \aletter, \astateppp)) \appl \nonterminal_{\astatep}^{\subseteq}}
        {(\astate, \aletter, \astatep) \in \transitions_{\adfa}} 
        \cup \setcond{\tskip}{\astate \in \finalstates_{\adfa}} \\
    \nonterminal_{\astate}^{\precsim} &= 
        \Echoicen{1}
        \setcond{(\Achoicen{1}_{(\astatepp, \aletter, \astateppp)\in \transitions_{\adfap}} 
        (\astatepp, \aletter, \astateppp)) \appl \nonterminal_{\astatep}^{\precsim}}
        {(\astate, \aletter, \astatep) \in \transitions_{\adfa}} 
        \cup \setcond{\tskip}{\astate \in \finalstates_{\adfa}} \\
    \nonterminal_{\aimp} &=
        \Achoicen{1}_{\aletter\in\analph}
        \Echoicen{1} \setcond{\nonterminal_{\aimp'}
            \appl\Echoicen{1}_{{(\astate,\aletter,\astatep) \in \transitions_{\adfa}, \impof{\astatep} = \aimp'}}
                (\astate,\aletter,\astatep)}
            {\aimp'\in\impset} 
        \cup \setcond{\tskip}{\astate \in \finalstates_{\adfa}}
\end{align*}
The only difference between the encodings of inclusion and simulation is the urgency of the angelic choice. 
In the case of inclusion, Adam only has to make a demonic choice when Eve has chosen the entire word. 
For simulation, the players take turns. 
For imperfect information games, Adam chooses the letter and Eve chooses the abstraction of the next state that will be made visible to Adam and the actual transition that is consistent with  the letter and the abstraction. 
The point is that Adam should see the abstraction but not see the transition.
We achieve this with a left-linear grammar where Eve's choice for the next transition will only be revealed when Adam has made all choices. 
By using urgency two, we could have avoided this trick and used a right-linear model.
In all cases, the task of the objective $\objective$ is to check that we have successive transitions. 

\begin{theorem}
    $\winsof{\nonterminal_{\astateinit_{\adfa}}^{\subseteq}}{\objective}$ if and only if $\langof{\adfa} \not \subseteq \langof{\adfap}$,\,
    $\winsof{\nonterminal_{\astateinit_{\adfa}}^{\precsim}}{\objective}$ if and only if $\adfa \not \precsim \adfap$,\,
    $\winsof{\nonterminal_{{\impof{\astateinit_{\adfa}}}}}{\objective}$ if and only if Eve wins $(\adfa, \imp)$.
\end{theorem}

It would strengthen the usefulness of urgency terms if the above reductions gave optimal complexity upper bounds. 
With what we have, this is true only for imperfect information games, which are \exptime-complete \cite{Reif84}. 
The inclusion problem is \pspace-complete, but \Cref{Corollary:DecWinner} only yields \kexptime{3} upper bound.
Simulation is \ptime-complete, but \Cref{Corollary:DecWinner} yields an \exptime\ upper bound. 
The problem with simulation can be fixed with a different encoding that is also natural but does not show the similarity to language inclusion.
There, we use the fact that the objective is fixed and a consequence of \Cref{Theorem:UpperLower} that is similar to \Cref{Corollary:DecWinner}. 
In the following, we close the gap for language inclusion. 

We identify natural subclasses of urgency terms and study their verification problems.
In the term constructed for inclusion, for every urgency there is a single player that makes a choice.
This is not true for simulation nor for imperfect information games. 
We call this fragment the \emph{weak urgency terms}.
%
%
Linear grammars as used in $\nonterminal_{\astateinit_{\adfa}}^{\subseteq}$, 
$\nonterminal_{\astateinit_{\adfa}}^{\precsim}$, 
and $\aterm_{\akripke, \hp}$ further reduce the complexity. 
The definition of weak and linear terms can be found in \Cref{Section:WeakLinearGrammars}.  
\begin{theorem}
    \decwhichover{\winsof{}{}}{\maxurg} for weak, linear terms is in \kexpspace{$(\maxurg - 2)$}.
\end{theorem}
This yields the desired tight upper bound not only for inclusion checking, but also for the hyper model checking of non-recursive programs \cite{rabe2016}. 


\paragraph{Procedure summaries.} 
Procedure summaries~\cite{SP78,Walukiewicz2001,HMM2016} are a standard technique for the analysis of recursive programs.
We now show how to derive procedure summaries from our axiomatization. 
The benefit is that the axiomatization guides the development of summaries, and thus reduces the moment of ingenuity required to come up with them.  
%
%

Summaries are a closed-form representation of the call-return behavior of a recursive procedure.
For functions, such a closed-form representation would capture the input-output behavior. 
In the context of linear-time model checking, the call-return behavior is the state change that the procedure may induce on the observer automaton. 
%
%
Procedure summaries are typically computed as a fixed point over a suitable domain. 
%
We now rediscover these domains as the normal forms for urgency terms.
Our development covers the game version as well. 

\newcommand{\transitionsint}{\transitions_{\mathit{int}}}
\newcommand{\transitionspush}{\transitions_{\mathit{push}}}
\newcommand{\transitionspop}{\transitions_{\mathit{pop}}}
\newcommand{\pds}{\mathcal{P}}

Pushdown games~\cite{Walukiewicz2001,HMM2016} extend recursive programs by alternation.
We can assume reachability as the winning condition because parity can be reduced to reachability in polynomial time~\cite{HMMZ18}. 
Let $\pds = (\states_{\mathit{E}}, \states_{\mathit{A}}, \analph, \stackalph, \transitions, \finalstates)$ be a pushdown game
with transitions $\transitions = (\transitionsint, \transitionspush, \transitionspop)$. 
For convenience, we assume that final states in $\finalstates$ do not have outgoing transitions, but only loops that empty the stack. 
We solve two verification tasks:
\begin{enumerate*}[label=(\alph*)]
    \item whether Eve can force the play into a final state
    \item while also producing a trace of observations from $\analph$ accepted by a DFA $\adfa$. 
\end{enumerate*}
We model both with the same urgency grammar. 
The alphabet are again state changes $(\astate, \aletter, \astatep)$, this time of the pushdown, and the objective $\objective$ checks whether they form a path. 
We define non-terminals $\nonterminal_{\stsymbol}^{\astate}$ that express the fact that the current procedure is~$\stsymbol$ and the current state is~$\astate$. 
In the definition, $\Pchoice$ is $\Echoice$ if Eve owns the state, and $\Achoice$ if Adam owns the state, $\aletter \in \analph$ is a letter, and $\stsymbol, \stsymbolp \in \stackalph$ are stack symbols:
\begin{align*}
    \nonterminal_{\stsymbol}^{\astate}\ =\ \Pchoicen{1} {} &\setcond{(\astate, \aletter, \astatep)\appl \nonterminal_{\stsymbol}^{\astatep}}{(\astate, \aletter, \astatep) \in \transitionsint}
        & \nonterminal_{\stsymbol}\ &=\ \Echoicen{1} \setcond{\nonterminal_{\stsymbol}^{\astate}}{\astate \in \states}
        \\  \cup\ & \setcond{(\astate, \aletter, \astatep)\appl\nonterminal_{\stsymbolp}^{\astatep}\appl\nonterminal_{\stsymbol}}{{(\astate, \aletter, \astatep, \stsymbolp) \in \transitionspush}}
        \\  \cup\ & \setcond{(\astate, \aletter, \astatep)}{{(\astate, \aletter, \astatep, \stsymbolp) \in \transitionspop}}\ .
\end{align*}

We observe that the normal form $\normof{\nonterminal_{\stsymbol}^{\astate}} = \Echoicen{1}_{i \in I} \Achoicen{1} R_{i}$ 
is an Eve choice over a set~$I$ followed by an Adam choice over sets $R_i \subseteq \states \times \factorize{\analphbot^{*}}{\lreq{\langof{\adfa}}} \times \states$.
By definition, every element of $R_i$ has $\astate$ as the first component.
Each $i \in I$ corresponds to a strategy $\strat$ that Eve can follow from state $\astate$ with the top of stack being $\stsymbol$ until this symbol is popped. 
An element $(\astate, \classof{\aword}, \astatep) \in R_i$ thus corresponds to a play through the pushdown according to $\strat$ that ends in state $\astatep$ after $\stsymbol$ has been popped, and the play produces an observation $\awordp \in \classof{\aword}$.
%
%
If $\langof{\adfa} = \analph^*$, the middle part of our summaries becomes irrelevant.
This is (up to parity tracking) the domain that Walukiewicz uses to solve pushdown games in \exptime~\cite{Walukiewicz2001}, which matches our complexity.
For other DFA, we obtain the summaries for pushdown games with inclusion objectives~\cite{HMM2016}. 
To be precise, we obtain an \exptime\ rather than a \kexptime{2}\ upper bound because we assume the objective to be deterministic.

\section{Related Work}\label{Section:RelatedWork}

Urgency annotations are related to priorities in process algebra \cite{CLN01}. 
The key difference is that priorities preserve the program order while urgencies do not. 
This out-of-order execution brings a new form of unbounded memory to the semantics (the unresolved choices) that our axiomatization explains how to handle.  
Also alternation~\cite{CS76} is not common in process algebra. 
Indeed, we have not found a study of angelic and demonic choice from the perspective of contextual equivalence. 
It is the special case of our work when the urgency is $\maxurg=1$. 

Our axiomatization is related to Salomaa's work on language equivalence~\cite{SalomaaAxioms} and Milner's on bisimilarity~\cite{Milner84}. 
The former reference is particularly close: our Axioms (D1) and (D2) generalize Salomaa's Axioms (A4) and (A5) to higher urgencies. 
Also our context lemma for proving soundness has relatives in process algebra~\cite{PS98}.   
Characteristic objects (here, contexts and terms) have appeared as early as~\cite{MP71}. 
Our contribution is to adapt the general idea to our setting. 



We showed that urgency programs can capture hyper model checking over regular languages.  
%
%
Hyper model checking remains decidable if one language is context-free and the others are regular~\cite{PT18}. 
We can model this fragment with urgency programs as well. 
%
%
%
Urgency programs combine recursion and imperfect information. 
This is interesting as the canonical  pushdown games with imperfect information~\cite{Aminof2013} are known to be undecidable even under strong restrictions~\cite{Bozelli2011}. 
%

The goal of effective~denotational~semantics~\cite{Aehlig07,Terui12} is to  solve verification problems with tailor-made denotational semantics. 
Given a specification~$\objective$, the task is to derive a denotational semantics $\semdom_{\objective}$ such that $\semdom_{\objective}(\aterm)$ answers the  question of whether program $\aterm$ satisfies specification~$\objective$. 
Satisfying a specification coincides with our notion of making an observation, $\winsof{\aterm}{\objective}$. The link is close. 
In Appendix~\ref{Appendix:Denotational}, we show how our axiomatization of the specialized contextual preorder induces a denotational semantics that is effective for finitary programs. 
%
%
This shows that the semantic domain can be obtained in a systematic way.    
The landmark result in this field 
is an effective denotational semantics that captures the higher-order model-checking problem~\cite{SalvatiW15b}. 
It would be interesting to consider it from the axiomatic point of view.
%
%
%

\section{Conclusion and Future Work}\label{Section:Conclusion}
We presented urgency annotations for alternating choices as a new programming construct and studied the standard notions of contextual equivalence for the resulting programs.  
We gave sound and complete axiomatizations and settled the complexity.  
Our findings can be used to obtain new algorithms for verification and synthesis tasks, as we demonstrated on examples.  

The next step is to extend urgency programs to infinite words. 
The challenge is to find the right semantics.
If we define the semantics via infinite unrollings, then it is unclear how to ever switch to choices with lower urgency. 
Instead, it seems appropriate to work with programs that contain an $\omega$-operator and can be rewritten a finite number of times. 
This, however, calls for a different set of algebraic techniques~\cite{Wilke94}. 

\newpage

\bibliographystyle{IEEEtran}

\bibliography{cited}

\cleardoublepage{}

\appendix
\section{Handling Infinitary Syntax}\label{Appendix:Infinitary}

For the development of the section, 
we call a term an $\maxurg$-term 
when no choice subterm has urgency higher than $\maxurg$
and the expansion of each non-terminal is an $\maxurg$-term.
Infinitary syntax requires care 
to make sure that set theoretic concepts remain sound.
The problem lies on the unbounded nature of the choice operator.
If we were to allow choices to range over \emph{arbitrary sets} of terms, 
the class of all terms would no longer form a set.
To see this, suppose that the class of terms with unbounded choice were a set $\terms$.
Then, we could build the term ${\bigEchoiceOf{\anurg}{\terms}\in\terms}$, 
which contradicts the Axiom of Regularity.
But for our development in the paper, 
it is a convenience to have $\terms$ as a set 
and not be distracted by subtle differences of classes and sets. 
For this reason, we impose a restriction on the terms, that does not hinder our developments.
It is defined over the structural depth of a term, which is
\begin{align*}
    \dpthof{\nonterminal}&=\dpthof{\aletter}=\dpthof{\terr}
    =\dpthof{\tskip}=0\\[0.5em]
    \dpthof{\aterm\appl\atermp}&=\max
    \set{\dpthof{\aterm}+1, \dpthof{\atermp}+1}\\
    \dpth(\bigPchoiceOf{\anurg}\atermset)&=
    \sup\set{\dpthof{\aterm'}+1\mid \aterm'\in\atermset}\,.
\end{align*}
The notion of structural depth lifts naturally to defining assignments~$(\nonterminals, \eqmap)$.
We let $\dpthof{\eqmap}$ be the smallest limit ordinal 
strictly greater than $\dpthof{\eqmapof{\nonterminal}}$ for all 
$\nonterminal\in\nonterminals$.
So, whenever the paper mentions the set of all terms $\terms$, we refer, in fact, to restriction 
$\terms_{\eqmap}=\setcond{\aterm}{\dpthof{\aterm}<\dpthof{\eqmap}}$.
The set of permitted contexts $\contexts_{\eqmap}$ is defined in the same way,
as contexts are defined as terms built from 
$(\nonterminals\discup\set{\contextvar}, \eqmap)$.
We let $\contexts_{\eqmap}=\set{\acontext{\contextvar}\mid 
\dpthof{\acontext{\contextvar}}<\dpthof{\eqmap}}$, 
and drop the subscript from this set as well.  
This restriction indeed does not hinder our development.
All definitions we apply to terms can be expressed 
as context free replacements and 
the set $\terms_{\eqmap}$ is closed under such replacements.
The only exceptions build \Cref{axiom:lattice-mono,axiom:lfp}, for which we present side conditions so to stay in the restricted set $\terms_{\eqmap}$.

\begin{lemma}\label{Lemma:SensibleOrdinal}
    Let $(\nonterminals, \eqmap)$ be a defining assignment to $\maxurg$-terms.
    For all $\aterm\in\terms_{\eqmap}$ 
    and all 
    $\acontext{\contextvar}\in\contexts_{\eqmap}$,
    $\acontext{\aterm}\in\terms_{\eqmap}$.
\end{lemma}

\begin{proof}
    Proof is by an induction on the structure of 
    $\acontext{\contextvar}$.
    Let $\aterm\in\terms_{\eqmap}$ and let 
    $\anordinal=\dpthof{\eqmap}$.
    Note that $\dpthof{\aterm}<\anordinal$.
    The case $\acontext{\contextvar}=\aterm\in\terms_{\eqmap}$ is
    clear, since $\acontext{\contextvar}=\acontext{\aterm}$.
    For the case $\acontext{\contextvar}=\contextvar$,
    $\dpthof{\acontext{\aterm}}=\dpthof{\aterm}<\anordinal$.
    
    For the concatenative inductive case, let
    $\acontext{\contextvar}=\atermp\appl\acontextd{\contextvar}$.
    The case $\acontext{\contextvar}=
    \acontextd{\contextvar}\appl\atermp$ is analogous.
    Since $\dpthof{\acontext{\contextvar}}<\anordinal$, 
    we also have $\dpthof{\acontextd{\contextvar}}<\anordinal$
    and $\dpthof{\atermp}<\anordinal$.
    Applying the induction hypothesis yields
    $\dpthof{\acontextd{\aterm}}<\anordinal$.
    And since $\anordinal$ is a limit ordinal, also 
    $\dpthof{\atermp}+1<\anordinal$ and 
    ${\dpthof{\acontext{\aterm}}+1<\anordinal}$.
    %
    
    For the choice inductive case, let 
    $\acontext{\contextvar}=\bigPchoiceOf{\anurg}
    \set{\acontextd{\contextvar}}\cup\atermsetp$.
    We have $\dpthof{\acontext{\aterm}}=
    \sup(\set{\dpthof{\acontextd{\aterm}}+1}\cup
    \set{\dpthof{\atermp}\mid \atermp\in\atermsetp})$.
    This is equal to the maximum of 
    $\dpthof{\acontextd{\aterm}}+1$ and 
    $\sup\set{\dpthof{\atermpp}+1\mid \atermpp\in\atermsetp}$.
    We know that $\dpthof{\acontext{\contextvar}}<\anordinal$.
    So, $\sup\set{\dpthof{\atermpp}+1\mid \atermpp\in\atermsetp}
    \leq\dpthof{\acontext{\contextvar}}<\anordinal$.
    Per definition, we also have 
    $\dpthof{\acontextd{\contextvar}}<\anordinal$.
    We apply the induction hypothesis to get 
    $\dpthof{\acontextd{\aterm}}<\anordinal$ 
    and due to $\anordinal$ being a limit ordinal, $\dpthof{\acontext{\aterm}}<\anordinal$.
\end{proof}

An important implication of \Cref{Lemma:SensibleOrdinal} 
is that the successor relation is well defined for $\terms_{\eqmap}$.
All rewriting has the form
$\acontext{\bigPchoiceOf{\anurg}\atermset}\gamemove\acontext{\aterm}$ 
with $\aterm\in\atermset$ 
or $\acontext{\nonterminal}\gamemove\acontext{\eqmapof{\nonterminal}}$
for some context $\acontext{\contextvar}$ 
with $\dpthof{\acontext{\contextvar}}<\dpthof{\eqmap}$.
Then, the successors have 
$\dpthof{\acontext{\aterm}} < \dpthof{\acontext{\bigPchoiceOf{\anurg}\atermset}} 
    < \dpthof{\eqmap}$ 
for all $\aterm\in\atermset$,
and $\dpthof{\eqmapof{\nonterminal}}<\dpthof{\eqmap}$
for all $\nonterminal\in\nonterminals$.

\subsection{Axiom Side Conditions on Depth}
The restriction to $\terms_{\eqmap}$ 
also restricts the axiom system to terms only in $\terms_{\eqmap}$
This affects \Cref{axiom:lattice-mono,axiom:eqmap}.

For \Cref{axiom:lattice-mono}, 
note that $\bigPchoiceOf{\anurg}\terms_{\eqmap}\not\in\terms_{\eqmap}$.
\[
    {\axdefleft{\labelcref{axiom:lattice-mono}}{\forall i\in I.\;\aterm_{i}\axleq\atermp_{i}}
		{\bigPchoiceOf{\anurg}{}{\set{\aterm_{i}\mid i\in I}}\axleq 
        \bigPchoiceOf{\anurg}{}{\setcond{\atermp_{i}}{i\in I}}}}
\]
So to keep terms in $\terms_{\eqmap}$, 
we impose the following side condition to \Cref{axiom:lattice-mono}:
$\sup\set{\dpthof{\aterm_i}+1\mid i \in I}<\dpthof{\eqmap}$ and 
$\sup\set{\dpthof{\atermp_i}+1\mid i \in I}<\dpthof{\eqmap}$.
This results in 
$\bigPchoiceOf{\anurg}\set{\aterm_i\mid i \in I}\in\terms_{\eqmap}$
and $\bigPchoiceOf{\anurg}\set{\atermp_i\mid i\in I}\in\terms_{\eqmap}$.

For \Cref{axiom:eqmap}, note that $\terms_{\eqmap}$ is closed under finite substitutions, 
but not under infinite substitutions.
However, $\nonterminals$ is infinite, 
so for arbitrary $\aterm_{\nonterminals}$, 
it is not guaranteed that 
$\eqmapof{\nonterminal}\replace{\nonterminals}{\atermn{\nonterminals}}\in\terms_{\eqmap}$.
\[
    \axdefleft{\labelcref{axiom:lfp}}
        {\forall \nonterminal \in \nonterminals .\; 
            \eqmapof{\nonterminal}
            \replace{\nonterminals}{\atermn{\nonterminals}}
            \axleq \atermn{\nonterminal}}
        {\nonterminalp\axleq\atermn{\nonterminalp}}
\]
We need to ensure that the substituted term belongs to $\terms_{\eqmap}$. 
Thus, we require $\dpthof{\eqmapof{\nonterminal}\replace{\nonterminals}{\atermn{\nonterminals}}}
    <\dpthof{\eqmap}$
for all $\nonterminal\in\nonterminals$. 
In fact, the requirement is already implicitly stated by the axiom.
The precondition requires 
$\eqmapof{\nonterminal}\replace{\nonterminals}{\atermn{\nonterminals}} 
    \axleq \atermn{\nonterminalp}$
and ${\axleq} \subseteq \terms_{\eqmap}\times\terms_{\eqmap}$.
So, $\dpthof{\eqmapof{\nonterminal}\replace{\nonterminals}{\atermn{\nonterminals}}}
    <\dpthof{\eqmap}$
is already required by the axiom implicitly.

\section{Strategy Tree Bounds}\label{Appendix:TreeBounds}
Let $G=(V, v, \own, E)$ be a game arena with reachability objective $O \subset V$
and $\sigma:V\to V$ be a winning strategy for Eve 
from starting position $v \in V$.
Consider the subgraph $T$ of $(V, E)$ reachable from $v$,
where for all $v \in V$ $\own$ed by Eve
only the successor $\sigma(v)$ is part of $T$.
Since $\sigma$ is a winning strategy, 
all paths from $v$ in $T$ must reach the objective in finitely many steps,
i.e.\ $T$ is a tree and every branch is finite.
We prove in a second, 
that this is sufficient to obtain an ordinal $\anordinal$ 
to bound the depth of $T$.
If Eve has a strategy tree with depth $\anordinal$,
we say that Eve wins in $\anordinal$ turns.

The existence of uniform positional strategies
results in an important property.
If Eve wins from $v\in V$ in $\anordinal$ turns,
then there must be a $w\in E(v)$ 
from which Eve wins in $\anordinalpp<\anordinalp$ turns, 
if $\ownof{v}=Eve$.
If $\ownof{v}=Adam$, Eve wins from all $w\in E(v)$ 
in $\anordinal_{w}<\anordinal$ turns.

It remains to prove that the absence of an infinite branch is sufficient to obtain $\anordinal$.
Let $T=(V, E)$ be a directed tree and $\anordinal_{V}$ be the smallest ordinal with $\sizeof{V} \leq \sizeof{\anordinal_{V}}$.
\begin{definition}
    A function $\dpth : V \to \anordinal_{V}$ is called depth assignment 
    when $\dpthof{v}<\dpthof{w}$ for all $(v,w) \in E$.
\end{definition}
Note that $E(v)=\emptyset$ implies $\dpthof{v}=\sup(\emptyset)=0$.

\begin{lemma}\label{Lemma:Depthable}
    If a directed tree $T=(V, E)$ has no infinite path then
    $T$ has a depth assignment $\dpth : V \to \anordinal_{V}$.
\end{lemma}

\begin{proof}
    We prove the contraposition.
    Let $T=(V, E)$ be a directed tree with root $r\in V$.
    We use $T_{v}=(V_{v}, E_{v})$ for the subtree rooted in $v\in V$.
    The key insight is that for all $v\in V$, 
    where $T_{v}$ has no depth assignment,
    there must be $w\in E_{v}(v)$ 
    where also $T_{w}$ has no depth assignment.
    To see the validity of this statement, 
    suppose the existence of $v\in V$ where $T_{v}$ has no depth assignments,
    while all $w\in E(v)$, $T_{w}$ have depth assignments $\dpth_{w}$.
    Note that all $V_{w}$ are disjoint. 
    Then $\dpth: V_{v}\to \anordinal_{V_{v}}$ is a depth assignment, 
    where $\dpth(u)=\dpth_{w}(u)$ if $u\in V_{w}$,
    and $\dpth(v)=\sup\set{\dpth(w)+1\mid w\in E_{v}(v)}$.
    The depth property is satisfied 
    for all $u \in V_{w}$ and also for $v\in V_{v}$.
    So $\dpth:V_{v}\to\anordinal_{|V_{v}|}$ is a depth assignment,
    which contradicts the assumption.

    Let $T=(V, E)$ have no depth assignment.
    We inductively construct an infinite sequence of nodes $(v_0, v_1, \ldots)$ 
    with $v_i,v_{i+1}\in E$ 
    and so that $T_{v_i}$ has no depth assignment.
    The root is $v_0=r$ 
    and to extend $(v_0, \ldots, v_n)$,
    where $T_{v_n}$ has no depth assignment,
    we choose any $v_{n+1}$ such that $T_{v_{n+1}}$ has no depth assignment either,
    which we have shown to exist.
\end{proof}

\section{Soundness: Proof Technique}\label{Appendix:ProofTechnique}

\begin{proof}[Proof of \Cref{Lemma:ContextLemma}]
    Assume Eve wins $\objective$ from $\acontext{\aterm}$. 
    Then she does so in at most $\anordinalp$-many moves, where $\anordinalp$ is an ordinal that is guaranteed to exist by results in \Cref{Appendix:TreeBounds}. 
    To be clear, even for transfinite $\anordinalp$, Eve wins each play after a finite number of moves.
    The ordinal $\anordinalp$ limits the size of the game arena reachable from $\acontext{\aterm}$ when she plays according to her strategy. 
    We show $\winsof{\acontext{\atermp}}{\objective}$
    by transfinite induction on $\anordinalp$.
    The base case is simple, yet instructive. 
    If $\anordinalp=0$ then $\acontext{\aterm}\in\objective$, meaning $\acontext{\aterm}\in\wordterms$.
    Then $\aterm$ is immediate for context $\acontext{\contextvar}$, and  
    %
    the premise yields 
    $\winsof{\acontext{\atermp}}{\objective}$.
    
    In the inductive case, it will make no difference whether~$\anordinalp$
    is a limit ordinal or a successor ordinal, 
    so we will not distinguish the two.
    We have $\winsof{\acontext{\aterm}}{\objective}$. 
    If  $\aterm$ or $\atermp$ is immediate for $\acontext{\contextvar}$,
    then the premise of the lemma already tells us~$\winsof{\acontext{\atermp}}{\objective}$.
    Therefore, assume both terms are paused for $\acontext{\contextvar}$.
    Intuitively, 
    we will see that Eve can copy her strategy 
    from $\acontext{\aterm}$ 
    to~$\acontext{\atermp}$ (until the inserted term becomes immediate). 
    %
    Since both terms are paused for $\acontext{\contextvar}$,
    we can apply \Cref{Lemma:ContextLeads}.
    It shows that, after insertion, the owner is the same,  
    $\ownof{\acontext{\aterm}}=\ownof{\acontext{\atermp}}$, 
    and also gives a set of contexts $\acontextset\subseteq\contexts$ capturing the successors.
    Let $\ownof{\acontext{\aterm}}=\ownof{\acontext{\atermp}}=\eve$.
    Then, there must be a context 
    $\acontextd{\contextvar}\in\acontextset$ so that Eve wins 
    $\acontextd{\aterm}$ in $\anordinalp'<\anordinalp$ moves.
    By the induction hypothesis, $\winsof{\acontextd{\atermp}}{\objective}$.
    Moreover, Eve can play 
    $\acontext{\atermp}\gamemove\acontextd{\atermp}$ and win.
    If $\ownof{\acontext{\aterm}}=
    \ownof{\acontext{\atermp}}=\adam$, 
    then for all
    $\acontextd{\contextvar}\in\acontextset$, 
    Eve must win $\acontextd{\aterm}$ in 
    $\anordinalp'<\anordinalp$ turns.
    By the induction hypothesis, we get 
    $\winsof{\acontextd{\atermp}}{\objective}$ for all 
    $\acontextd{\contextvar}\in\acontextset$.
    This is exactly $\successorsof{\acontext{\atermp}}$,
    so we have $\winsof{\acontext{\atermp}}{\objective}$ 
    as well.
    \end{proof}

\section{Soundness: Missing Axiom Proofs}\label{Appendix:MissingAxioms}

\begin{proof}[Proofs of the remaining axioms]
    In all of the soundness proofs, we conclude with 
    \Cref{Lemma:ContextLemma} to generalize from contexts 
    for which at least one side of the conclusion is immediate, 
    to all contexts.
    To avoid repetition, this conclusion is omitted.

    Most axioms have a direct proof.
    Only \Cref{axiom:lattice-mono,axiom:lfp} require a simultaneous induction on the proof structure.
    The induction is kept implicit.

    \textbf{\Cref{axiom:lfp}:} 
    The proof makes heavy use of substitution.
    For the sake of readability, we use $\replaceLFPd{\atermp}$ 
    for $\atermp\replaceLFP$.
    We also extend this notation to contexts and write $\acontextdLFP{\contextvar}$ 
    for the context obtained from $\acontextd{\contextvar}\in\contexts$ by replacing all
    non-terminals $\nonterminal\in\nonterminals$ (different from $\contextvar$) by 
    $\aterm_{\nonterminal}$.

    Assuming $\replaceLFPd{\eqmapof{\nonterminal}}\axleq\atermn{\nonterminal}$ holds for all
    $\nonterminal\in\nonterminals$, the axiom yields $\nonterminalp\axleq\aterm_{\nonterminalp}$. 
    We proceed by an (outer) induction on the ordinal height of proof trees. 
    The induction hypothesis yields   
    $\replaceLFPd{\eqmapof{\nonterminal}}\congleq\atermn{\nonterminal}$ for all
    $\nonterminal\in\nonterminals$.
    We have to show $\nonterminalp\congleq\aterm_{\nonterminalp}$. 
    Consider an objective~$\objective$ and a context~$\acontextd{\contextvar}$ for which $\aterm_{\nonterminalp}$ or~$\nonterminalp$ is immediate.  
    We prove that $\winsof{\acontextd{\nonterminalp}}{\objective}$ implies~$\winsof{\acontextd{\aterm_{\nonterminalp}}}{\objective}$ with a detour. 
    If Eve wins $\objective$ from $\nonterminalp$, she does so in  $\anordinalp$-many moves, with $\beta$ an ordinal, \Cref{Appendix:TreeBounds}.
    We apply transfinite induction on $\beta$ to establish the following more general statement. 
    For all contexts $\acontext{\contextvar}$ and all terms $\aterm$, if Eve wins $\objective$ from $\acontext{\aterm}$ in~$\beta$ moves, then she wins $\objective$ from $\acontext{\replaceLFPd{\aterm}}$. 
    Letting $\acontext{\contextvar}=\acontextd{\contextvar}$ and $\aterm=\nonterminalp$ gives us the desired conclusion.

    %
    The base case $\anordinalp=0$ is trivial: the term $\acontext{\aterm}$ must be a word term and hence $\acontext{\aterm}=\acontext{\replaceLFPd{\aterm}}$.
    Before moving on with the inductive step, we make a remark.
    Since the urgency of non-terminals is maximal, we have $\urgencyof{\nonterminal}\geq\urgencyof{\aterm}$ for all non-terminals $\nonterminal$ and all terms $\aterm$.
    Moreover, the urgency of a term is monotonic in the urgency 
    of its subterms, so in particular $\urgencyof{\atermp}\geq\urgencyof{\replaceLFPd{\atermp}}$ holds.

    In the inductive step, let $\aterm$ be a term and $\acontext{\contextvar}$ be a context so that Eve wins~$\objective$ from $\acontext{\aterm}$ in~$\beta$ moves. 
    We first consider the case that $\aterm$ is 
    paused for $\acontext{\contextvar}$.
    Since
    $\urgencyof{\aterm}\geq\urgencyof{\replaceLFPd{\aterm}}$,
    \Cref{Lemma:TermLeads} tells us that $\replaceLFPd{\aterm}$ is also paused for $\acontext{\contextvar}$.
    Similar to the proof of \Cref{Lemma:ContextLemma}, we can use the induction 
    hypothesis to argue that Eve wins~$\objective$ from $\acontext{\replaceLFPd{\aterm}}$. 
    It is worth noting that in the paused case the move will change the surrounding context, which is why we strengthened the inductive statement to universally quantify over contexts. 

    Assume $\aterm$ is immediate for $\acontext{\contextvar}$. 
    Let $\atermp=\leadingof{\aterm}$ be the leading subterm and recall that $\aterm = \enclosingctx{\aterm}{\atermp}$, the term can be written as the unique context enclosing the leading subterm with the leading subterm inserted. 
    The substitution distributes to all subterms and we also have $\replaceLFPd{\aterm}=\enclosingctxLFP{\aterm}{\replaceLFPd{\atermp}}$. 
    %
    %
    %
    Showing that Eve wins from 
    $\acontext{\replaceLFPd{\aterm}}$ thus means to show that she wins from 
    $\acontext{\enclosingctxLFP{\aterm}{\replaceLFPd{\atermp}}}$.
    Since $\atermp$ is leading in $\aterm$, it is a choice or a non-terminal.
    We begin with the choice, $\atermp=\bigPchoiceOf{\anurg}\atermsetp$. 
    We argue that $\replaceLFPd{\atermp}$ must be immediate
    for $\acontext{\enclosingctxLFP{\aterm}{\contextvar}}$.
    To see this, note that  $\aterm$ is immediate for $\acontext{\contextvar}$ and so $\atermp$ is 
    immediate for $\acontext{\enclosingctx{\aterm}{\contextvar}}$. 
    The substitution distributes over the choice and we have $\replaceLFPd{(\bigPchoiceOf{\anurg}\atermsetp)}=
    \bigPchoiceOf{\anurg}\set{\replaceLFPd{\atermpp} \mid  \atermpp\in\atermsetp}$. 
    This shows $\urgencyof{\atermp}=\urgencyof{\replaceLFPd{\atermp}}$.
    For the outermost actions in $\acontext{\enclosingctx{\aterm}{\contextvar}}$, the substitution can only lower the urgencies. 

    The fact that $\replaceLFPd{\atermp}$ is  immediate
    for $\acontext{\enclosingctxLFP{\aterm}{\contextvar}}$ yields 
    \begin{align*}
    \successorsof{\acontext{\enclosingctxLFP{\aterm}{\makeleading{\replaceLFPd{\atermp}}}}}\quad=\quad
        \set{\acontext{\enclosingctxLFP{\aterm}{\replaceLFPd{\atermpp}}}\mid \atermp\gamemove\atermpp}
        \quad=\quad\set{\acontext{\replaceLFPd{\atermppp}}\mid \aterm\gamemove\atermppp}.
    \end{align*}
    We also have $\successorsof{\acontext{\makeleading{\aterm}}}=
    \set{\acontext{\atermppp}\mid \aterm\gamemove\atermppp}$.
    Similar to the proof of \Cref{Lemma:ContextLemma},
    Eve wins from $\acontext{\atermppp}$ in 
    $\anordinalp'<\anordinalp$ turns
    for all/one $\atermppp\in\successorsof{\aterm}$,
    depending on the owner of $\acontext{\aterm}$. 
    For every successor $\atermppp$, the induction 
    hypothesis tells us that 
    $\winsof{\acontext{\atermppp}}{\objective}$ implies 
    $\winsof{\acontext{\replaceLFPd{\atermppp}}}{\objective}$. 
    Since the owner of both $\acontext{\aterm}$ and 
    $\acontext{\replaceLFPd{\aterm}}$ is the owner of the 
    choice $\ovoid$, Eve can copy her strategy. 

    It remains to consider the case that $\atermp$ is a non-terminal 
    $\nonterminal$.
    Then, $\acontext{\enclosingctx{\aterm}{\makeleading{\nonterminal}}}$ 
    can only be played into 
    $\acontext{\enclosingctx{\aterm}{\eqmapof{\nonterminal}}}$ and Eve wins from this position in $\anordinalp'<\anordinalp$ 
    moves.
    By the hypothesis,  
    $\acontext{\replaceLFPd{\enclosingctx{\aterm}{\eqmapof{\nonterminal}}}}$
    is also won by Eve.
    We can write this term as 
    $\acontext{\enclosingctxLFP{\aterm}{\replaceLFPd{\eqmapof{\nonterminal}}}}$.
    The hypothesis of the outer induction yields  
    $\replaceLFPd{\eqmapof{\nonterminal}}\congleq\aterm_{\nonterminal}$.
    Therefore, Eve must also win from 
    $\acontext{\enclosingctxLFP{\aterm}{\aterm_{\nonterminal}}}=
    \acontext{\replaceLFPd{\enclosingctx{\aterm}{\nonterminal}}}
    =\acontext{\replaceLFPd{\aterm}}$.

    \textbf{\Cref{axiom:norm}:}
    We only show the case 
    $\bigEchoiceOf{\anurgp} \bigPchoiceOf{\anurg}\atermsetp
    \wineq\bigEchoiceOf{\anurgp} \bigPchoiceOf{\anurgp}\atermsetp$ with
    $\anurgp<\anurg$. 
    Like in the previous proof, both terms are immediate for the context 
    $\acontext{\contextvar}$ of interest.
    The key is to note that the inner choice cannot be resolved until 
    the outer choice has been made.
    The game arenas are:
    
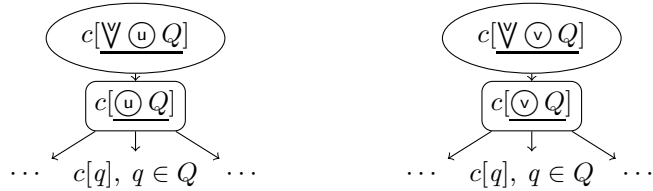
\begin{gametrees}
    \begin{tikzpicture}[scale=0.9]
        \tikzstyle{level 1}=[sibling distance=4.5em, level distance=10mm]
        \node[peve] {$\acontext{\makeleading{\bigEchoiceOf{\anurgp} \bigPchoiceOf{\anurg}\atermsetp}}$}
        child {
            node[punk] {$\acontext{\makeleading{\bigPchoiceOf{\anurg}\atermsetp}}$}
            child {node[draw=none] {$\cdots$} edge from parent[->]}
            child {node[draw=none] {$\acontext{\atermp}$, $\atermp\in\atermsetp$} edge from parent[->]}
            child {node[draw=none] {$\cdots$} edge from parent[->]}
            edge from parent[->]
        };
    \end{tikzpicture}%
    \hfill
    \begin{tikzpicture}[scale=0.9]
        \tikzstyle{level 1}=[sibling distance=4.5em, level distance=10mm]
        \node[peve] {$\acontext{\makeleading{\bigEchoiceOf{\anurgp} \bigPchoiceOf{\anurgp}\atermsetp}}$}
        child {
            node[punk] {$\acontext{\makeleading{\bigPchoiceOf{\anurgp}\atermsetp}}$}
            child {node[draw=none] {$\cdots$} edge from parent[->]}
            child {node[draw=none] {$\acontext{\atermp}$, $\atermp\in\atermsetp$} edge from parent[->]}
            child {node[draw=none] {$\cdots$} edge from parent[->]}
            edge from parent[->]
        };
    \end{tikzpicture}%
    \end{gametrees}
    As before, translation of strategies is straightforward. 

    \textbf{\Cref{axiom:least}:} 
    To show $\bot\congleq\aterm$, note that Eve never wins from a term $\acontext{\bot}$ and hence the implication is trivial.  

    \textbf{\Cref{axiom:lattice-mono}}: 
    Assume $\aterm_{i}\axleq\atermp_{i}$ for all $i\in I$ for which also $\aterm_{i}\congleq\atermp_{i}$ holds.
    Further let $\bigPchoiceOf{\anurg}\set{\aterm_{i}\mid i\in I}$
    and $\bigPchoiceOf{\anurg}\set{\atermp_{i}\mid i\in I}$ 
    be valid terms (i.e.\ they satisfy the depth constraints
    from \Cref{Appendix:Infinitary}).
    Acquire $\acontext{\contextvar}$ for which one of these terms 
    are immediate.
    Since $\urgencyof{\bigPchoiceOf{\anurg}\set{\aterm_{i}\mid i\in I}}
    =\urgencyof{\bigPchoiceOf{\anurg}\set{\atermp_{i}\mid i\in I}}
    =\anurg$,  
    \Cref{Lemma:TermLeads} tells us that 
    both terms are immediate for
    $\acontext{\contextvar}$.
    Fix an objective $\objective\subseteq\analph^{*}$.
    Let 
    $\winsof{\acontext{\makeleading{
        \bigPchoiceOf{\anurg}\set{\aterm_{i}\mid i\in I}
    }}}{\objective}$.
    We have that $\successorsof{\acontext{\makeleading{
        \bigPchoiceOf{\anurg}\set{\aterm_{i}\mid i\in I}
    }}}=\set{\acontext{\aterm_{i}}\mid i\in I}$.
    Then, for some $i \in I$ ($\ovoid=\echoice$)
    [for all $i \in I$ ($\ovoid=\achoice$)]
    holds $\winsof{\acontext{\aterm_{i}}}{\objective}$.
    Since $\aterm_{i}\congleq\atermp_{i}$, 
    also $\winsof{\acontext{\atermp_{i}}}{\objective}$.
    %
    Thus, $\winsof{\acontext{\makeleading{
        \bigPchoiceOf{\anurg}\set{\atermp_{i}\mid i\in I}
    }}}{\objective}$.

    \textbf{\Cref{axiom:lattice-assoc}}:
    Let $\bigPchoiceOf{\anurg}_{i\in I}\bigPchoiceOf{\anurg}\atermset_{i}$
    and 
    $\bigPchoiceOf{\anurg}\bigcup_{i\in I}\atermset_{i}$ be terms.
    Let $\acontext{\contextvar}$ be a context where one, 
    and by \Cref{Lemma:TermLeads} both, terms are immediate.
    After one move from $\acontext{
        \makeleading{\bigPchoiceOf{\anurg}_{i\in I}
        \bigPchoiceOf{\anurg}\atermset_{i}}}$, 
    the resulting term is always of the form 
    $\acontext{\bigPchoiceOf{\anurg}\atermset_{i}}$ 
    for some $i\in I$.
    Since $\urgencyof{\bigPchoiceOf{\anurg}\atermset_{i}}=\anurg$,
    \Cref{Lemma:TermLeads} states that this term is 
    also immediate for $\acontext{\contextvar}$.
    Then $\successorsof{\acontext{\makeleading{
        \bigPchoiceOf{\anurg}\atermset_{i}}}}=
    \set{\acontext{\aterm}\mid \aterm\in\atermset_{i}}$.
    This position is owned by the same player that owns the 
    initial position, $\acontext{
        \makeleading{\bigPchoiceOf{\anurg}_{i\in I}
        \bigPchoiceOf{\anurg}\atermset_{i}}}$
    and the position $\acontext{\makeleading{
        \bigPchoiceOf{\anurg}\bigcup_{i\in I}\atermset_{i}
        }}$.
    Then, in two moves, this player reaches 
    $\bigcup_{i\in I}\set{\acontext{\aterm}\mid \aterm\in\atermset_{i}}=
    \set{\acontext{\aterm}\mid \aterm\in\bigcup_{i\in I}\atermset_{i}}$
    from $\acontext{
        \makeleading{\bigPchoiceOf{\anurg}_{i\in I}
        \bigPchoiceOf{\anurg}\atermset_{i}}}$.
    We also have 
    $\successorsof{\acontext{\makeleading{
        \bigPchoiceOf{\anurg}\bigcup_{i\in I}\atermset_{i}
        }}}
    =\set{\acontext{\aterm}\mid \aterm\in\bigcup_{i\in I}\atermset_{i}}$.
    So it is straightforward to 
    lift the strategies from one term to the other 
    easily under any objective.

    \textbf{\Cref{axiom:lattice-ord}}:
    Let $\aterm,\atermp\in\terms$ with $\urgencyof{\aterm}\leq\anurg$.
    Acquire a context $\acontext{\contextvar}$ for which 
    one of $\aterm$ or $\aterm\echoicen{\anurg}\atermp$
    is immediate.
    If $\aterm$ is immediate for $\acontext{\contextvar}$, 
    since $\urgencyof{\aterm}\leq\anurg=
    \urgencyof{\aterm\echoicen{\anurg}\atermp}$, 
    $\aterm\echoicen{\anurg}\atermp$ is also immediate 
    for $\acontext{\contextvar}$.
    Then, in any case, $\aterm\echoicen{\anurg}\atermp$ is 
    immediate for $\acontext{\contextvar}$.
    Fix an objective $\objective\subseteq\analph^{*}$ and 
    let $\winsof{\acontext{\aterm}}{\objective}$.
    Since $\acontext{\makeleading{\aterm\echoicen{\anurg}\atermp}}$,
    Eve can choose $\aterm$ in the 
    inserted choice to reach $\acontext{\aterm}$ and 
    win, i.e. 
    $\winsof{
        \acontext{\makeleading{\aterm\echoicen{\anurg}\atermp}}
    }{\objective}$.

    \textbf{\Cref{axiom:lattice-abs}}:
    Let $\aterm,\atermp\in\terms$ and let 
    $\urgencyof{\aterm}\leq\anurg$.
    We will only show 
    $\aterm\congeq\aterm\echoicen{\anurg}
            (\aterm\achoicen{\anurg}\atermp)$.
    The proof of the dual statement is analogous.
    Fix an objective $\objective\subseteq\analph^{*}$.
    Acquire a context $\acontext{\contextvar}$ for 
    which one of $\aterm$ or 
    $\aterm\echoicen{\anurg}(\aterm\achoicen{\anurg}\atermp)$
    be immediate.
    Similarly to \labelcref{axiom:lattice-ord}, 
    $\aterm\echoicen{\anurg}(\aterm\achoicen{\anurg}\atermp)$
    is guaranteed to be immediate.
    Let $\winsof{\acontext{\aterm}}{\objective}$.
    Then Eve can play 
    $\acontext{
        \makeleading{
            \aterm\echoicen{\anurg}(\aterm\achoicen{\anurg}\atermp)
        }
    }\gamemove\acontext{\aterm}$ 
    and win, so $\winsof{\acontext{
        \makeleading{
            \aterm\echoicen{\anurg}(\aterm\achoicen{\anurg}\atermp)
        }
    }}{\objective}$ as well.
    Let $\winsof{\acontext{\aterm}}{\objective}$ not hold.
    Then Adam has a winning strategy
    from $\acontext{\aterm}$, since reachability games 
    are determined.
    If Eve were to play 
    $\acontext{
        \makeleading{
            \aterm\echoicen{\anurg}(\aterm\achoicen{\anurg}\atermp)
        }
    }\gamemove\acontext{\aterm}$,
    Adam would win from this position.
    If Eve were to instead play 
    $\acontext{
        \makeleading{
            \aterm\echoicen{\anurg}(\aterm\achoicen{\anurg}\atermp)
        }
    }\gamemove\acontext{\makeleading{\aterm\achoicen{\anurg}\atermp}}$,
    Adam can play 
    $\acontext{\makeleading{\aterm\achoicen{\anurg}\atermp}}\gamemove 
    \acontext{\aterm}$ and win.
    Thus $\winsof{\acontext{\makeleading{
            \aterm\echoicen{\anurg}(\aterm\achoicen{\anurg}\atermp)
        }
    }}{\objective}$ does not hold. 

    \textbf{\Cref{axiom:lattice-dist}}: Let 
    $\bigAchoiceOf{\anurg}_{i\in I}\bigEchoiceOf{\anurg}\atermset_{i}\in\terms$.
    Per Axiom of choice, $\set{f\mid f:I\to\atermset_{I}}\neq\emptyset$
    and $\bigEchoiceOf{\anurg}_{f:I\to\atermset_{I}}
    \bigAchoiceOf{\anurg}\set{f(i)\mid i \in I}$ is well defined.
    Let $\bigEchoiceOf{\anurg}_{f:I\to\atermset_{I}}
    \bigAchoiceOf{\anurg}\set{f(i)\mid i \in I}\in\terms$.
    Fix an objective $\objective\subseteq\analph^{*}$ and 
    acquire a context $\acontext{\contextvar}$ 
    for which one of $\bigEchoiceOf{\anurg}_{f:I\to\atermset_{I}}
    \bigAchoiceOf{\anurg}\set{f(i)\mid i \in I}$
    and $\bigAchoiceOf{\anurg}_{i\in I}\bigEchoiceOf{\anurg}\atermset_{i}$
    is immediate.
    Since both terms have urgency $\anurg$, \Cref{Lemma:TermLeads}
    states that both terms are immediate.
    Per definition, we have 
    $\winsof
    {\acontext{
        \makeleading{
            \bigAchoiceOf{\anurg}_{i\in I}\bigEchoiceOf{\anurg}\atermset_{i}
        }
    }}
    {\objective}$
    if and only if for all $i\in I$, there is a $\aterm\in\atermset_{i}$ 
    with $\winsof{\acontext{\aterm}}{\objective}$.
    Since we assume axiom of choice, we can apply Skolemization 
    to get the following equivalent statement: 
    There is a $f:I\to \atermset_{I}$ where for all $i\in I$,
    $\winsof{\acontext{f(i)}}{\objective}$.
    But this is equivalent to 
    $\winsof
    {\acontext{
        \makeleading{
            \bigEchoiceOf{\anurg}_{f:I\to\atermset_{I}}
            \bigAchoiceOf{\anurg}\set{f(i)\mid i \in I}
        }
    }}
    {\objective}$.
    Then $\winsof
    {\acontext{
        \makeleading{
            \bigEchoiceOf{\anurg}_{f:I\to\atermset_{I}}
            \bigAchoiceOf{\anurg}\set{f(i)\mid i \in I}
        }
    }}
    {\objective}$ if and only if 
    $\winsof
    {\acontext{
        \makeleading{
            \bigAchoiceOf{\anurg}_{i\in I}\bigEchoiceOf{\anurg}\atermset_{i}
        }
    }}
    {\objective}$.

    \textbf{\Cref{axiom:dist-right}:} 
    Let $\aterm,
    \bigPchoiceOf{\anurg}\atermsetp\in\terms$ where 
    $\urgencyof{\aterm}\leq\anurg$.
    Acquire a context for which at least one, and 
    per \Cref{Lemma:TermLeads} both, of 
    $(\bigPchoiceOf{\anurg}\atermsetp)\appl\aterm$ and
    $\bigPchoiceOf{\anurg}
    \set{\atermp\appl\aterm\mid \atermp\in\atermsetp}$
    are immediate.
    We have the owner
    $\ownof{\acontext{
        \makeleading{(\bigPchoiceOf{\anurg}\atermsetp)\appl\aterm}}}=
    \ownof{\acontext{\bigPchoiceOf{\anurg}
    \set{\atermp\appl\aterm\mid \atermp\in\atermsetp}}}$.
    It follows that we have 
    $\successorsof{
        \acontext{
        \makeleading{(\bigPchoiceOf{\anurg}\atermsetp)\appl\aterm}}
    }=\set{\acontext{\atermp\appl\aterm}\mid \atermp\in\atermsetp}$
    due to $\successorsof{(\makeleading{\bigPchoiceOf{\anurg}\atermsetp})
    \appl\aterm}=\set{\atermp\appl\aterm\mid \atermp\in\atermsetp}$.
    Further, $\successorsof{\acontext{\bigPchoiceOf{\anurg}
    \set{\atermp\appl\aterm\mid \atermp\in\atermsetp}}}=
    \set{\acontext{\atermp\appl\aterm}\mid \atermp\in\atermsetp}$
    as well.
    So under any objective $\objective\subseteq\analph^{*}$,
    lifting the strategies from one term to the other is straightforward.

    \textbf{\Cref{axiom:monoid}:}
    Let $\aword, \awordp\in\wordterms$ with $\aword\sigeq\awordp$.
    Acquire a context $\acontext{\contextvar}$ that is immediate
    for one of $\aword$ or $\awordp$.
    Since these are both word terms, 
    the only way one of these terms can be immediate 
    is if $\acontext{\contextvar}$ is a concatenation of 
    terminals and $\contextvar$.
    Then, we have $\acontext{\aword}\sigeq\acontext{\awordp}$.
    For an objective $\objective\subseteq\analph^{*}$,
    we also see that
    $\winsof{\acontext{\aword}}{\objective}$ if and only if 
    $\winsof{\acontext{\awordp}}{\objective}$.

    \textbf{\Cref{axiom:spec}:}
    Let $\objective\subseteq\analph^*$ be an objective.
    Let $\aword,\awordp\in \analphbot^*$ be word terms
    with  $\aword\axleq_{\objective}\awordp$ due to $\aword\sgleq{\objective}\awordp$.
    Let $\acontext{\contextvar}$ be a context for which one of the words is immediate.
    As the words have urgency zero, the context muss be a word as well: $\acontext{\contextvar}=\awordpp\appl\contextvar\appl\awordppp$
    (ignoring the bracketing) for
    $\awordpp,\awordppp\in \analphbot^*$.
    Assume Eve wins~$\objective$ from $\acontext{\aword}=\awordpp\appl\aword\appl\awordppp$.
    Then $\awordpp\appl\aword\appl\awordppp\in\objective$.
    By definition of $\sgleq{\objective}$, we get
    $\awordpp\appl\awordp\appl\awordppp\in\objective$.
    So Eve wins $\objective$ from $\acontext{\awordp}$ as well.

    \textbf{\Cref{axiom:eqmap}:} Let $\nonterminal\in\nonterminals$.
    Acquire a context $\acontext{\contextvar}$ for which at least 
    one of $\nonterminal$ or $\eqmapof{\nonterminal}$ 
    is immediate.
    Since 
    $\urgencyof{\nonterminal}=\maxurg\geq\eqmapof{\nonterminal}$,
    \Cref{Lemma:TermLeads} 
    tells us that $\nonterminal$ is guaranteed to be immediate.
    The term $\acontext{\makeleading{\nonterminal}}$ 
    has exactly one successor, $\acontext{\eqmapof{\nonterminal}}$.
    Then under any objective $\objective\subseteq\analph^{*}$,
    $\winsof{\acontext{\makeleading{\nonterminal}}}{\objective}$ 
    if and only if 
    $\winsof{\acontext{\eqmapof{\nonterminal}}}{\objective}$.
\end{proof}
\section{Normalization}\label{Appendix:Normalization}
We rely on proof rules that can be derived from the axioms.
These follow from the axioms in \Cref{fig:axioms}.
Before we address normalization,
we take a brief detour to prove the utilized proof rules correct.
These will also be used in Appendix~\Cref{Appendix:SpecializedNormalForm}.
\begin{center}
	\begin{axdefenv}{\rrepname}
		\axdefleft{\rrepname}
    		{\forall \nonterminal\in\nonterminals.\;
    			\aterm_{\nonterminal}\axleq\nonterminal}
    		{\atermp\replaceLFP \axleq \atermp}%
		\label{axiom:rep}
	\end{axdefenv}
	\hspace{0.6cm}
	\begin{axdefenv}{\rsinglename}
		\axdefleft{\rsinglename}
			{\urgencyof{\aterm}\leq \anurg}
			{\aterm\axeq\bigPchoiceOf{\anurg}{}{\aterm}}%
		\label{axiom:single}
	\end{axdefenv}\\
    \begin{axdefenv}{\rdnormname}
        \axdefleft{\rdnormname}
            {\anurgp<\anurg}
            {\bigAchoiceOf{\anurgp}{}{\bigPchoiceOf{\anurg}{}{\atermsetp}}
		        \axeq\bigAchoiceOf{\anurgp}{}{\bigPchoiceOf{\anurgp}{}{\atermsetp}}}    
        \label{axiom:dual-norm}
    \end{axdefenv}
\end{center}

\begin{proof}[Proof sketch]
    For \labelcref{axiom:rep}, we use \labelcref{axiom:lattice-dist} and the congruence rule inductively on the subterms.
    Utilizing \labelcref{axiom:lattice-abs} twice yields \labelcref{axiom:single}. 
    \labelcref{axiom:dual-norm} follows from 
    \begin{multline*}
        \bigAchoiceOf{\anurgp}_{i\in I}{ \bigPchoiceOf{\anurg}{}{\atermsetp_i}}
        \axeqper{\labelcref{axiom:single}} \bigEchoiceOf{\anurgp}{}{\bigAchoiceOf{\anurgp}_{i \in I}{ \bigPchoiceOf{\anurg}{}{\atermsetp_i}}}
        \axeqper{\labelcref{axiom:lattice-dist}} \bigAchoiceOf{\anurgp}_{i \in I}{\bigEchoiceOf{\anurgp}{}{\bigPchoiceOf{\anurg}{}{\atermsetp_i}}}\\
        \axeqper{\labelcref{axiom:norm}} \bigAchoiceOf{\anurgp}_{i \in I}{\bigEchoiceOf{\anurgp}{}{\bigPchoiceOf{\anurgp}{}{\atermsetp_i}}}
        \axeqper{\labelcref{axiom:single}} \bigAchoiceOf{\anurgp}_{i\in I}{ \bigPchoiceOf{\anurgp}{}{\atermsetp_i}}
    \end{multline*}
\end{proof}

\subsection{Proof of \Cref{Lemma:NTNormalization}, Part 1}\label{Appendix:NTFreeNormalization}
The function $\normof{\aterm}$ is defined by induction.
For $\aterm$ being a terminal or $\tskip, \terr$, 
we use \labelcref{axiom:single} to introduce a sequence of $2\maxurg$ choices 
over singleton sets and arrive at a term 
$\normof{\aterm}\axeq\aterm$ in normal form.
For a concatenation or choice, we recursively normalize the operands and then invoke specialized functions
that rely on the operands being normalized: 
\begin{align*}
\normof{\aterm\appl\atermp}\quad &=\quad
\normconcof{\normof{\aterm}.\normof{\atermp}}\\
\normof{\bigPchoiceOf{\anurg}\atermset}\quad &=\quad
\normchoiceof{\bigPchoiceOf{\anurg}\setcond{\normof{\aterm}}{\aterm\in\atermset}}\ . 
\end{align*}


\begin{lemma}\label{Lemma:ResolveChoice}
    Let $\atermsetpp\subseteq\gnfn{\maxurg}$
    and $\atermpp=\bigPchoiceOf{\anurg}\atermsetpp$.
    We can find $\normchoiceof{\atermpp}
    \in\gnfn{\maxurg}$ 
    with 
    $\normchoiceof{\atermpp}
    \axeq\atermpp$.
\end{lemma}
\begin{proof}
    Consider the proof of \Cref{Lemma:EffectiveResolveChoice}.
\end{proof}
\begin{lemma}\label{Lemma:ResolveConcat}
    For $\aterm,\atermp\in\gnfn{\maxurg}$
    we can find
    $\normconcof{\aterm\appl\atermp}\in\gnfn{\maxurg}$
    with $\normconcof{\aterm\appl\atermp}\axeq\aterm\appl\atermp$.
\end{lemma}
\begin{proof}
We strengthen the statement and show that for all urgencies $\anurg$, if we have normal form terms $\aterm,\atermp\in\gnfn{\anurg}$, then we can obtain a normal form in $\gnfn{\anurg}$. 
We proceed by induction on $\anurg$.
For the base case $\anurg=0$, \labelcref{axiom:monoid}
yields $\aterm\appl\atermp \sigeq \atermpp \in \gnfn{0}$. 
  
For the inductive case, let $\anurg>0$.
Then $\aterm=\bigEchoiceOf{\anurg}_{i\in I}\bigAchoiceOf{\anurg}\atermset_{i}$
and $\atermp=\bigEchoiceOf{\anurg}_{j\in J}\bigAchoiceOf{\anurg}\atermsetp_{j}$.
So we can write:
\begin{align*}
    \aterm\appl\atermp
    &\overset{\hphantom{2\times\labelcref{axiom:dist-right}}}{=}(\bigEchoiceOf{\anurg}_{i\in I}\bigAchoiceOf{\anurg}\atermset_{i})
    \appl(\bigEchoiceOf{\anurg}_{j\in J}\bigAchoiceOf{\anurg}\atermsetp_{j})\\
    &\axeqper{2\times\labelcref{axiom:dist-right}}\bigEchoiceOf{\anurg}_{i\in I}
    \bigAchoiceOf{\anurg}_{\aterm'\in\atermset_{i}}\aterm'\appl
    (\bigEchoiceOf{\anurg}_{j \in J}\bigAchoiceOf{\anurg}\atermsetp_{j})\\
    &\axeqper{2\times\labelcref{axiom:dist-left}}
    \bigEchoiceOf{\anurg}_{i\in I}
    \bigAchoiceOf{\anurg}_{\aterm'\in\atermset_{i}}
    \bigEchoiceOf{\anurg}_{j \in J}
    \bigAchoiceOf{\anurg}_{\atermp'\in\atermsetp_{j}}\aterm'\appl\atermp'\ . 
    \intertext{We have $\aterm',\atermp'\in\gnfn{\anurg-1}$.
    We apply I.H. to obtain 
    $\normconcof{\aterm'\appl\atermp'}\in\gnfn{\anurg-1}$}
    &\axeqper{\mathrm{I.H.}}
    \bigEchoiceOf{\anurg}_{i\in I}
    \bigAchoiceOf{\anurg}_{\aterm'\in\atermset_{i}}
    \bigEchoiceOf{\anurg}_{j \in J}
    \bigAchoiceOf{\anurg}_{\atermp'\in\atermsetp_{j}}\normconcof{\aterm'\appl\atermp'}\ . 
\end{align*}
The term is not in normal form due to the two layers of~$\anurg$ choices. 
We apply \Cref{Lemma:ResolveChoice} to obtain a normal form.
\end{proof}

\subsection{Proof of \Cref{Lemma:NTNormalization}, Part 2}\label{Appendix:NTNormalization}
    We proceed by a transfinite Kleene iteration.
    By induction on $\anordinal$, we construct normal form terms 
    $\nonterminal^{(\anordinal)}\in \gnfn{\maxurg}$ for all non-terminals
    $\nonterminal$ and all ordinals $\anordinal$.
    We write $\replaceLFPIt{\aterm}{\anordinalpp}$ to denote 
    $\aterm\replace{\nonterminals}{\nonterminals^{(\anordinalpp)}}$,  
    where $\nonterminals^{(\anordinalpp)}$ refers to a vector of 
    terms that has $\nonterminal^{(\anordinalpp)}$ as its 
    $\nonterminal$ component.
    We let $\nonterminal^{(0)}=\normof{\terr}$
    and for all $\anordinal>0$: 
    \begin{align*}
        \nonterminal^{(\anordinal)}\quad =\quad
        \normof{\bigEchoiceOf{\maxurg}
        \set{
            \replaceLFPIt{\eqmapof{\nonterminal}}{\anordinalp}
        \mid \anordinalp < \anordinal}}\ . 
    \end{align*}
    In both cases, we rely on the normalization from \Cref{Appendix:NTFreeNormalization}.
    
    We claim that 
    $\nonterminal^{(\anordinalp)}\axleq\nonterminal^{(\anordinal)}\axleq\nonterminal$
    for all ordinals $\anordinalp<\anordinal$.
    By definition, the alternatives available to Eve in $\nonterminal^{(\anordinalp)}$ are contained 
    in the alternatives available to her in $\nonterminal^{(\anordinal)}$.  
    %
    The former precongruence $\nonterminal^{(\anordinalp)}\axleq\nonterminal^{(\anordinal)}$ then follows from a 
    simple application of \Cref{axiom:lattice-mono,axiom:lattice-assoc,axiom:lattice-ord}.

    With a transfinite induction we show that for all ordinals~$\anordinal$ and for all non-terminals $\nonterminal$ we have $\nonterminal^{(\anordinal)}\axleq\nonterminal$.
    For the base case, we already have 
    $\nonterminal^{(0)}\axeq\terr\axleq\nonterminal$.
    For the inductive case, let $\anordinal$ be an ordinal so that 
    for all ordinals $\anordinalp<\anordinal$ and all non-terminals $\nonterminal$ we have $\nonterminal^{(\anordinalp)}\axleq\nonterminal$.
    %
    %
    Using \labelcref{axiom:rep}, we see that
    $\replaceLFPIt{\eqmapof{\nonterminal}}{\anordinalp}
    \axleq\eqmapof{\nonterminal}$.
    So we can apply \labelcref{axiom:lattice-mono} to get 
    \[
        \bigEchoiceOf{\maxurg}
    \setcond{\replaceLFPIt{\eqmapof{\nonterminal}}{\anordinalp}}{\anordinalp < \anordinal}
    \ \axleq\ 
    \bigEchoiceOf{\maxurg}
    \eqmapof{\nonterminal}\ 
    \axeqper{\labelcref{axiom:single}}\
    \eqmapof{\nonterminal}\
    \axeqper{\labelcref{axiom:eqmap}}\ 
    \nonterminal\ .
    \]
    
    It remains to show that $\nonterminal\axleq\nonterminal^{(\anordinalpp)}$ for some ordinal $\anordinalpp$. 
    The largest $\axleq$-chain of strictly increasing elements in $\gnfn{\maxurg}$ has size~$\sizeof{\agnfn{\maxurg}}$.
    Then the largest such chain in $\nonterminals\to \gnfn{\maxurg}$ has size $\sizeof{\nonterminals}\sizeof{\agnfn{\maxurg}}$.
    As a chain forms a well-ordered set, there is an ordinal $\anordinalpp$ having at least this size.   
    This means we have $\nonterminal^{(\anordinalpp)}\axeq\nonterminal^{(\anordinalpp+1)}$ for all non-terminals, and
    %
    %
    \begin{multline*}
        \eqmapof{\nonterminal}^{(\anordinalpp+1)} \axleqper{\labelcref{axiom:lattice-ord}} \eqmapof{\nonterminal}^{(\anordinalpp)} \;\echoicen{\maxurg}\;
        \bigEchoiceOf{\maxurg}{
            \setcond{\replaceLFPIt{\eqmapof{\nonterminal}}{\anordinalp}}{\anordinalp < \anordinalpp}
        } \\
        \axeqper{\labelcref{axiom:lattice-assoc}} \bigEchoiceOf{\maxurg}{
            \setcond{\replaceLFPIt{\eqmapof{\nonterminal}}{\anordinalp}}{\anordinalp < \anordinalpp+1}
        } 
        \axeq \nonterminal^{(\anordinalpp+1)} \axeq \nonterminal^{(\anordinalpp)}\ .
    \end{multline*}
    %
    %
    Applying \labelcref{axiom:lfp}, we get $\nonterminal \axleq \nonterminal^{(\anordinalpp)}$ 
    for all $\nonterminal\in\nonterminals$.
    %
    %

\section{Completeness: Missing Proof}\label{Appendix:CompletenessProof}
\begin{proof}
We focus on the former implication, the latter is simple. 
The implication trivially holds for terms $\aterm$ that are minimal in the domination preorder. 
We will need information about the shape of these minimal terms. 
In urgency zero, minimal are all terms $\aword\sgleq{\objective}\terr$, meaning there is no chance to extend~$\aword$ to a word in $\objective$.  
For higher urgencies $\anurg>0$, terms $\bigEchoiceOf{\anurg}\atermset$ are minimal where all elements in $\atermset$ are minimal. 
Terms $\bigAchoiceOf{\anurg}\atermset$ are minimal where an element in~$\atermset$ is minimal.
        %
        %

To show the implication for terms $\aterm$ and $\atermp$ where $\aterm$ is not minimal, we use \emph{characteristic contexts}. 
Given a normal form term $\aterm$, we construct a context $\charcontext{\aterm}$ so that for all normal form terms $\atermp$ of the same urgency as $\aterm$ and owned by the same player we have:
\begin{align}
\winsof{\charcontextof{\aterm}{\atermp}}{\objective}\quad\text{iff}\quad\aterm\discleq\atermp\ .\label{Equation:CharContexts}
\end{align} 
%
%
The implication from $\aterm\speccongleq{\objective}\atermp$ to $\aterm\discleq\atermp$ indeed follows. 
Since $\aterm\discleq\aterm$, we obtain $\winsof{\charcontextof{\aterm}{\aterm}}{\objective}$ by Equivalence~\eqref{Equation:CharContexts}. 
The assumption $\aterm\speccongleq{\objective}\atermp$ now yields $\winsof{\charcontextof{\aterm}{\atermp}}{\objective}$. 
Hence, again by Equivalence~\eqref{Equation:CharContexts}, we have $\aterm\discleq\atermp$. 
For the maximal urgency, we need a special treatment.

It remains to give the construction of $\charcontext{\aterm}$. 
It will have the shape $\contextvar\appl\atermpppp_{\aterm}$ where $\atermpppp_{\aterm}$ is again in normal form. 
We proceed by induction on the urgency $\anurg$ and need a special case for the maximal urgency $\maxurg$. 

\noindent\emph{Base Case $\aterm\in\gnfn{0}$.} 
We define 
\[
    \atermpppp_{\aterm}\quad=\quad
    \bigAchoiceOf{1}\setcond{\awordppp\in\analph^*}{\aterm.\awordppp\in\objective}\ .
\]
The set is non-empty as $\aterm$ is not minimal. 
Moreover, for every normal form term $\atermp$ of urgency zero, 
Equivalence~\eqref{Equation:CharContexts} holds by the definition of $\sgleq{\objective}$. 

\noindent\emph{Inductive Case
    $\aterm=\bigEchoiceOf{\anurg}\atermset\in\gnfn{\anurg}$, $\anurg<\maxurg$.}
    We define
    \[
        \atermpppp_{\aterm}\quad=\quad
        \Achoice_{\anurg + 1}
        \setcond{\atermpppp_{\atermpp}}{\atermpp\in\atermset\text{ not minimal}}\ .
    \]
    Note that we can increase the urgency because $\anurg<\maxurg$.
    A non-minimal $\atermpp$ is guaranteed to exist in $\atermset$ by the assumption that $\bigEchoiceOf{\anurg}\atermset$ itself is not minimal. 
    To prove Equivalence~\eqref{Equation:CharContexts}, we consider $\bigEchoiceOf{\anurg}\atermsetp\in\gnfn{\anurg}$ and argue as follows:
    \begin{align*}
    &\quad\winsof{(\bigEchoiceOf{\anurg}\atermsetp)\appl(\Achoice_{\anurg + 1}\setcond{\atermpppp_{\atermpp}}{\atermpp\in\atermset\text{ not minimal}})}{\objective}\\
    \text{iff}&\quad
    \forall \atermpp\in\atermset\text{ not minimal}.\; 
    \exists \atermppp\in\atermsetp.\;
    \winsof{\atermppp\appl\atermpppp_{\atermpp}}{\objective}\\ 
    \text{\{I.H.\}\quad iff}&\quad
    \forall \atermpp\in\atermset\text{ not minimal}.\; 
    \exists \atermppp\in\atermsetp.\;\atermpp\discleq\atermppp\\ 
    \text{iff}&\quad 
    \forall \atermpp\in\atermset.\; 
    \exists \atermppp\in\atermsetp.\;
    \atermpp\discleq\atermppp\\
    \text{iff}&\quad \bigEchoiceOf{\anurg}\atermset\;\discleq\;\bigEchoiceOf{\anurg}\atermsetp\ . 
    \end{align*}
The inductive case for $\aterm=\bigAchoiceOf{\anurg}\atermset\in\agnfn{\anurg}$, $\anurg\leq\maxurg$, is similar.

        \noindent\emph{Special Case} $\gnfn{\maxurg}$:\;
        Since the urgency $\maxurg+1$ is not allowed,
        we are not able to construct the context in the way we did above.
        Instead, we show 
        that for $\aterm,\atermp\in\gnfn{\maxurg}$ with 
        $\aterm\not\discleq\atermp$, there is a term
        $\atermpp\in\agnfn{\maxurg}$ where $\aterm\appl\atermpppp_{\atermpp}$ 
        is won by Eve and $\atermp\appl\atermpppp_{\atermpp}$ is won by Adam.
        This yields $\aterm\not\congleq\atermp$.

        Let $\aterm=\bigEchoiceOf{\maxurg}\atermset$ and 
        $\atermp=\bigEchoiceOf{\maxurg}\atermsetp$ both in $\gnfn{\maxurg}$ with $\aterm\not\discleq\atermp$.
        By definition of the domination preorder, there is $\atermpp\in\atermset$ so that for all 
        $\atermppp\in\atermsetp$ we have~$\atermpp\not\discleq\atermppp$.
        The term $\atermpp$ cannot be minimal, and hence~$\atermpppp_{\atermpp}$ is guaranteed to exist.
        We claim that Eve wins~$\objective$ from~$\aterm\appl\atermpppp_{\atermpp}$ while 
        Adam wins~$\objective$ from~$\atermp\appl\atermpppp_{\atermpp}$. 
        Note that the leading terms are 
        $\makeleading{\aterm}\appl\atermpppp_{\atermpp}$ resp.~$\makeleading{\atermp}\appl\atermpppp_{\atermpp}$.
        To see that Eve wins from 
        $\makeleading{\aterm}\appl\atermpppp_{\atermpp}$, let her choose $\atermpp\in\atermset\subseteq\agnfn{\maxurg}$ to 
        reach $\atermpp\appl\atermpppp_{\atermpp}$.
        Since $\atermpp\discleq\atermpp$, Equivalence~\eqref{Equation:CharContexts} yields $\winsof{\atermpp\appl\atermpppp_{\atermpp}}{\objective}$. 
        To see that Adam wins from~$\makeleading{\atermp}\appl\atermpppp_{\atermpp}$, let Eve choose 
        $\atermppp\in\atermsetp\subseteq\agnfn{\maxurg}$.
        As $\atermpp\not\discleq\atermppp$, Equivalence~\eqref{Equation:CharContexts} implies 
        that Eve loses from $\atermppp\appl\atermpppp_{\atermpp}$.
\end{proof}

\section{Completeness: Characteristic Context}
The missing inductive case is $\aterm=\bigAchoiceOf{\anurg}\atermset\in\agnfn{\anurg}$.
    We set
    \[
        \atermpppp_{\aterm}\quad=\quad
        \bigEchoiceOf{\anurg}\setcond{\atermpppp_{\atermpp}}{\atermpp\in\atermset}.
    \]
    Since $\aterm$ is not minimal, no element $\atermpp\in\atermset$ is minimal and hence
    the $\atermpppp_{\atermpp}$ are guaranteed to exist. 
    To see Equivalence~\eqref{Equation:CharContexts}, consider $\bigAchoiceOf{\anurg}\atermsetp\in\agnfn{\anurg}$:
    \begin{align*}
        &\quad\winsof{(\bigAchoiceOf{\anurg}\atermsetp)\appl(\bigEchoiceOf{\anurg}
        \setcond{\atermpppp_{\atermpp}}{\atermpp\in\atermset})}{\objective}\\
        \text{iff}&\quad
        \forall \atermppp\in\atermsetp.\; 
       \exists \atermpp\in\atermset.\;
        \winsof{\atermppp\appl\atermpppp_{\atermpp}}{\objective}\\ 
        \text{\{I.H.\}\quad iff}&\quad \forall \atermppp\in\atermsetp.\; 
       \exists \atermpp\in\atermset.\;
        \atermpp\discleq\atermppp. \\
\text{iff}&\quad \bigAchoiceOf{\anurg}\atermset\;\discleq\; \bigAchoiceOf{\anurg}\atermsetp. 
\end{align*}
\section{Completeness: Domain Shattering Objectives}\label{Appendix:CompletenessDomainShattering}

We show \Cref{Lemma:DomainShattering},
namely there are right-separating and domain-shattering objectives.
%
%
For $\sizeof{\analph}>1$, we take
\[
    \objective=\setcond{\aword\appl\aword^{\mathit{reverse}}}{\aword\in\analph^{*}}\ .
\]
For $\analph=\set{\aletter}$, we let $\objective=\setcond{\aletter^{(n^2)}}{n\in\nat}$. 
It is known that these objectives have a syntactic congruence with singleton classes. 
Even the syntactic precongruence is the reflexive relation, ${\lrleq{\objective}} = {\lreq{\objective}}$.

Let $\objective=\set{a^{(n^2)}\mid n\in \nat}$
if $\analph=\set{a}$ and
$\objective=
\set{\aword\appl\aword^{reverse}\mid \aword\in\analph^{*}}$
if $|\analph|>1$.
We show that for all  
$\aword, \awordp\in\wordterms$ with $\aword\not\sigeq\awordp$,
we can find a $\awordpp\in\analph^{+}$ with
$\aword\appl\awordpp\in\objective$ and 
$\awordp\appl\awordpp\not\in\objective$.
This proves that $\objective$ is right separating.
In particular, it also showcases that \Cref{axiom:spec} cannot relate terms $\aword, \awordp$ unless $\aword \sigeq \awordp$.
So it also proves that $\objective$ is domain-shattering.
This also shows for all $\aword\in\analph^{+}$
the existence of $\awordpp$ 
with $\aword\appl\awordpp\in\objective$ 
to establish the right separability of $\aword$ and $\bot$.

\textbf{Case $\sizeof{\analph}=1$:}
Let $\aword=\aletter^{i}$ and $\awordp=\aletter^{j}$ with $i\neq j$.
Then set $\awordpp=\aletter^{t^{2}-i}$ with $t=i + j+1$.
Obviously, $\aword\appl\awordpp \in \objective$.
Suppose $\aletter^{j}\appl\aletter^{t^{2}-i}\in\objective$.
Then, $j+t^{2} - i=k^{2}$ for some $k\in\nat$.
So $j - i=t^{2}-k^{2} = (t-k)\cdot(t + k)$.
$i\neq j$ implies $|t-k|\geq 1$.
Note that $k, i, j \geq 0$ 
so we get the contradiction $|j-i| < i+j < t+k < |t-k|\cdot |t-k| = |j-i|$.

\textbf{Case $\sizeof{\analph}>1$:}
Let $\aword = \aword_0\aword_1\ldots\aword_n$ and $\awordp = \awordp_0\awordp_1\ldots\awordp_m$.
Let $\aletter \neq \aword_m$ be a terminal.
We set $\awordpp=\aletter\appl\aletter\appl\aword^{reverse}$,
where 
$\aletter\neq\aword_{|\awordp|}$ if $|\aword|>|\awordp|$,
$\aletter\neq\awordp_{|\aword|}$ if $|\awordp|>|\aword|$, and
$\aletter$ arbitrary if $|\aword|=|\awordp|$.
Suppose $\awordp\appl\aletter\appl\aletter\appl\aword^{reverse}\in\objective$.
Then, $\awordp\appl\aletter\appl\aletter\appl\aword^{reverse}=\aword\appl\aletter\appl\aletter\appl\awordp^{reverse}$.
If $|\aword|=|\awordp|$, $\aword=\awordp$ is a contradiction.
Otherwise, if $|\awordp|<|\aword|$, we have 
${(\awordp\appl\aletter)}_{|\awordp|}=\aletter\neq\aword_{|\awordp|}
={(\aword\appl\aletter)}_{|\awordp|}$ for a contradiction.
Similarly, if $|\aword|<|\awordp|$, the contradiction is
${(\aword\appl\aletter)}_{|\aword|}=\aletter\neq\awordp_{|\aword|}
={(\awordp\appl\aletter)}_{|\aword|}$.

\section{Specialized Normal Form}\label{Appendix:SpecializedNormalForm}
We now provide a sketch of the normalization algorithm 
with the complexity $(\sizeof{\aterm}+\sizeof{\eqmap}\sizeof{\nonterminals})\cdot\sgexpof{2\maxurg-1}$
\footnote{We use $\sgexpof{\anurg}$ for $\repexpof{\anurg}{\bigoof{\sizeof{\specgnfn{0}}}}$. Usages of $\leq$ and $=$ are to be understood by means of $\in$ or $\subseteq$.}
as given in \Cref{Lemma:EffNormalForm}.

We use methods from \Cref{Section:Normalization} with 
slight modifications.
To avoid repetition, 
we only note the modifications that need to be made, 
instead of giving a full algorithm.
The section takes the form of constructing modified normalization
functions $\specnormof{.}$, $\specnormchoiceof{.}$, and 
$\specnormconcof{.}$.
We provide lemmas that reference those in 
\Cref{Section:Normalization}.
The $\specnormchoiceof{.}$ and $\specnormconcof{.}$ 
calls made by $\specnormof{.}$ are kept the same.
Note that implementing $\specnormchoiceof{.}$ and $\specnormconcof{.}$ 
with time complexity $|\aterm|\cdot\sgexpof{2\maxurg-1}$
for the input term $\aterm$ means that 
$\specnormof{\atermp}$ is also constructed in 
$|\atermp|\cdot\sgexpof{2\maxurg-1}$ time, 
if $\atermp$ has no non-terminals.
For the rest of the section,
fix a regular objective $\objective\subseteq\analph^{*}$,
a finite set of non-terminal symbols $\nonterminals$, 
and a finitary defining assignment 
$\eqmap:\nonterminals\to\terms$.

\begin{lemma}[Modifies
    \Cref{Lemma:ResolveChoice}]\label{Lemma:EffectiveResolveChoice}
    Let $\atermset\subseteq\specgnfn{\maxurg}$ and 
    $\atermpp=\bigPchoiceOf{\anurg}\atermset$. 
    We can find $\specnormchoiceof{\atermpp}\in\specgnfn{\maxurg}$
    with ${\specnormchoiceof{\atermpp}\specaxeq{\objective}\atermpp}$ in 
    $|\atermpp|\cdot\sgexpof{2\maxurg - 1}$ time.
\end{lemma}
\begin{proof}
    By induction on $\maxurg$. Note that for normalforms in $\specgnfn{\anurg}$ also have maximal urgency $\anurg$.
    In any case of $\anurg$ we need to normalize a term $\atermpp = \bigPchoiceOf{\anurg}_{i\in I}{\bigEchoiceOf{\maxurg}{}{\atermset_i}}$.
    The harder case is $\Pchoice = \Achoice$.
    We apply \labelcref{axiom:norm} if $\anurg < \maxurg$ and \labelcref{axiom:lattice-dist} to obtain $\bigEchoiceOf{\anurg}_{f:I \to \atermset_I}{\bigAchoiceOf{\anurg}_{i \in I}{f(i)}}$ and utilize \labelcref{axiom:lattice-assoc} to combine the $\bigAchoiceOf{\anurg}_{i\in I}{}$ with the $f(i) \in \specagnfn{\anurg}$:
    \[
        \specnormchoiceof{\atermpp} = \bigEchoiceOf{\anurg}_{f:I \to \atermset_I}{\bigAchoiceOf{\anurg}{}{\atermsetp_f}}\,
    \]
    where $\atermsetp_f = \bigcup \setcond{\atermsetp}{f(i) = \bigAchoiceOf{\anurg}{}{\atermsetp}, i \in I}$.
    Naively, applying \Cref{axiom:lattice-dist} requires us to enumerate all the choice functions $f:I\to\atermset_{I}$.
    However, we do not need to account for all of them, because after applying \labelcref{axiom:lattice-assoc} for a single $f$, we know that $\bigAchoiceOf{\anurg}{}{\atermsetp_f} \in \specagnfn{\anurg}$.
    Knowing this means that the distribution considers way more functions $f$ (namely $\prod_{i \in I} \sizeof{\atermset_i} \geq 2^{\sizeof{\specgnfn{\anurg}}}$) than there can be sets $\atermsetp_f$ ($\sizeof{\specagnfn{\anurg}}$).
    Instead, we can use \labelcref{axiom:lattice-assoc} to split the application of \labelcref{axiom:lattice-dist} into $\sizeof{I}-1$ many single applications of \labelcref{axiom:lattice-dist} and keeping the size of the intermediary results bound by $\sizeof{\specagnfn{\anurg}} = \sgexpof{2\anurg-1}$:
    In the case of $J \subseteq I$ with $\sizeof{J} = 2$, i.e.\ constant, the number of functions $f : J \to \atermset_J$ is bound by $\sizeof{\atermset_1} \cdot \sizeof{\atermset_2} \leq \sizeof{\specagnfn{\anurg}}^2$.
    For such an $f$ (with binary co-domain) the union $\atermsetp_f$ can be computed in time $\sizeof{\specgnfn{\anurg-1}}^2$ as long as the terms $f(j)$ are in normalform $f(j) \in \specagnfn{\anurg}$.
    In case of $\maxurg = \anurg$, this will be the case.
    Otherwise, $\anurg < \maxurg$ and we apply the induction hypothesis to obtain normalforms equivalent to $f(1)$ and $f(2)$ before computing $\atermsetp_f$.
    That way, we compute no more than $\sizeof{I}\leq \sizeof{\atermpp}$ many \labelcref{axiom:lattice-assoc} for splitting, $\sizeof{I}$ many \labelcref{axiom:lattice-dist} enumerating $\sizeof{\specagnfn{\anurg-1}}^2 \leq \sgexpof{2\anurg-1}$ functions $f$ each, and for each $f$ we apply \labelcref{axiom:lattice-assoc} to create the union $\atermsetp_f$ in time~$\sizeof{\specgnfn{\anurg}}^2\leq \sgexpof{2\anurg-2}$.
    Together (for $\maxurg = \anurg$):
    \[
        \underbrace{\sizeof{I}}_{\labelcref{axiom:lattice-assoc}}{} 
        + \sizeof{I} 
            \cdot \underbrace{\sizeof{\specagnfn{\anurg}}^2}_{\labelcref{axiom:lattice-dist}} 
            \cdot \underbrace{\sizeof{\specgnfn{\anurg-1}}^2}_{\labelcref{axiom:lattice-assoc}, \atermsetp_f} \leq \sizeof{\atermpp} \cdot \sgexpof{2\anurg-1}
    \]

    In case of $\anurg < \maxurg$, the right summand changes to
    \begin{multline*}
        \sizeof{I} 
            \cdot \underbrace{\sizeof{\specagnfn{\anurg}}^2}_{\labelcref{axiom:lattice-dist}} 
            \cdot (\underbrace{\sizeof{\specgnfn{\anurg-1}}^2}_{\labelcref{axiom:lattice-assoc}, \atermsetp_f} 
        + \underbrace{2\sizeof{\specgnfn{\maxurg-1}} \cdot \sgexpof{2\maxurg-3}}_{\text{I.H.}}) \\
        \leq \sizeof{\atermpp} \cdot \sgexpof{2\maxurg-1}
    \end{multline*}
    Adding singleton choice operators \labelcref{axiom:single} to reobtain a term in $\specgnfn{\maxurg}$ costs close to no time.
\end{proof}

\begin{lemma}[Modifies
    \Cref{Lemma:ResolveConcat}]\label{Lemma:EffectiveResolveConcat}
    For any $\aterm,\atermp\in\specgnfn{\maxurg}$, 
    we can find $\specnormconcof{\aterm\appl\atermp}
    \in\specgnfn{\maxurg}$
    with
    $\specnormconcof{\aterm\appl\atermp}
    \specaxeq{\objective}\aterm\appl\atermp$ 
    in $\sgexpof{2\maxurg - 1}$ time.
\end{lemma}

\begin{proof}
    The proof in \Cref{Section:Normalization} proceeds 
    by an induction on urgency.
    We change the base case to $\maxurg = 1$ and for both, base and inductive case, apply \labelcref{axiom:dist-left,axiom:dist-right} as in the proof of \Cref{Lemma:ResolveConcat}.
    The application of \labelcref{axiom:dist-left} creates at most $\sizeof{\aterm}$ copies of $\atermp$, and \labelcref{axiom:dist-right}
    creates another $\sizeof{\atermp}$ copies of each $\aterm'$.
    Call the resulting term $\atermppp$ with
    $|\atermppp|\leq 2\cdot|\aterm|\cdot|\atermp|$.
    
    In the base case, $\aterm'$ and $\atermp'$ belong to $\specgnfn{0}$. 
    We employ~\labelcref{axiom:eqmap} and find 
    an $\awordpp\in\specgnfn{0}$ with 
    $\awordpp\lreq{\objective}\aword\appl\awordp$.
    Finding the syntactical congruence class 
    of $\aword\appl\awordp$ can be done in time 
    ${(\sgexpof{0})}^{k}$ for some fixed $k\in\nat$ by naively
    checking.
    
    In the inductive case, we apply the induction hypothesis to normalize the 
    concatenative subterms $\aterm'\appl\atermp'$ into $\specgnfn{\maxurg}$.
    A term in $\specgnfn{\maxurg}$ term is bound by size $\sizeof{\specagnfn{\maxurg}} \leq \sgexpof{2\maxurg-1}$.
    After applying the induction hypothesis to each pair in $\atermppp$, the resulting term $\atermpppp$ 
    has size~$|\atermpppp|\leq2\cdot|\aterm|\cdot|\atermp|\cdot\sgexpof{2\maxurg-1} = \sgexpof{2\maxurg-1}$.
    In \Cref{Section:Normalization} 
    the layers of choices are resolved by invoking 
    \Cref{Lemma:ResolveChoice}.
    In our case, we employ the 
    modified \Cref{Lemma:EffectiveResolveChoice} for another $|\atermpppp|\cdot\sgexpof{2\maxurg-1}
    \leq{(\sgexpof{2\maxurg-1})}^{2}
    =\sgexpof{2\maxurg-1}$.
\end{proof}

\begin{lemma}[Modifies \Cref{Lemma:NTNormalization}]\label{Lemma:EffectiveResolveNT}
    For any $\nonterminal\in\nonterminals$, 
    we can find $\normof{\nonterminal}\in\specgnfn{\maxurg}$ 
    with 
    $\normof{\nonterminal}\specaxeq{\objective}\nonterminal$
    in $|\eqmap|\cdot|\nonterminals|
    \cdot\sgexpof{2\maxurg-1}$ time.
\end{lemma}

\begin{proof}[Proof Sketch]
    We construct an increasing chain of $\specgnfn{\maxurg}$ 
    terms using the same least fixed point construction from 
    \Cref{Section:Normalization}.
    The constructed term 
    for component $\nonterminal$ has size at most 
    $|\eqmapof{\nonterminal}|\cdot{(\sgexpof{2\maxurg - 1})}^{2}$.
    Per \Cref{Lemma:EffectiveResolveChoice,Lemma:EffectiveResolveConcat},
    this term is normalized in 
    $|\eqmapof{\nonterminal}|\cdot{(\sgexpof{2\maxurg - 1})}^{3}
    =|\eqmapof{\nonterminal}|\cdot\sgexpof{2\maxurg-1}$
    time.
    So normalizing all components takes 
    $|\eqmap|\cdot\sgexpof{2\maxurg-1}$ time.
    The chain
    converges in at most 
    $|\nonterminals|\cdot|\specagnfn{\maxurg}|=
    |\nonterminals|\cdot\sgexpof{2\maxurg-1}$
    iterations.
    So the normalization process takes 
    $|\eqmap|\cdot|\nonterminals|\cdot\sgexpof{2\maxurg-1}$
    time.
\end{proof}

We complete the normalization function
analogously to \Cref{Section:Normalization}.
To calculate $\specnormof{\aterm}$ for any finitary
$\aterm\in\terms$ with non-terminals,
the algorithm first constructs $\specnormof{\nonterminal}$ for all 
$\nonterminal\in\nonterminals$.
This takes $|\eqmap|\cdot|\nonterminals|\cdot\sgexpof{2\maxurg-1}$ 
time.
Then, we let 
$\specnormof{\aterm}=\specnormof{\aterm\replace{\nonterminals}{\aterm_{\nonterminals}}}$
where $\aterm_{\nonterminal}=\specnormof{\nonterminal}$.
The normalization takes $|\aterm|\cdot\sgexpof{2\maxurg-1}$
time
per \Cref{Lemma:EffectiveResolveChoice,Lemma:EffectiveResolveConcat}.
This concludes the construction of the algorithm and 
thus proves \Cref{Lemma:EffNormalForm}.

\section{Decidability and Complexity}\label{Appendix:Decidability}

\begin{proof}[Proof of \Cref{Proposition:Undec}]
    Consider two CFG~$\cfg_i = (\nonterminals_i, \eqmap_i, S_i), i \in \set{0,1}$ over some alphabet~$\analph$ and with disjoint sets of non-terminals~$\nonterminals_1 \cap \nonterminals_2 = \emptyset$.
    Note that each non-terminal has a single production rule of the shape $\eqmap_i(\nonterminal) = \aword_1 \mid \ldots \mid \aword_n$, $\aword_i \in {(\nonterminals \cup \analph)}^*$.
    We reproduce this non-deterministic structure by yielding the choice to Eve. 
    We construct the program-term grammar $(\nonterminals_1 \cup \nonterminals_2, \eqmap)$,
    where $\eqmap$ is defined via $ \eqmap(\nonterminal) = \bigEchoiceOf{1} \set{\aword_1, \ldots, \aword_n}$ for each production rule of above shape.
    Kleene iteration yields the normal form 
    \[
        \normof{S_i} = \bigEchoiceOf{1} \setcond{\bigAchoiceOf{1} \aword}{\aword \in L(\cfg_i)} \congeq \bigEchoiceOf{1} L(\cfg_i)\,.
    \]
    We utilize ${\speccongleq{\objective}} = {\congleq}$ by choosing the shattering objective $\objective=\setcond{\aword\appl\aword^{\mathit{reverse}}}{\aword\in\analph^{*}}$.
    Since $\speccongleq{\objective}$ is the reflexivity relation on $\Sigma^* \times \Sigma^*$, the domination pre-order yields $\bigEchoiceOf{1} L \disceq \bigEchoiceOf{1} M$ if and only if $ L = M $ for any $L,M \subseteq \Sigma^*$.
    Using \Cref{Lemma:DominationPreorder}, \Cref{Theorem:FullAbstractionSpecialized}, and that $\objective$ is shattering and right-separating, $L(\cfg_1) = L(\cfg_2)$ if and only if
    $S_1 \congeq \bigEchoiceOf{1} L(\cfg_1) \congeq \bigEchoiceOf{1} L(\cfg_2) \congeq S_2$. 
\end{proof}

\begin{proof}[Proof of \Cref{Theorem:UpperLower}]
    We show that 
    \decwhichover{\speccongleq{\objective}}{\maxurg} is 
    $\mathsf{PTIME}$-hard (wrt. $\log$-space reductions) 
    by sketching out a reduction
    from the Monotonic Circuit Value Problem~\cite{goldschlager1977monotone}
    The problem consists of assignments of boolean 
    variables $\mathsf{P}_{i\leq n}$ to $\vee$/$\wedge$ 
    clauses built out 
    of the variables
    $\mathsf{P}_{j<i}$, to $\mathsf{true}$,
    or to $\mathsf{false}$.
    The input is accepted 
    if and only if the variable $\mathsf{P}_{n}$ evaluates 
    to true.
    Let non-empty $\objective\subseteq\analph^{*}$ 
    and let $\aword\in\objective$.
    For each boolean variable $\mathsf{P}_i$, the $\log$-space 
    Turing Machine
    outputs a non-terminal 
    $P_i$ and a defining equation $\eqmapof{P_i}$.
    The machine outputs $\eqmapof{P_i}=\aword$
    if $\mathsf{P}_i=\mathsf{true}$, and it outputs 
    $\eqmapof{P_i}=\terr$ 
    if $\mathsf{P}_i=\mathsf{false}$.
    For $\mathsf{P}_i=\mathsf{P}_{j_0}
    \wedge \ldots \wedge \mathsf{P}_{j_k}$, the machine outputs
    $\eqmapof{P}=\bigAchoiceOf{1}\set{P_{j_l}\mid  l \leq k}$.
    For $\mathsf{P}_i=\mathsf{P_{j_0}}
    \vee \ldots\vee \mathsf{P_{j_k}}$, it outputs
    $\bigEchoiceOf{1}\set{P_{j_l}\mid l\leq k}$.
    The reduced problem instance asks whether 
    $\aword\speccongleq{\objective} P_{n}$. 
    By induction on $i$ 
    we show  
    $P_{i}\specaxeq{\objective}\aword$ if 
    $\mathsf{P}_{i}$ evaluates to true, 
    and 
    $P_{i}\specaxeq{\objective}\aword$ and if 
    $\mathsf{P}_{i}$ evaluates to false.
    Since $\specaxeq{\objective}$ sound,
    and $\aword\not\speccongleq{\objective}\terr$
    (with the witnessing context 
    $\acontext{\contextvar}=\contextvar$)
    this proof suffices.

    For the base case, we have $\mathsf{P}_0=\mathsf{true}$ or
    $\mathsf{P}_0=\mathsf{false}$.
    After applying \Cref{axiom:eqmap}, we get 
    $P\specaxeq{\objective} \aword$ and $P\specaxeq{\objective}\terr$
    respectively.
    For the inductive case, we have 
    $\mathsf{P}_i=\mathsf{true}$, 
    $\mathsf{P}_i=\mathsf{false}$, 
    $\mathsf{P}_i
    =\mathsf{P}_{j_0}\wedge 
    \ldots\wedge\mathsf{P}_{j_k}$, or $\mathsf{P}_i=\mathsf{P}_{j_0}\vee 
    \ldots\vee\mathsf{P}_{j_k}$.
    The first two sub-cases are already handled in the base case.
    The latter two cases are dual, 
    so we only handle the $\wedge$-case.
    Per \Cref{axiom:eqmap}, we have 
    $P_i\specaxeq{\objective}\bigAchoiceOf{1}
    \set{P_{j_l}\mid l\leq k}$.
    The variable $\mathsf{P}_{i}$ evaluates to true, if and only if 
    $\mathsf{P}_{j_l}$ evaluates to true for all $l\leq k$.
    Let all variables evaluate to true.
    Induction hypothesis delivers us 
    $\mathsf{P}_{j_l}\specaxeq{\objective}\aword$ 
    and thus 
    $P_i\specaxeq{\objective}\bigAchoiceOf{1}\set{\aword}$.
    The rule \Cref{axiom:single} tells us 
    $\bigAchoiceOf{1}\set{\aword}\specaxeq{\objective}\aword$.
    W.l.o.g.\ let a variable $P_{j_0}$ evaluate to false.
    then induction hypothesis delivers us 
    $P_{j_0}\specaxeq{\objective}\terr$ and 
    $P_{j_l}\specaxeq{\objective}\aword$ or 
    $P_{j_l}\specaxeq{\objective}\terr$ for all 
    $0 < l\leq k$.
    We have $P_i\specaxeq{\objective}\bigAchoiceOf{1}\set{\terr}$
    or $P_i\specaxeq{\objective}\bigAchoiceOf{1}\set{\terr, \aword}$.
    In the former case we have $P_i\specaxeq{\objective}\terr$
    per \Cref{axiom:single}.
    In the latter case, the lattice axioms tell us that 
    $\bigAchoiceOf{1}\set{\aword, \terr}\specaxeq{\objective}\terr$,
    and thus $P_i\specaxeq{\objective}\terr$.
\end{proof}

\section{Upper Bound: Compact Terms Representation and Characteristic Terms}
\label{Appendix:Upperbound}
\subsection{Compact Term Representation}\label{Section:CompactEncoding}
Recall that the idea behind our compact term representation is to translate syntactic congruence classes $\classof{\aletter} : \states \to \states$ into angelic choices over state changes.  
With a representative system for the syntactic congruence, we can assume the source alphabet is $\Sigma$ rather than $\analphbot^*/\lreq{\objective}$. 
The translation of term~$\aterm$ over~$\analph$ 
is the term $\arrof{\aterm}$ over the alphabet $\arralph$ defined by\vspace{-0.3cm}
$$\begin{aligned}
    \arrof{\aletter}&=\bigEchoiceOf{1}{}{\setcond{(\astate, \astatep)}
        {\transitions(\astate,\aletter) = \astatep}} 
        &
        \arrof{\terr}&=\terr
        \\
        \arrof{\bigPchoiceOf{\anurg}{}{\atermset}}&=\bigPchoiceOf{\anurg}
        {}{\setcond{\arrof{\aterm}}{\aterm\in\atermset}}
        & 
        \arrof{\tskip} &= \tskip
        \\
        \arrof{\aterm\appl\atermp}&=\arrof{\aterm}\appl\arrof{\atermp}
        &
        \arrof{\nonterminal}&=\nonterminal \,,
\end{aligned}
$$
and we also translate the equations, $\arreqmapof{\nonterminal} =\arrof{\eqmapof{\nonterminal}}$.
The new alphabet calls for a translation of the objective. 
%
The DFA $\arrof{\objective}$ is obtained from $\objective$ by 
modifying the transitions and adding a failure state 
$\bot$.
We let $\arrof{\transitions}(\astate, (\astatep, \astatepp)) = \astatepp$
if $\astate=\astatep$ and 
$\arrof{\transitions}(\astate, (\astatep, \astatepp))=\bot$ 
if $\astate\neq\astatep$.
The accepting and initial states remain the same.
The translation is faithful when it comes to our notion of observable behavior. 
\begin{lemma}\label{Lemma:DFASimulation}
$\winsof{\aterm}{\objective}$ if and only if $\winsof{\arrof{\aterm}}{\arrof{\objective}}$. 
\end{lemma}

By combining \Cref{Lemma:DFASimulation} and \Cref{Corollary:Algorithm}, we can decide the specialized contextual preorder $\aterm\speccongleq{\objective}\atermp$ by checking whether $\winsof{\acontext{\arrof{\aterm}}}{\arrof{\objective}}$
implies $\winsof{\acontext{\arrof{\atermp}}}{\arrof{\objective}}$ 
for all contexts  
$\arrof{\atermppp}\appl\contextvar\appl\arrof{\atermpppp}$
with $\atermppp\in\specgnfnof{\maxurg-1}{\objective}$ 
and $\atermpppp\in\specgnfnof{\maxurg}{\objective}$.
The problem with this algorithm is that, for complexity reasons, we cannot work with an explicit translation of terms in specialized normal form. 
To overcome the problem, an attempt would be to generalize the above set of contexts to all 
${\atermppp\appl\contextvar\appl\atermpppp}$ with $\atermppp\in\specgnfnof{\maxurg-1}{\arrof{\objective}}$ and $\atermpppp\in\specgnfnof{\maxurg}{\arrof{\objective}}$. 
However, the set~$\specgnfnof{\maxurg}{\arrof{\objective}}$ contains terms over the new alphabet $\states^2$ that do not result from a translation of an $\specgnfnof{\maxurg}{\objective}$ term.
Unfortunately, these extra contexts may, incorrectly so, disprove the specialized contextual preorder of interest. 
%

Our solution is to come up with a direct construction for the image of $\specgnfnof{\maxurg}{\objective}$ under $\trspecnorm\circ\arr$, the translation followed by a normalization.
This is the appropriate subset of $\specgnfnof{\maxurg}{\arrof{\objective}}$ over which we should form contexts. 
The idea behind the construction is to explicitly translate the urgency~$1$ terms in specialized normal form, and build up the higher orders in the standard way. 
We define the \emph{translated $\objective$-specialized normal form} terms (with $\anurg>1$) by
\begin{align*}
    \plgnfn{0} = \specgnfnof{0}{\arrof{\objective}} &\qquad \plgnfn{1} = \trspecnormof{\arrof{\specgnfn{1}}}\\
    \plagnfn{\anurg} =&\ \setcond
        {\bigAchoiceOf{\anurg}{}{\atermset}}
        {\emptyset\neq\atermset\subseteq\plgnfn{\anurg-1}}\\
    \plgnfn{\anurg} =&\ \setcond{\bigEchoiceOf{\anurg}\atermset}
        {\emptyset\neq\atermset\subseteq\plagnfn{\anurg}}\ .
\end{align*}
The set of contexts we should iterate over is thus
\begin{align*}
\arrcontexts\ =\ \setcond{\atermpp\appl\contextvar\appl\atermppp}{\atermpp\in \plgnfn{\maxurg-1}\text{ and }\atermppp\in\plgnfn{\maxurg}}\ . 
\end{align*}
The argumentation leads to the following algorithm for checking the specialized contextual preorder.
\begin{proposition}\label{Proposition:PlausibleNormalContexts}
    Let $\objective$ be regular and $\aterm,\atermp$ finitary.  
    Then $\aterm\speccongleq{\objective}\atermp$ iff
    for all $\acontext{\contextvar}\in\arrcontexts$, we have that
    $\winsof{\acontext{\arrof{\aterm}}}{\arrof{\objective}}$ implies 
    $\winsof{\acontext{\arrof{\atermp}}}{\arrof{\objective}}$.
\end{proposition}

The benefit over \Cref{Corollary:Algorithm} is that the translated normal form terms are in $\arrof{\objective}$-specialized normal form, $\plgnfn{\maxurg}\subseteq \specgnfnof{\maxurg}{\arrof{\objective}}$, so we inherit the bound.
\begin{lemma}\label{Lemma:TranslatedNormalFormTerms}
$\sizeof{\plgnfn{\maxurg}}\leq \repexpof{2\maxurg}{\bigoof{\sizeof{\states}^2}}$. 
\end{lemma}
\subsection{Characteristic Terms}
With \Cref{Proposition:PlausibleNormalContexts}, we need to iterate over $2\maxurg$-exponentially many contexts. 
We now eliminate another exponent by factorizing the contexts with the help of characteristic terms.
Recall that term $\aterm$ is characteristic for context $\acontext{\contextvar}$ wrt. $\arrof{\objective}$, if its $\specaxleq{\arrof{\objective}}$-upward closure is the solution space of the context: for all $\atermp$ we have $\aterm \specaxleq{\arrof{\objective}} \atermp$ 
    if and only if $\winsof{\acontext{\atermp}}{\arrof{\objective}}$. 

For the contexts $\acontext{\contextvar} \in \arrcontexts$ we just defined, giving Adam a choice over the solution space yields a characteristic term:
\begin{align}
\charof{\acontext{\contextvar}}\ =\ 
    \bigAchoiceOf{\maxurg}{}{\setcond
        {\aterm\in\specgnfnof{\maxurg-1}{\arrof{\objective}}}
        {\winsof{\acontext{\aterm}}{\arrof{\objective}}}}\ . \label{Equation:CharTerms}
\end{align}
To see that the term is characteristic indeed, we rely on the domination preorder introduced in the completeness proof. 
\begin{lemma}\label{Lemma:CharTerm}
Term $\charof{\acontext{\contextvar}}$ is characteristic for $\acontext{\contextvar} \in \arrcontexts$ wrt. $\arrof{\objective}$ and can be computed in time~$\expof{2\maxurg-1}{\bigoof{\sizeof{\states}^2}}$.
\end{lemma}
\begin{proof}
To prove that $\charof{\acontext{\contextvar}}$ is characteristic, consider term $\aterm=
    \bigEchoiceOf{\maxurg}_{i\in I}{
        \bigAchoiceOf{\maxurg}{}{\atermset_{i}}
    }$ in $\arrof{\objective}$-specialized normal form. 
%
    %
    As $\aterm$ is immediate for $\acontext{\contextvar}$, we get $\winsof{\acontext{\aterm}}{\arrof{\objective}}$ 
    if and only if 
    there is an index $i\in I$ so that for all $\atermp\in\atermset_{i}$ we have
    $\winsof{\acontext{\atermp}}{\arrof{\objective}}$. 
 This can be shown to be equivalent to the domination preorder
    $\bigEchoiceOf{\maxurg}
    \bigAchoiceOf{\maxurg}\setcond{\aterm\in
    \specgnfnof{\maxurg-1}{\arrof{\objective}}}{ 
    \winsof{\acontext{\aterm}}{\arrof{\objective}}}
    \discleqof{\arrof{\objective}}
    \bigEchoiceOf{\maxurg}_{i\in I}\bigAchoiceOf{\maxurg}\atermset_{i}$. 
For the terms at hand, this domination preorder is equivalent to  
    $\bigAchoiceOf{\maxurg}\setcond{\aterm\in
    \specgnfnof{\maxurg-1}{\arrof{\objective}}}{ 
    \winsof{\acontext{\aterm}}{\arrof{\objective}}}
    \specaxleq{\arrof{\objective}}
    \bigEchoiceOf{\maxurg}_{i\in I}\bigAchoiceOf{\maxurg}\atermset_{i}$, 
    even if the objective is not right-separating.
    
    To compute the characteristic term, we have to check, for every term $\aterm\in\specgnfnof{\maxurg-1}{\arrof{\objective}}$, whether $\winsof{\acontext{\aterm}}{\arrof{\objective}}$ holds. 
    Such a check requires a normalization of $\acontext{\aterm}$, followed by a polynomial-time evaluation of the resulting term.  
    The normalization takes time $\sizeof{\acontext{\aterm}} \cdot \expof{2\maxurg-1}{\bigoof{\sizeof{\states}^2}}$,~\Cref{Lemma:EffNormalForm}. 
    The dominating factor in $\sizeof{\acontext{\aterm}}$ is the size of the $\plgnfn{\maxurg}$ term in the context, which is bounded by $\expof{2\maxurg-1}{\bigoof{\sizeof{\states}^2}}$. 
    There are $\expof{2\maxurg-2}{\bigoof{\sizeof{\states}^2}}$ terms $\aterm$ we have to go through.
    The overall runtime is thus bounded by $\expof{2\maxurg-1}{\bigoof{\sizeof{\states}^2}}$. 
\end{proof}

Let $\charof{\arrcontexts} = \setcond{\charof{\acontext{\contextvar}}}{\acontext{\contextvar} \in \arrcontexts}$ denote the set of characteristic terms. 
As these terms belong to $\specagnfnof{\maxurg}{\arrof{\objective}}$, we inherit the following bound.
\begin{lemma}\label{Lemma:NumberCharTerms}
$\sizeof{\charof{\arrcontexts}}\leq \expof{2\maxurg-1}{\bigoof{\sizeof{\objective}^2}}$. 
\end{lemma}

Compared to \Cref{Lemma:TranslatedNormalFormTerms}, there are exponentially fewer characteristic terms than contexts.  
To decide $\aterm\speccongleq{\objective}\atermp$, we thus intend to iterate over all
$\atermppppp \in \charof{\arrcontexts}$ and check 
whether 
$\atermppppp \speccongleq{\arrof{\objective}} \aterm$ 
implies 
$\atermppppp \speccongleq{\arrof{\objective}} \atermp$.
We will use the domination preorder for these checks.
\begin{lemma}\label{Lemma:ComparisonTime}
    Given $\aterm, \atermp \in \specgnfnof{\maxurg}{\arrof{\objective}}$, we can decide $\aterm \discleqof{\arrof{\objective}} \atermp$ in time $\sizeof{\aterm} \cdot \sizeof{\atermp}$.
\end{lemma}

There is a last obstacle: we do not know the characteristic terms, like we did not know the translated normal form terms above. 
Going through all contexts and determining the characteristic terms is prohibitively expensive. 
Generalizing from $\charof{\arrcontexts}$ to $\specagnfnof{\maxurg}{\arrof{\objective}}$ is incorrect. 
%
The way out is to give a direct construction of the characteristic terms.

The key insight is that the characteristic terms satisfy the following equation, where we have $\atermppp\in\specgnfnof{\maxurg-1}{\arrof{\objective}}$,  $\atermsetpppp\subseteq \specagnfnof{\maxurg}{\arrof{\objective}}$, and $\acontext{\contextvar}=\atermppp\appl\contextvar\appl\bigEchoiceOf{\maxurg}{}{\atermsetpppp}$:
%
\begin{align}
    \charof{\acontext{\contextvar}}
    \ \specaxeq{\arrof{\objective}}\ 
    \bigAchoiceOf{\maxurg}{}{
        \setcond{\charof{\atermppp\appl\contextvar\appl\atermpppp}}
        {\atermpppp\in \atermsetpppp}}\ . 
    \label{Equation:CharTermRecursive}
\end{align}
The equation follows from \Cref{Equation:CharTerms}, \Cref{Appendix:OptimalUB}.   
%
%
%
%
The impact of \Cref{Equation:CharTermRecursive} may not be immediate: we still have to make sure to construct the characteristic term for every set $\atermsetpppp\subseteq \specagnfnof{\maxurg}{\arrof{\objective}}$.  
What the equation does is to give us an inductive formulation of the characteristic terms which allows us to compute the set of all characteristic terms in a fixed point.
We first construct the characteristic terms for singleton sets $\sizeof{\atermsetpppp}=1$. 
Then we conjoin the characteristic terms as prescribed by \Cref{Equation:CharTermRecursive} to obtain the characteristic terms for sets of size $\sizeof{\atermsetpppp}\leq 2$. 
We repeat the latter conjunction until we reach a fixed point. 
Throughout the process, we work up to~$\specaxeq{\arrof{\objective}}$. 
With \Cref{Lemma:NumberCharTerms}, the sets we compute with have size at most $\expof{2\maxurg-1}{\bigoof{\sizeof{\objective}^2}}$. 
Moreover, we are guaranteed to reach the fixed point after at most $\expof{2\maxurg-1}{\bigoof{\sizeof{\objective}^2}}$ steps.
To state the correctness, define for $\atermset\subseteq\specgnfnof{\maxurg-1}{\arrof{\objective}}$ 
    and $\atermsetp\subseteq\specagnfnof{\maxurg}{\arrof{\objective}}$: 
\begin{align*}
    \allchartermsof{\atermset}{\atermsetp}=
    \bigcup_{{\atermppp\in\atermset, \atermsetpppp\subseteq\atermsetp}}
    \set{
    \charof{\atermppp\appl\contextvar\appl\bigEchoiceOf{\maxurg}{}{\atermsetpppp}
    }}\ .
\end{align*}
\begin{lemma}\label{Lemma:CharTermsAtScale}
    $\charof{\arrcontexts} = \allchartermsof{\plgnfn{\maxurg-1}}{\plagnfn{\maxurg}}$. 
    The set can be computed in time 
    $\expof{2\maxurg-1}{\bigoof{\sizeof{\objective}^{2}}}$.
\end{lemma}

The following proposition yields the overall algorithm. 
\begin{proposition}\label{Proposition:Overall}
    Let $\objective$ be regular and $\aterm,\atermp$ finitary.  
    Then $\aterm\speccongleq{\objective}\atermp$ iff
    for all $\atermppppp\in \allchartermsof{\plgnfn{\maxurg-1}}{\plagnfn{\maxurg}}$, we have 
    $\atermppppp \discleqof{\arrof{\objective}} \trspecnormof{\aterm}$ implies 
    $\atermppppp \discleqof{\arrof{\objective}} \trspecnormof{\atermp}$.
\end{proposition}
The domination preorder is sound for checking the axiomatic preorder even for objectives that fail to be right-separating because we have characteristic terms on the left. 
The time for computing the characteristic terms is given in \Cref{Lemma:CharTermsAtScale}. 
The normalization is \Cref{Lemma:EffNormalForm}, and we make use of the fact that the syntactic congruence of $\arrof{\objective}$ has size quadratic in~$\sizeof{\states}$. 
By \Cref{Lemma:ComparisonTime}, the comparison takes quadratic time.  
This concludes the proof of \Cref{Proposition:DecideUpperBoundEff}.

\section{Upper Bound: More Details on Compact Term Representation and Characteristic Terms}\label{Appendix:OptimalUB}

In this section, we handle the omitted proofs from 
\Cref{Section:UpperBound}.
Namely, we prove \Cref{Proposition:PlausibleNormalContexts},
\Cref{Lemma:DFASimulation}, the special case $\maxurg=1$ 
in \Cref{Lemma:CharTermsAtScale} and we justify 
\Cref{Equation:CharTermRecursive}.

\textbf{Reasoning for \Cref{Equation:CharTermRecursive}:}
To see \Cref{Equation:CharTermRecursive}, note that by \Cref{Equation:CharTerms} the characteristic term $\charof{\acontext{\contextvar}}$ is an urgency-$\maxurg$ choice owned by Adam over the set of $\specgnfnof{\maxurg-1}{\arrof{\objective}}$ solutions of $\acontext{\contextvar}$. 
This set of solutions is the union of the $\specgnfnof{\maxurg-1}{\arrof{\objective}}$ solutions for $\atermppp\appl\contextvar\appl\atermpppp$ with $\atermpppp\in\atermsetpppp$.
The reason is that the choice over $\atermsetpppp$ in the context has a higher urgency than the inserted term. 
We can thus stratify Adam's choice into a choice over $\atermpppp\in\atermsetpppp$ followed by a choice of the $\specgnfnof{\maxurg-1}{\arrof{\objective}}$ solutions for $\atermppp\appl\contextvar\appl\atermpppp$. 
Again by \Cref{Equation:CharTerms}, this is precisely the right-hand side of \Cref{Equation:CharTermRecursive}.

\textbf{Proofs of \Cref{Lemma:DFASimulation} and 
\Cref{Proposition:PlausibleNormalContexts}:}
Proofs of these statements
require us to observe the inner workings of terms.
To do so cleanly,
we extend the relation $\sigeq$ to terms.
The relation $\sigeq$ on 
terms is the smallest equivalence relation that 
contains $\sigeq$ on $\wordterms$, and the equalities
$\aterm\appl(\atermp\appl\atermpp)\sigeq
(\aterm\appl\atermp)\appl\atermpp$,
$\aterm\appl\tskip\sigeq\aterm$, 
$\tskip\appl\aterm\sigeq\aterm$,
$\aterm\appl\terr\appl\atermp\sigeq\terr$
for all $\aterm, \atermp, \atermpp\in\terms$.
Note that this is different than \Cref{axiom:monoid}, 
which allows us to apply $\sigeq$ to words enclosed by
arbitrary contexts.
We also define $\termhead:\terms\to\wordterms$ 
and $\termbody: \terms\to\terms$.
We let $\termheadof{\aterm}=\aword\in\wordterms$ 
be the concatenation of outermost terminals (including any $\tskip$ and $\terr$)
that appear in $\aterm$ before the leftmost outermost action.
If there are no such terms, $\termheadof{\aterm}=\tskip$.
We let $\termbodyof{\aterm}=\atermp$ be the 
concatenation of the remaining outermost actions and commands.
If no outmost action exists in the term, we let 
$\termbodyof{\aterm}=\tskip$.
Note that $\aterm\sigeq\termheadof{\aterm}\appl\termbodyof{\aterm}$
for all $\aterm\in\terms$.
Finally, we define a set of quasi-runs $\quasirunsof{\aword}$ 
for all $\aword\in\analph^{*}\cup\set{\terr}$.
We employ the notation $\quasirunsof{\aword}$ for a word term 
$\aword\in\wordterms$, as a shorthand for $\quasirunsof{\awordp}$ 
where $\awordp$ is the monoid element from 
$\analph^{*}\cup\set{\terr}$.
We let $\quasirunsof{\tskip}=\set{\tskip}$ and
$\quasirunsof{\terr}=\set{\terr}$.
For a word $\aword=\aletter_{0}\ldots\aletter_{k-1}\in\analph^{+}$, 
the set 
$\quasirunsof{\aword}$ contains all terms
of the form 
$(\astateinit, \astate_1)\appl(\astate_1, \astate_2)\ldots(\astate_{n-1}, \astate_{n})
    \appl\arrof{\aletter_{n}}\ldots\arrof{\aletter_{k-1}}$, 
where $\astateinit\overset{\aletter_{0}}{\to}\astate_1
    \ldots \astate_{n-1}\overset{\aletter_{n-1}}{\to}\astate_n$
is a run in the DFA for $\objective$.
So an element of $\quasirunsof{\aword}$ runs the DFA on the prefix of $\aword$ up to $n$ 
and has undetermined transitions in form of the terms $\arrof{\aletter_{i}}$ for the remainder of $\aword$.

We find it useful to prove a stronger version of 
\Cref{Lemma:DFASimulation}.

\begin{lemma}\label{Lemma:DFASimulationExtended}
    Let $\aterm\in\terms$, $\arun
    \in \quasirunsof{\termheadof{\aterm}}$
    and $\objective\subseteq\analph^{*}$.
    Then 
    for all $\atermp\sigeq\arun
    \appl\arrof{\termbodyof{\aterm}}$,
    $\winsof{\aterm}{\objective}$ if and only if 
    $\winsof{\atermp}{\arrof{\objective}}$.
\end{lemma}

Assuming \Cref{Lemma:DFASimulationExtended} we show 
\Cref{Proposition:PlausibleNormalContexts}.

\begin{proof}[Proof of \Cref{Proposition:PlausibleNormalContexts}]
    Let $\objective\subseteq\analph^{*}$.
    Per \Cref{Corollary:Algorithm}, we know that 
    $\aterm\speccongleq{\objective}\atermp$ holds if and only if
    $\winsof{\atermppp\appl\aterm\appl\atermpppp}{\objective}$
    implies
    $\winsof{\atermppp\appl\atermp\appl\atermpppp}{\objective}$
    for all $\atermppp\in\specgnfn{\maxurg-1}$ and 
    $\atermpppp\in\specgnfn{\maxurg}$.
    We apply \Cref{Lemma:DFASimulation} to see that 
    this is equivalent to the statement
    $\winsof{\arrof{\atermppp\appl\aterm\appl\atermpppp}}
    {\arrof{\objective}}$
    implies
    $\winsof{\arrof{\atermppp\appl\atermp\appl\atermpppp}}
    {\arrof{\objective}}$
    for all $\atermppp\in\specgnfn{\maxurg-1}$ and 
    $\atermpppp\in\specgnfn{\maxurg}$.
    Per definition, 
    we have $\arrof{\atermppp\appl\aterm\appl\atermpppp}=
    \arrof{\atermppp}\appl\arrof{\aterm}\appl\arrof{\atermpppp}$ 
    and 
    $\arrof{\atermppp\appl\atermp\appl\atermpppp}=
    \arrof{\atermppp}\appl\arrof{\atermp}\appl\arrof{\atermpppp}$.
    For the moment,
    assume
    $\plgnfn{\anurg}=
    \set{\arrof{\atermpp}\mid \atermpp\in\specgnfn{\anurg}}$
    (up to $\specaxeq{\arrof{\objective}}$)
    for all $\anurg\geq 1$ without proof.
    For $\maxurg>1$,
    this makes the previous statement equivalent with the desired
    $\winsof{\atermppp\appl\arrof{\aterm}\appl\atermpppp}{\objective}$
    implies 
    $\winsof{\atermppp\appl\arrof{\aterm}\appl\atermpppp}{\objective}$
    for all $\atermppp\in\plgnfn{\maxurg-1}$ and 
    $\atermpppp\in\plgnfn{\maxurg}$.
    Now let $\maxurg=1$.
    We show that for all $\atermppp\in\specgnfn{0}$ and 
    $\atermpppp\in\specgnfn{1}$, there are 
    $\atermppp'\in \plgnfn{0}$ and $\atermpppp'\in\plgnfn{1}$ 
    where $\winsof{\atermppp\appl\atermpp\appl\atermpppp}{\objective}$
    if and only if $\winsof{\atermppp'\appl\arrof{\atermpp}\appl\atermpppp'}
    {\arrof{\objective}}$
    for all $\atermpp\in\specgnfn{1}$.
    We can also conversely find $\atermppp\in\specgnfn{0}$ 
    and $\atermpppp\in\specgnfn{1}$ for all 
    $\atermppp'\in\plgnfn{0}$ and $\atermpppp'\in\plgnfn{1}$
    with the same property.
    Then, the statement
    $\winsof{\atermppp\appl\aterm\appl\atermpppp}{\objective}$ 
    implies $\winsof{\atermppp\appl\atermp\appl\atermpppp}{\objective}$ 
    for all $\atermppp\in\specgnfn{0}$, $\atermpppp\in\specgnfn{1}$ 
    is equivalent to the statement
    $\winsof{\atermppp'\appl\arrof{\aterm}\appl\atermpppp'}
    {\arrof{\objective}}$ 
    implies $\winsof{\atermppp'\appl\arrof{\atermp}\appl\atermpppp'}
    {\arrof{\objective}}$ for all $\atermppp'\in\plgnfn{0}$,
    $\atermpppp'\in\plgnfn{1}$.
    Note that, we have 
    $\termheadof{\atermppp\appl\aterm\appl\atermpppp}=\atermppp$
    and 
    $\termbodyof{\atermppp\appl\aterm\appl\atermpppp}=\aterm\appl\atermpppp$.
    Per \Cref{Lemma:DFASimulationExtended}, we know that 
    for all $\arun_{\atermppp}\in \quasirunsof{\atermppp}$,
    $\winsof{\atermppp\appl\aterm\appl\atermpppp}{\objective}$ 
    if and only if 
    $\winsof{\arun_{\atermppp}\appl\arrof{\aterm\appl\atermpppp}}{\objective}$.
    We can let 
    $\atermppp' = \arun_{\atermppp}=(\astateinit, \astate_1)(\astate_1 , \astate_2)\ldots(\astate_{n-1}, \astate_{n})
    \in\plgnfn{0}\cap \quasirunsof{\atermppp}$
    where the DFA runs on $\atermppp$ from $\astateinit$ to $\astate_{n}$ and $\atermpppp' = \arrof{\atermpppp}$.
    
    Finally, we show our assumption
    $\plgnfn{\anurg}=
    \set{\arrof{\atermpp}\mid \atermpp\in\specgnfn{\anurg}}$
    (up to $\specaxeq{\arrof{\objective}}$)
    for all $\anurg\geq 1$.
    The inclusion
    $\plgnfn{\anurg}\subseteq
    \set{\arrof{\atermpp}\mid \atermpp\in\specgnfn{\anurg}}$
    follows from the fact that 
    $\arrof{\bigPchoiceOf{\anurg}\atermset}
    =\bigPchoiceOf{\anurg}\set{\arrof{\aterm}\mid 
    \aterm\in\atermset}$ along with the definitions
    of $\plgnfn{\anurg}$ and $\plagnfn{\anurg}$.
    The inclusion 
    $\set{\arrof{\atermpp}\mid \atermpp\in\specgnfn{\anurg}}
    \subseteq\plgnfn{\anurg}$
    is proven by induction on $\anurg$.
    For the base case, 
    we have $\plgnfn{1}=\set{\normof{\arrof{\atermpp}}\mid 
    \atermpp\in\specgnfn{1}}$ per definition.
    For the inductive case, the cases 
    $\specgnfn{\anurg}$ and $\specagnfn{\anurg}$ are 
    analogous, so we only handle one.
    Letting $\bigEchoiceOf{\anurg}\atermset\in\specgnfn{\anurg}$,
    we get  
    $\normof{\arrof{\bigEchoiceOf{\anurg}\atermset}}
    =\normof{\bigEchoiceOf{\anurg}\set{\normof{\arrof{\atermpp}}}\mid 
    \atermpp\in\atermset}$.
    Induction hypothesis tells 
    us that $\normof{\arrof{\atermpp}}\in\plagnfn{\anurg}$ 
    for all $\atermpp\in\atermset$, so 
    $\bigEchoiceOf{\anurg}\set{\normof{\arrof{\atermpp}}}\in\plgnfn{\anurg}$.
\end{proof}

The proof of \Cref{Lemma:DFASimulationExtended}
is more involved and relies on the following \Cref{Lemma:TermAssoc,Lemma:ArrSuccCommute}, 
which we prove first.

\begin{lemma}\label{Lemma:TermAssoc}
    Let $\aterm,\atermp\in\terms$ with
    $\aterm\sigeq\atermp$.
    Then $\ownof{\aterm}=\ownof{\atermp}$ and 
    Eve wins from $\aterm$ in $\anordinalp$ 
    moves if and only if Eve wins from $\atermp$ in 
    $\anordinalp$ moves.
\end{lemma}

\begin{proof}[Proof Sketch]
    The definition of ownership 
    already implies $\ownof{\aterm}=\ownof{\atermp}$ for 
    all $\aterm,\atermp\in\terms$ with 
    $\aterm\sigeq\atermp$.
    The remainder of the 
    statement is proven by transfinite induction on $\anordinalp$.
    The base case follows from the monoid evaluation of 
    word terms.
    We sketch out the inductive case.
    W.l.o.g.\ we can only handle one direction of the implication.
    Let $\aterm\sigeq\atermp$ and 
    let Eve win from $\aterm$ in $\anordinalp$ moves.
    In both terms, the same $i$-th concatenation operand 
    (ignoring the bracketing) will be leading.
    Then, the successor sets are equal up to $\sigeq$.
    If $\ownof{\aterm}=\ownof{\atermp}=Eve$,
    then Eve wins from at least one $\aterm'\in\successorsof{\aterm}$
    in $\anordinalpp<\anordinalp$ moves.
    Here, we can apply the induction hypothesis to lift the 
    strategy to $\atermp$.
    The case $\ownof{\aterm}=\ownof{\atermp}=Adam$ is dual.
\end{proof}

We cal $\aterm\in\terms$ 
\emph{headless} if $\termheadof{\aterm}=\tskip$. 
For headless terms, $\arrof{.}$ and $\successorsof{.}$
commute.

\begin{lemma}\label{Lemma:ArrSuccCommute}
    For headless $\aterm\in\terms$,
    $\successorsof{\arrof{\aterm}}=\arrof{\successorsof{\aterm}}$.
\end{lemma}

\begin{proof}
    The proof is by (transfinite) structural induction on 
    $\aterm\in\terms$.
    For the base case, we let $\aterm=\nonterminal\in\nonterminals$.
    This is the only base case, 
    because $\aterm$ is only headless if 
    $\aterm=\nonterminal\in\nonterminals$.
    Per definition, we have $\successorsof{\arrof{\nonterminal}}=
    \set{\arrof{\eqmapof{\nonterminal}}}
    =\arrof{\successorsof{\nonterminal}}$.
    
    The first inductive case is $\aterm=\bigPchoiceOf{\anurg}\atermset$.
    Per definition, $\arrof{\bigPchoiceOf{\anurg}\atermset}=
    \bigPchoiceOf{\anurg}\set{\arrof{\aterm}\mid \aterm\in\atermset}$.
    And thus,
    $\successorsof
    {\bigPchoiceOf{\anurg}\set{\arrof{\aterm}\mid \aterm\in\atermset}}
    =\set{\arrof{\aterm}\mid \aterm\in\atermset}
    =\arrof{\successorsof{\bigPchoiceOf{\anurg}\atermset}}$.

    The second inductive case is $\aterm=\atermp\appl\atermpp$.
    If $\atermp$ is not headless, then 
    $\aterm$ would not be headless.
    So we deduce that $\atermp$ must be headless.
    A headless term must contain an outermost choice or 
    a non-terminal, since terms can not be empty.
    So, $\urgencyof{\atermp}\geq 1$.
    We observe that for all 
    $\atermppp\in\terms$,
    $\urgencyof{\atermppp}>0$
    implies $\urgencyof{\arrof{\atermppp}}=\urgencyof{\atermppp}$ 
    and $\urgencyof{\atermppp}=0$ 
    implies $\urgencyof{\arrof{\atermppp}}\leq 1$.
    This is clear from the construction of 
    $\arrof{.}$:
    Only the urgencies of 
    urgency $0$ subterms change.
    These get replaced by urgency $1$ terms,
    unless they are $\terr$ or $\tskip$.
    
    The first case is $\makeleading{\atermp}\appl\atermpp$.
    So, we have $\urgencyof{\atermp}\geq\urgencyof{\atermpp}$.
    When $\urgencyof{\atermpp}=0$, then
    $\urgencyof{\arrof{\atermp}}\geq
    1 \geq\urgencyof{\arrof{\atermpp}}$.
    Otherwise for $\atermp$ to be leading, $\urgencyof{\atermp} \geq 1$ must hold and
    \[
        \urgencyof{\arrof{\atermp}}=\urgencyof{\atermp}\geq 
        \urgencyof{\atermpp}=\urgencyof{\arrof{\atermpp}}\,.
    \]
    In either case we have $\makeleading{\arrof{\atermp}}\appl\arrof{\atermpp}$.
    By I.H., 
    $\successorsof{\arrof{\atermp}}=\arrof{\successorsof{\atermp}}$.
    This results in
    $\successorsof{\arrof{\atermp\appl\atermpp}}
    =\successorsof{\arrof{\atermp}}\appl\arrof{\atermpp}
    =\arrof{\successorsof{\atermp}}\appl\arrof{\atermpp}
    =\arrof{\successorsof{\makeleading{\atermp}\appl\atermpp}}$.
    
    Now let $\atermp\appl\makeleading{\atermpp}$.
    Since the term $\atermp$ is headless, $\atermpp$ has urgency $\urgencyof{\atermpp}>\urgencyof{\atermp}\geq 1$.
    Let $\atermpp=\atermppp\appl
    \makeleading{\atermppppp}\appl\atermpppp$
    for some action $\atermppppp$.
    If we are strict,
    we also need to handle the 
    cases where $\atermpp$ 
    equals to $\atermppp\appl\makeleading{\atermppppp}$,
    $\makeleading{\atermppppp}\appl\atermpppp$, and 
    $\makeleading{\atermppppp}$.
    We omit them to avoid repetition.
    Since $\urgencyof{\atermppppp}=\urgencyof{\atermpp}>1$,
    the term $\atermppppp$ can only be a choice 
    $\bigPchoiceOf{\anurg}\atermset$ or 
    a non-terminal $\nonterminal$.
    Both of these terms are 
    headless, so we apply the induction hypothesis and obtain
    $\successorsof{\arrof{\atermppppp}}=
    \arrof{\successorsof{\atermppppp}}$.
    Because $\urgencyof{\atermppppp}>\urgencyof{\atermp\appl\atermppp}$,
    the leading subterm is
    $\arrof{\atermp}\appl\arrof{\atermppp}\appl
    \makeleading{\arrof{\atermppppp}}\appl\arrof{\atermpppp}$.
    Indeed, $\urgencyof{\arrof{\atermp\appl\atermppp}}
        \leq \max(1, \urgencyof{\atermp\appl\atermppp})$ 
    and $\urgencyof{\arrof{\atermp\appl\atermppp}}<\urgencyof{\arrof{\atermppppp}}$.
    Finally, we derive
    \begin{align*}
        \successorsof{\arrof{\atermp\appl\makeleading{\atermpp}}}
    &=\successorsof{\arrof{\atermp\appl\atermppp\appl
    \makeleading{\atermppppp}\appl\atermpppp}}\\
    &=\arrof{\atermp\appl\atermppp}\appl\successorsof{\arrof{\atermppppp}}
        \appl\arrof{\atermpppp}\\
    &=\arrof{\atermp\appl\atermppp}\appl\arrof{\successorsof{\atermppppp}}
        \appl\arrof{\atermpppp}\\
    &=\arrof{\successorsof{\atermp\appl\atermppp\appl\makeleading{\atermppppp}\appl\atermpppp}}\\
    &=\arrof{\successorsof{\atermp\appl\makeleading{\atermpp}}} 
    \end{align*}
\end{proof}

\begin{proof}[Proof of \Cref{Lemma:DFASimulationExtended}]
    We show both directions
    by an induction on the number of moves
    Eve needs to win.
    
    \emph{Forward Direction:}
    Let $\aterm\in\terms$ and $\arun \in \quasirunsof{\termheadof{\aterm}}$.
    For the base case, let Eve wins $\objective$ from $\aterm$ in $0$ moves.
    Then, $\termheadof{\aterm}=\aterm \in \wordterms$ and $\termbodyof{\aterm}=\tskip$.
    Because Eve wins, $\aterm$ can not contain $\terr$.
    Then, $\arun \in \quasirunsof{\aterm}$ must be of the form 
    $(\astateinit, \astate_1)\appl(\astate_1, \astate_2)\ldots(\astate_{n-1}, \astate_{n})
        \appl\arrof{\aletter_{n}}\ldots\arrof{\aletter_{k-1}}$.
    We know that the word $\aword= \aletter_0 \ldots \aletter_n \ldots \aletter_{k-1}$ that corresponds to 
    $\aterm \sigeq \aword$ has a run on the DFA for $\objective$.
    Per definition of $\quasirunsof{\aterm}$, 
    the DFA runs $\aword$ from 
    $\astateinit$ to $\astate_n$ in $n-1$ steps, 
    and Eve can choose the remaining transitions 
    to reach $\astate_k\in \finalstates$.
    
    For the inductive case, let Eve reach $\objective$ from $\aterm$ in
    $\anordinalp$ moves.
    Per \Cref{Lemma:TermAssoc}, this also holds from $\aword\appl\atermp$, 
    where $\aword=\termheadof{\aterm}$ and $\atermp=\termbodyof{\aterm}$.
    We show that $\winsof{\atermpp}{\arrof{\objective}}$
    for some $\arun_{\aword}\in \quasirunsof{\aword}$  and
    $\atermpp\sigeq\arun_{\aword}\appl\arrof{\atermp}$.
    Per 
    \Cref{Lemma:TermAssoc},
    showing this for one such $\atermpp$ suffices.
    We first observe that Eve has a strategy to reach
    $\arun_{\aword}'\appl\makeleading{\arrof{\atermp}}$
    where $\arun_{\aword}'\in \quasirunsof{\aword}$.
    %
    In case of $\urgencyof{\atermp} \geq 2$ we have $\arun_{\aword} = \arun_{\aword}'$ 
    and the leading subterm $\arun_{\aword}\appl\makeleading{\arrof{\atermp}}$ 
    due to $\urgencyof{\arun_{\aword}}\leq 1 < \urgencyof{\arrof{\atermp}}$.
    Otherwise, $\urgencyof{\arrof{\atermp}} = \urgencyof{\atermp}=1$ (it can't be 0 because $\atermp = \termbodyof{\aterm}$).
    If $\arun_{\aword} \in \quasirunsof{\aword} \cap \wordterms$, 
    we already have $\arun_{\aword}\appl\makeleading{\arrof{\atermp}}$.
    So let $\arun_{\aword} \in \quasirunsof{\aword} \setminus \wordterms$,
    i.e.\ $\urgencyof{\arun_{\aword}}=1$.
    Eve resolves each term of the form $\bigEchoiceOf{1}{\setcond{(\astatep, \astatepp)}
        {\astatepp\in\transitions(\astatep, \aletter)}}$ 
    for some
    $\aletter\in\Sigma$ and extend the determined part of the run 
    in $\quasirunsof{\aword}$.
    We know that $\aword\not\sigeq\terr$, because Eve wins 
    $\aword\appl\atermp$.
    So, per definition of $\quasirunsof{\aword}$, Eve can find 
    fitting transitions.
    Exhaustive application of this strategy 
    rewrites $\arun_{\aword}$ to 
    $\arun_{\aword}'\in \quasirunsof{\aword} \cap \wordterms$
    for the desired $\arun_{\aword}'\appl\makeleading{\arrof{\atermp}}$.
    
    We now show $\winsof{\arun_{\aword}'\appl\arrof{\atermp}}{\arrof{\objective}}$.
    Let $\ownof{\aword\appl\atermp}=Eve$.
    The case $\ownof{\aword\appl\atermp}=Adam$ is dual.
    Remember $\aterm = \aword\appl\atermp$ and $\aword = \termheadof{\aterm}$,
    so we have $\aword\appl\makeleading{\atermp}$.
    Then, there is a $\aword\appl\atermp' \in \successorsof{\aword\appl\atermp}$
    from which Eve wins in $\anordinalpp<\anordinalp$ moves.
    Since $\aword\appl\atermp'\sigeq
    \aword\appl\termheadof{\atermp'}\appl\termbodyof{\atermp'}$,
    \Cref{Lemma:TermAssoc} tells us that 
    Eve also wins from $\aword\appl\termheadof{\atermp'}
    \appl\termbodyof{\atermp'}$.
    We have $\arun_{\aword}'\appl\arrof{\termheadof{\atermp'}}\in 
    \quasirunsof{\aword\appl\termheadof{\atermp'}}
    =\quasirunsof{\termheadof{\aword\appl\atermp'}}$,
    so we can apply the induction hypothesis to see that 
    Eve wins from $\arun_{\aword}'\appl\arrof{\termheadof{\atermp'}}
    \appl\arrof{\termbodyof{\atermp'}}$ in $\anordinalpp$ turns.
    The fact 
    $\arun_{\aword}'\appl\arrof{\termheadof{\atermp'}}
    \appl\arrof{\termbodyof{\atermp'}}
    \sigeq 
    \arun_{\aword}'\appl\arrof{\atermp'}$ and 
    \Cref{Lemma:TermAssoc} imply 
    that $\winsof{\arun_{\aword}'\appl\arrof{\atermp'}}{\arrof{\objective}}$.
    Since $\atermp$ is headless, \Cref{Lemma:ArrSuccCommute} tells us
    $\successorsof{\arrof{\aterm}}=\arrof{\successorsof{\aterm}}$ 
    and thus $\arun_{\aword}'\appl\arrof{\atermp'}
        \in\successorsof{\arun_{\aword}'\appl\makeleading{\arrof{\atermp}}}$.
    Then, we also have
    $\winsof{\arun_{\aword}'\appl\arrof{\atermp}}{\arrof{\objective}}$.
    This concludes this direction of the proof.

    \emph{Backward Direction:}
    Let $\aterm\in\terms$, $\aword=\termheadof{\aterm}$,
    $\atermp=\termbodyof{\aterm}$ and 
    $\objective\subseteq\analph^{*}$.
    For the base case, let Eve win 
    $\arun_{\aword}\appl\arrof{\atermp}$ in 
    $0$ moves.
    Then, $\arrof{\atermp}$ is a command and this is only 
    possible if $\atermp=\tskip$.
    Since $\arun_{\aword}\in\arrof{\objective}$, the run
    $\arun_{\aword}$ is 
    accepting in the DFA for $\objective$.
    So, $\aword=\aterm\in\objective$.

    For the inductive case, 
    let Eve win $\arun_{\aword}\appl\arrof{\atermp}$ 
    in $\anordinalp$ moves
    for some $\arun_{\aword}\in \quasirunsof{\aword}$.
    The first case is $\makeleading{\arun_{\aword}}\appl\arrof{\atermp}$.
    Then, there is a term $\arrof{\aletter}$
    in $\arun_{\aword}$,
    so Eve can simply choose the corresponding 
    transition in the DFA to extend $\arun_{\aword}$ 
    to $\arun_{\aword}'\in \quasirunsof{\aword}$
    per definition of $\quasirunsof{\aword}$.
    Eve then wins from $\arun_{\aword}'\appl\arrof{\atermp}$ in 
    $\anordinalpp<\anordinalp$ moves.
    We can apply the induction hypothesis to see that 
    Eve wins from $\aword\appl\atermp$, and per \Cref{Lemma:TermAssoc} from $\aterm$.
    The second case is
    $\arun_{\aword}\appl\makeleading{\arrof{\atermp}}$.
    Let $\ownof{\arun_{\aword}\appl\arrof{\atermp}}=Eve$.
    Again, the case $\ownof{\arun_{\aword}\appl\arrof{\atermp}}=Adam$ is dual.
    Since $\atermp$ is headless, 
    we have $\successorsof{\arrof{\atermp}}=\arrof{\successorsof{\atermp}}$ 
    per \Cref{Lemma:ArrSuccCommute}.
    Then, Eve wins from some $\arun_{\aword}\appl\arrof{\atermp'}$ 
    where $\atermp'\in\successorsof{\atermp}$
    in $\anordinalpp<\anordinalp$ moves.
    Write $\arun_{\aword}\appl\arrof{\atermp'}\sigeq
    \arun_{\aword}\appl\arrof{\termheadof{\atermp'}}\appl
    \arrof{\termbodyof{\atermp'}}$.
    As in the previous direction, we have 
    $\arun_{\aword}\appl\arrof{\termheadof{\atermp'}}\in 
    \quasirunsof{\aword\appl\termheadof{\atermp'}}$.
    Per definition, $\termbodyof{\atermp'}$ is headless.
    We can apply the induction hypothesis to 
    get that Eve wins from $\aword\appl\termheadof{\atermp'}\appl 
    \termbodyof{\atermp'}\sigeq\aword\appl\atermp'\in
    \successorsof{\aword\appl\makeleading{\atermp}}$.
    So Eve wins from $\aword\appl\atermp\sigeq\aterm$ as well, 
    completing the proof.
\end{proof}

\textbf{Proof of \Cref{Lemma:CharTermsAtScale} for $\maxurg=1$:}~\\
\emph{Preliminary Facts:}
Before moving on to the computation
of $\chi(\arrcontexts)$ for $\maxurg=1$, we study 
$\specgnfnof{0}{\arrof{\objective}}$ more closely.
\begin{lemma}\label{Lemma:TrSyntacticCongruence}
    The syntactic monoid of $\arrof{\objective}$ is $\specgnfnof{0}{\arrof{\objective}}
    =\states^{2}\cup\set{\terr, \tskip}$.
    Furthermore, 
    the
    ${(\astate, \astatep)\appl(\astatep, \astatepp)}
    \specaxeq{\arrof{\objective}}
    (\astate,\astatepp)$
    and 
    ${(\astate, \astatep)\appl(\astateppp, \astatepp)}
    \specaxeq{\arrof{\objective}}
    \terr$ hold for all 
    $\astate, \astatep, \astatepp, \astateppp \in\states$ with $\astatep\neq\astateppp$.
\end{lemma}
\begin{proof}
    First we show that $\states^{2}\cup\set{\terr, \tskip}
        \subseteq \specgnfnof{0}{\arrof{\objective}}$, 
    i.e.\ that these elements are pairwise not equivalent.
    It is clear that for any $(\astate,\astatep)\in
    \states^{2}$,
    $
    \winsof{(\astateinit,\astate)\appl(\astate,\astatep)\appl(\astatep, 
    \afinalstate)}{\objective}$,
    so we get $(\astate,\astatep)\not\speccongleq{\arrof{\objective}}
    \terr$ and thus 
    $(\astate, \astatep)\not\specaxleq{\arrof{\objective}}\terr$
    per soundness.
    It is also clear to see that for any 
    $(\astate, \astatep) \neq (\astateppp, \astatepp)\in\states^{2}$,
    we have
    $\winsof{(\astateinit, \astate)\appl(\astate,\astatep)
    \appl(\astatep, \afinalstate)}{\arrof{\objective}}$ but
    $\winsof{(\astateinit, \astate)\appl(\astateppp,\astatepp)
    \appl(\astatep, \afinalstate)}{\arrof{\objective}}$
    fails to hold since $\astate \neq \astateppp$ or $\astatep \neq \astatepp$.
    $\tskip$ is the neutral element of the syntactic monoid 
    and unless $\sizeof{\states} = 1$ it is different from any $(\astate, \astatep) \in \states^2$.

    Now, we show $\specgnfnof{0}{\arrof{\objective}}\subseteq 
    \states^{2}\cup\set{\terr, \tskip}$.
    These are exactly the terminal symbols, so
    it will suffice to show that 
    the set $\states^{2}\cup\set{\terr}$ 
    is closed under concatenation
    up to $\specaxeq{\arrof{\objective}}$.
    Concatenations that involve $\terr$ or 
    $\tskip$ elements are reduced by \Cref{axiom:monoid}
    to an element from $\states^{2}\cup\set{\terr, \tskip}$.
    
    It remains to show the claimed congruences.
    For that, let $\astate,\astatep, \astatepp, \astateppp\in\states$ 
    with $\astatep \neq \astateppp$
    and let $\acontext{\contextvar} = \atermppp\appl\contextvar\appl\atermpppp$ 
    be a concatenative context with $\atermppp, \atermpppp \in \wordterms$.
    The first observation is that
    $\acontext{(\astate,\astatep)\appl(\astatepp, \astateppp)}$ 
    is losing for Eve, since she can not form a continuous run, 
    i.e.\ she can not reach $\arrof{\objective}$.
    Then, per \Cref{axiom:spec} we get 
    $(\astate,\astatep)\appl(\astateppp, \astatepp)\specaxeq{\arrof{\objective}}\terr$.
    An accepting, continuous run needs to have only three properties.
    Namely, it must start from $\astateinit$, it must not have 
    discontinuities, and it must end at some 
    $\afinalstate\in\finalstates$.
    But, $(\astate,\astatep)\appl(\astatep, \astatepp)$ 
    and $(\astate, \astatepp)$ are both continuous while 
    starting and ending at the same states.
    So we see that 
    $\acontext{(\astate, \astatep)\appl(\astatep,\astatepp)}$
    is an accepting run if and only if 
    $\acontext{(\astate, \astatepp)}$ is an accepting run.
    Thus, $(\astate,\astatep)\appl(\astatep, \astatepp)
    \specaxeq{\arrof{\objective}}
    (\astate, \astatepp)$.
\end{proof}

We define 
$\transofword{\aword}=\set{(\astate,\astatep)\in\states^{2}\mid 
\astate\overset{\awordp}{\to}\astatep, \aword\sigeq\awordp\in\analph^{+}}$
for all $\aword\in\wordterms$ without $\terr$ that are not $\tskip$.
\begin{lemma}\label{Lemma:TrDFARuns}
    $\arrof{\aword}\specaxeq{\arrof{\objective}}\bigEchoiceOf{1}{\transofword{\aword}}$.
\end{lemma}
\begin{proof}
    This is done by induction on $\sizeof{\aword}$, 
    the number of symbols contained in $\aword$.
    For the base case, we have ${\sizeof{\aword}=1}$.
    Then, $\arrof{\aword}=\bigEchoiceOf{1}{\set{(\astate,\astatep)
    \mid \astatep\in\transitions(\astate, \aword)}}$.
    For the inductive case, let $\aword\appl\awordp\in\wordterms$ 
    unequal to $\terr$ or $\tskip$ by $\sigeq$.
    Then, 
    \begin{align*}
        \arrof{\aword\appl\awordp}
        &=
        \arrof{\aword}\appl\arrof{\awordp}\\
        &\specaxeq{\arrof{\objective}}
        \bigEchoiceOf{1}{\transofword{\aword}}
        \appl \bigEchoiceOf{1}{\transofword{\awordp}}\\
        \intertext{Applying \labelcref{axiom:dist-left,axiom:dist-right}
        and flattening with \labelcref{axiom:lattice-assoc}:}
        &\specaxeq{\arrof{\objective}}
        \bigEchoiceOf{1}
        \set{(\astate,\astatep)\appl(\astateppp, \astatepp)\mid \\
        &\qquad\qquad\qquad
        (\astate,\astatep)\in\transofword{\aword}, 
        (\astateppp, \astatepp)\in\transofword{\awordp}}\\
        \intertext{Using 
        \Cref{Lemma:TrSyntacticCongruence} 
        and \labelcref{axiom:lattice-ord} to remove inconsistent runs:}
        &\specaxeq{\arrof{\objective}}
        \set{(\astate,\astatep)\appl(\astatep,\astateppp)\mid \\
        &\qquad\qquad\qquad
        (\astate, \astatep)\in\transofword{\aword}, 
        (\astatep, \astatepp)\in\transofword{\awordp}}\\
        &\specaxeq{\arrof{\objective}}
        \bigEchoiceOf{1}\transofword{\aword\appl\awordp}
    \end{align*}
\end{proof}

\subsection{Iterating the Characteristic Terms}
Recall that we have constructed $\plgnfn{1}$ to be the 
image of $\specgnfn{1}$ under $\trspecnorm\circ\arr$.
This poses a problem for $\maxurg=1$.
An Adam choice alone can not represent the image of 
a $\specagnfn{1}$ term under normalization.
Namely, for some 
$\bigEchoiceOf{1}_{i\in I}\bigAchoiceOf{1}\atermset_{i}\in
\specgnfnof{1}{\arrof{\objective}}$,
we might not have any $\aterm\in\terms$ with 
$\arrof{\aterm}\specaxeq{\arrof{\objective}}
\bigAchoiceOf{1}\atermset_{i}$.
So, if we naively apply the construction 
for $\maxurg=1$, we are forced to iterate over 
a substantial subset of $\specgnfnof{1}{\arrof{\objective}}$.
This results in a $\expof{2}{\bigoof{|\objective|^2}}$ time complexity.
For this reason,
we compute $\chi(\arrcontexts)$ directly,
by exploiting the Myhill-Nerode right-precongruence 
on the states of the DFA.\@
\begin{definition}
    For any $\astate,\astatep\in\states$, we have
    $\astate\leq_{N}\astatep$ if and only if for all 
    $\aword\in\Sigma^{*}$, 
    $\astate\overset{\aword}{\to}\afinalstate\in\finalstates$
    implies $\astatep\overset{\aword}{\to}\afinalstate'$
    for some $\afinalstate'\in\finalstates$.
\end{definition}

We extend this to pairs of states $(\astate, \astate'), 
(\astatep, \astatep')\in\states^{2}$.
We let $(\astate, \astate')\leq_{N}(\astatep, \astatep')$ 
if and only if $\astate=\astatep$ and 
$\astate'\leq_{N}\astatep'$.
We call a state pair $(\astate, \astatep)$ dead, if 
there is no run from $\astatep$ to a state in $\finalstates$
in the DFA for $\objective$.
For some $\acontext{\contextvar}\in\arrcontexts$,
we call $\mathcal{S}(\acontext{\contextvar})
=\set{\aterm\in\specgnfnof{0}{\arrof{\objective}}
\mid \winsof{\acontext{\aterm}}{\objective}}$ the solution space 
of $\acontext{\contextvar}$.
Towards computing $\chi(\arrcontexts)$, 
we show two important facts that expose the relationship
between the extended
$\leq_{N}$
and 
solution spaces of contexts.

First, we show that there is a 
context $\acontext{\contextvar}\in\arrcontexts$
for each not-dead $(\astate,\astatep)\in\states^{2}$,
where
$\mathcal{S}(\acontext{\contextvar})$ 
is the $\leq_{N}$-upward closure of $(\astate,\astatep)$.
\begin{lemma}\label{Lemma:SingleUpwardClosure}
    For all $(\astate,\astatep)\in\states^{2}$ that 
    is not dead, there is a 
    $\acontext{\contextvar}\in\arrcontexts$ with 
    $\mathcal{S}(\acontext{\contextvar})=
    \set{(\astate',\astatep')\in\states^2\mid 
    (\astate, \astatep)\leq_{N}(\astate', \astatep')}$.
\end{lemma}

\begin{proof}
    We claim that $\acontext{\contextvar}=
    (\astatep_0, \astate)\appl\contextvar\appl
    \bigAchoiceOf{1}_{\aword\in\atermset}\arrof{\aword}$ 
    where $\atermset=\set{\awordp\mid 
    \astatep\overset{\awordp}{\to}\astatep_{f}\in\finalstates}$.
    Note that $\atermset\neq\emptyset$, because 
    $(\astate, \astatep)$ is not dead.
    For $(\astateppp, \astatepppp)\in\states^{2}$ with 
    $\astate\neq\astatepp$,
    $\acontext{(\astateppp, \astatepppp)}$ is losing for Eve,
    since $(\astatep_0, \astate)\appl(\astateppp, \astatepppp)
    \specaxeq{\objective}\terr$ 
    per \Cref{Lemma:TrSyntacticCongruence}.
    Now assume $(\astate, \astateppp)\in\states^{2}$.
    We use \Cref{Lemma:TrDFARuns} and get 
    $\acontext{(\astate, \astateppp)}\specaxeq{\arrof{\objective}}
    (\astatep_0, \astate)\appl(\astate, \astateppp)\appl
    \bigAchoiceOf{1}_{\aword\in\atermset}\bigEchoiceOf{1}\reachof{\aword}$.
    Then, 
    Eve wins $\acontext{(\astate, \astateppp)}$ if and only if 
    Eve wins $(\astatep_{0}, \astate)\appl(\astate,\astateppp)\appl
    \bigEchoiceOf{1}\reachof{\aword}$ 
    for all $\aword\in\atermset$.
    This is equivalent to 
    $\aword$ 
    having a run from $\atermppp$ to some 
    $\astatep_{f}\in\finalstates$
    for all $\aword\in\atermset$,
    This is the definition of $\astatep\leq_{N}\astateppp$.
\end{proof}

Then, we show that for all $\acontext{\contextvar}\in\arrcontexts$,
the solution spaces are upward closed.

\begin{lemma}\label{Lemma:SolutionsUpwardClosed}
    Let $\acontext{\contextvar}\in\arrcontexts$
    and $(\astate, \astatep), (\astate', \astatep')\in\states$
    with $(\astate, \astatep)\leq_{N}(\astate', \astatep')$.
    Then,
    $(\astate, \astatep)\in\mathcal{S}(\acontext{\contextvar})$
    implies $(\astate', \astatep')\in\mathcal{S}(\acontext{\contextvar})$.
\end{lemma}

\begin{proof}[Proof of \Cref{Lemma:SolutionsUpwardClosed}]
    W.l.o.g.\ let $(\astate, \astatep), (\astate, \astatep')
    \in\states^{2}$ with 
    $(\astate, \astatep)\leq_{N}(\astate, \astatep')$
    and $\acontext{\contextvar}=(\astateppp, \astatepppp)\appl 
    \contextvar\appl\bigEchoiceOf{1}_{i\in I}\bigAchoiceOf{1}_{\aword\in\atermset_{i}}
    \bigEchoiceOf{1}\reachof{\aword}$.
    If $\astateppp\neq\astatep_0$ or 
    $\astatepppp\neq\astate$ no matter which
    $(\astate, \astatepp)$ is inserted to this context,
    Eve can never derive a continuous run.
    Now let $\acontext{\contextvar}=
    (\astatep_0, \astate)\appl 
    \contextvar\appl\bigEchoiceOf{1}_{i\in I}\bigAchoiceOf{1}_{\aword\in\atermset_{i}}
    \bigEchoiceOf{1}\reachof{\aword}$
    and assume that 
    Eve wins from 
    $\acontext{(\astate, \astatep)}$.
    Then, there is an $i\in I$ where for all $\aword\in\atermset$,
    Eve wins $(\astatep_0, \astate)\appl(\astate,\astatep)\appl
    \bigEchoiceOf{1}\reachof{\aword}$.
    So, there is an $i\in I$ where for all $\aword\in\atermset_{i}$,
    there is a run $\astatep\overset{\aword}{\to}\astatep_{f}$ 
    for some $\astatep_{f}\in\finalstates$.
    The latter part of the 
    statement and $(\astate, \astatep)\leq_{N}(\astate, \astatep')$
    implies that there also is a run 
    $\astatep'\overset{\aword}{\to}\astatep_{f}'$
    for some $\astatep_{f}'\in\finalstates$.
    So, from
    $\acontext{(\astate, \astatep')}$, 
    Eve plays the same $i\in I$.
    Let Adam play some $\aword\in\atermset_{i}$.
    The resulting term is $(\astatep_0, \astate)\appl 
    (\astate, \astatep')\appl\bigEchoiceOf{1}\reachof{\aword}$.
    So Eve chooses $(\astatep', \afinalstate')\in\reachof{\aword}$ 
    to reach 
    $(\astatep_0, \astate)\appl 
    (\astate, \astatep')\appl(\astatep', \afinalstate')$ and win.
\end{proof}

By applying \Cref{Equation:CharTermRecursive}
along with \Cref{Lemma:SingleUpwardClosure} and 
\Cref{Lemma:SolutionsUpwardClosed}, we see that
a set $X\subseteq\specgnfnof{0}{\arrof{\objective}}$ is 
a solution space if and only if it is an upward closure
that does not contain dead pairs.
We can iterate through all left-sides 
$\astate\in\states$ and right-sides $\states'\subseteq\states$ 
and build the upward closures, therefore solution spaces, in 
$\expof{1}{\bigoof{|\objective|^2}}$ time.
Recall that for $\acontext{\contextvar}\in\arrcontexts$,
we had $\chi(\acontext{\contextvar})=
\bigAchoiceOf{1}\mathcal{S}(\acontext{\contextvar})$
(\Cref{Equation:CharTerms}).
So, we can build 
$\set{\bigAchoiceOf{\anurg}\mathcal{S}(\acontext{\contextvar})
\mid \acontext{\contextvar}\in\arrcontexts}=
\set{\chi(\acontext{\contextvar})\mid 
\acontext{\contextvar}\in\arrcontexts}=\chi(\arrcontexts)$
in $\expof{1}{\bigoof{|\objective|^{2}}}$ time.

\section{Hyperproperties and Urgency}\label{Appendix:Hyper}

Hyperproperties emerged as a unifying approach to information flow and security properties, 
which cannot be stated as classical safety or liveness properties on a single trace.
A hyperproperty relates traces and is formulated over a set of traces 
rather than a single trace.
A novel approach are logics like HyperLTL~\cite{FRS15} to describe hyperproperties.
We will stay more general 
and define an $n$-trace hyperproperty 
to be any DFA $\adfa$ over ${(\analph^n)}^*$.
Note that we consider finite traces rather than infinite ones.
In particular, a hyperproperty takes about the relation of traces 
and does so in a highly synchronized manner:
A set of traces is accepted by $\adfa$ ensures that they all share the exact same length.
To model different traces of the same system 
so that the resulting observations are synchronized 
is a common problem with hyperproperties.
We will not tackle this issue here, 
but assume that the traces of a system $\akripke$ are partitioned into sets of same length.
That means that the set of traces $\tracesof{\akripke}\subseteq\analph^*$ is partitioned into sets of traces $\tracesofn{l}{\akripke}\subseteq\analph^l$ of length $l \in \nat$.
\begin{definition}\label{definition:hyperproperty}
    Let $\akripke$ be a system (Kripke structure) 
    with trace-set $\tracesof{\akripke} \subseteq \analph^*$.
    An $n$-trace hyperproperty $\adfa$ is satisfied by $\akripke$, $\akripke \models \adfa$, if there is $l \in \nat$ such that
    \[
        \exists \aword_1 \in \tracesofn{l}{\akripke}
        \forall \aword_2 \in \tracesofn{l}{\akripke}
        \ldots
        Q \aword_n \in \tracesofn{l}{\akripke}.\, 
        \prod_{i=1}^{n} \aword_i \in \langof{\adfa}
    \]
\end{definition}

The definition might be non-intuitive at first glance: 
Only for a single length $l$, the set $\tracesofn{l}{\akripke}$ has to satisfy the hyperproperty.
Further, every hyperproperty begins to quantify with an $\exists$ quantifier.
This makes it impossible to formulate hyperproperties from the $\forall\exists$-fragment of HyperLTL.\@
However, the ability to deterministically decide hyperproperties of the above shape 
is already sufficient to decide any hyperproperty,
including ones that start with $\forall$, 
for which we check its negation.

We show how to model check a hyperproperty 
by translating the system $\akripke$ 
and the hyperproperty $\adfa$ 
into a term $\aterm$ 
and an objective $\objective$,
such that $\adfa$ is satisfied by $\akripke$
if and only if~$\winsof{\aterm}{\objective}$.
Intuitively, we model the quantifiers in \Cref{definition:hyperproperty} by the players.
Each quantifier is resolved in one level of urgency throughout the whole term.
Only then, 
and in knowledge of the resolution of previous quantifiers,
a player resolves the next quantifier.
Formally, we define the terms $\aterm_{i}$ for $i \in [1,n+1]$ 
over the alphabet $\analph = \transitions$ by induction:
\[
    \aterm_{n+1} = \tskip
    \qquad \aterm_{i} = \begin{cases}
        \Echoice_{n-i+1} \transitions\appl \aterm_{i+1} & \text{$i$ odd} \\
        \Achoice_{n-i+1} \transitions\appl \aterm_{i+1} & \text{$i$ even}
    \end{cases}
\]
The final term $\aterm$ is a single non-terminal $\aterm = \nonterminal$ 
with $\eqmapof{\nonterminal} = \tskip \echoicen{n} \aterm_1\appl\nonterminal$.
This lets Eve choose the first trace 
and while doing so, 
she also fixed the length of the considered traces.

For the objective $\objective$, 
we want to capture the hyperproperty $\adfa$, 
which already is a DFA.\@
For the transformation to $\objective$ 
we only need to transform the input alphabet from $\transitions^n$ to $\transitions$ 
and make sure the players choose actual traces of the system $\akripke$.
Technically, for a word $\aword \in {(\analph^n)}^*$, 
we obtain $\aword_i \in \analph^*$ by restricting $\aword$ to the $i$th component.
We define the flattening $\flatwords$ 
of a word in ${(\analph^n)}^*$ 
to $\analph^*$ 
inductively by $\flatof{\tskip} = \tskip$ 
and $\flatof{(\atrans_1,\ldots,\atrans_n).\awordp} = \atrans_1\ldots\atrans_n\appl\flatof{\awordp}$.
We set the set of correct runs on component $i$ 
by $\objective_i = \setcond{\aword \in {(\analph^n)}^*}{\aword_i \in \tracesof{\akripke}}$.
\[
    \objective = \flatof{\langof{\adfa} \cap \bigcap_{i\;\mathrm{ odd}} \objective_i \cup \bigcup_{i\;\mathrm{ even}} \overline{\objective_i}}
\]
Note that $\sizeof{\objective}$ is linear in $\sizeof{\adfa} + \sizeof{\akripke}$. 

\begin{theorem}
    $\akripke \models \adfa$ if and only if $\winsof{\aterm}{\objective}$.
\end{theorem}

\begin{proof}
    By construction, $\winsof{\aterm}{\objective}$ 
    if and only if there is $l \in \nat$ with $\winsof{{(\aterm_1)}^l}{\objective}$,
    because unrolling $\nonterminal$ an infinite number of times yields the win to Adam.

    Next we inspect the terms $\aterm(\aword_1, \ldots,\aword_m)$ 
    that can occur whenever urgency $n-m$ is next to be resolved.
    Here, $\aword_i$ are from $\tracesofn{l}{\akripke}$ and for odd $i$ are chosen corresponding to the $\exists$-quantifiers.
    By construction, $\aterm(\aword_1, \ldots,\aword_m)$ has shape
    \[
        \atrans_{1,1}\atrans_{2,1}\ldots\atrans_{m,1}\aterm_{m+1} 
        \;\ldots\; \atrans_{1,l}\atrans_{2,l}\ldots\atrans_{m,l}\aterm_{m+1} \,.
    \]
    where $\aword_i = \atrans_{i,1}\ldots\atrans_{i,l}$. 
    
    We prove the statement: Eve wins $\winsof{\aterm(\aword_1, \ldots,\aword_m)}{\objective}$ 
    if and only if
    $m$ is even and there is $\aword_{m+1} \in \tracesofn{l}{\akripke}$
    with $\winsof{\aterm(\aword_1, \ldots,\aword_m,\aword_{m+1})}{\objective}$,
    or $m$ is odd and for all $\aword_{m+1} \in \tracesofn{l}{\akripke}$
    holds $\winsof{\aterm(\aword_1, \ldots,\aword_m,\aword_{m+1})}{\objective}$.
    
    If $m$ is even, then Eve can choose a sequence $\aword_{m+1} \in \transitions^l$ 
    to transform $\aterm(\aword_1, \ldots,\aword_m)$ into $\aterm(\aword_1, \ldots,\aword_m,\aword_{m+1})$.
    If she chooses a word outside $\tracesofn{l}{\akripke}$ she loses due to $\objective_{m+1}$.
    Thus, she wins if and only if there is $\aword_{m+1} \in \tracesofn{l}{\akripke}$ she can choose 
    and $\winsof{\aterm(\aword_1, \ldots,\aword_m,\aword_{m+1})}{\objective}$.

    The case of $m$ odd is similar.

    Finally, $n=m$ makes $\aterm_{m+1} = \tskip$, so ${\winsof{\aterm(\aword_1, \ldots,\aword_n)}{\objective}}$ if and only if $\aterm(\aword_1, \ldots,\aword_n) \in \flatof{\langof{\adfa}}$ by definition.
\end{proof}

The complexity of checking an $n$-trace hyperproperty ${\akripke \models \adfa}$ 
using this approach is $\kexptime{(2n - 1)}$ bounded (\Cref{Corollary:DecWinner}).
This is still far from optimal, 
considering that $n$-trace hyperproperties can be decided in $\kexpspace{(n - 1)}$~\cite{FRS15}.
The overhead stems from the shape of normalforms, 
which allow for Adam and Eve choices in each layer of urgencies, 
while our approach produces only one type of non-determinism for each urgency.
We intend to investigate this special case in a separate study.

\section{Model Checking Hyperproperties for Recursive Programs}

Checking hyperproperties on pushdown systems is known to be undecidable in the general case~\cite{PT18}.
But showing it undecidable does not satisfy the desire 
to check recursive programs against hyperproperties.
Indeed, a first approach was proposed directly in~\cite{PT18}, 
where one type of quantifier has to work with finite state approximations 
instead of the actual system.
We take a different route to restrict the general setting.
To motivate the restriction, consider the algorithm in \Cref{listing:Karatsuba} 
for multiplication of high-bit numbers.
\definecolor{codegreen}{rgb}{0,0.6,0}
\definecolor{codegray}{rgb}{0.5,0.5,0.5}
\definecolor{codepurple}{rgb}{0.58,0,0.82}
\definecolor{backcolour}{rgb}{0.95,0.95,0.92}
\lstdefinestyle{mystyle}{
    commentstyle=\color{codegreen},
    keywordstyle=\color{magenta},
    numberstyle=\tiny\color{codegray},
    stringstyle=\color{codepurple},
    basicstyle=\ttfamily\footnotesize,
    breakatwhitespace=false,         
    breaklines=true,                 
    captionpos=b,                    
    keepspaces=true,                 
    numbers=left,                    
    numbersep=5pt,                  
    showspaces=false,                
    showstringspaces=false,
    showtabs=false,                  
    tabsize=2
}
\lstset{style=mystyle}
\begin{lstlisting}[language=C++, caption={Karatsuba multiplication.}, label={listing:Karatsuba}]
bit[] KM(bit[] x, bit[] y) { 
    if (len(x) < 64) return (int) x * (int) y;

    mid = len(x) / 2;
    x1 = x[:mid];
    y1 = y[:mid];
    x2 = x[mid+1:];
    y2 = y[mid+1:];

    z2 = KM(x1, y1);
    z0 = KM(x2, y2);
    z1 = KM(x1 + x2, y1 + y2) - z2 - z0;

    return (z2 << 2mid) + (z1 << mid) + z0;
}
\end{lstlisting}%
The algorithm performs a high-bit multiplication of two numbers \texttt{x} and \texttt{y}.
Instead of naively multiplying 64-bit windows of the bit streams,
the algorithm recursively splits the inputs in half.
It performs 3 multiplications on the half-sized integers 
and adds them together for the final result.
The actual algorithm needs some more care for bit-overflows, 
but we shall ignore it here.

Algorithms like \Cref{listing:Karatsuba} are used by security protocols like \texttt{OpenSSL}~\cite{openssl}.
A known weakness of security protocols are timing-based attacks~\cite{Kocher96,BrumleyB03}.
In these attacks, one does not run the program once to derive a secret.
Instead, we measure its execution time over multiple runs with different controllable inputs
and deduct secret values used by the program from the measured execution times.
So, safety from timing attacks can be obtained by requiring a 2-trace hyperproperty on \texttt{KM}:
For all traces $\aword, \awordp$ of \texttt{KM}, their execution time does not differ.
We do not go into modelling details of how to phrase this property as a DFA $\adfa$.
But we make one crucial discovery 
when we want to compare multiple runs of the Karatsuba algorithm:
Its recursion depth is only dependent on the length of \texttt{x} and \texttt{y}, 
a parameter that is very often public knowledge.
Even more important, runs of different recursion depth have close to no chance for a similar execution time.
The previously stated hyperproperty would very likely not be satisfied.
But when the recursion depth parameter is usually known, 
we instead want to ask for a different hyperproperty to hold:
For all recursion depths $d$, and for all traces $\aword, \awordp$ with recursion depth $d$, 
their execution time does not differ.
More generally speaking, the traces $\aword$ and $\awordp$ agree on their \emph{recursive structure}.
We utilize this observation and restrict our check for a hyperproperty $\adfa$ to sets of traces 
that agree on their recursive structure. 

\subsection{Recursive Programs}

We consider the following simple language,
where we abstract away from the actual commands and focus on the recursion principle of the language.
\begin{align*}
    \acode \quad&\Coloneqq\quad \acmd \bnf {\afunc()} \bnf \acode\appl\acode \bnf \cif{\anexpr}{\acode}{\acode} \\
    \afunc \quad&\Coloneqq\quad \text{\texttt{$\afunc()$ -> $\acode\appl\cret$}}
\end{align*}
We have commands $\acmd \in \analph$ to modify the state, 
function calls $\afunc()$ 
and concatenation.
Parameters and return values are passed to/from $\afunc$ by state-manipulation.
A program $\aprog = (\funcs, \eqmap)$ is a set of function symbols $\funcs$ 
with a distinguished initial function $\fmain\in \funcs$
and the function definitions $\eqmap$.

\subsection{Semantics and Recursion Structure}
We assume the domain of booleans $\semdomain = [\vars \to \bool]$. 
The semantics are kept arbitrary for expressions $\semof{\anexpr} : \semdomain \to \bool$ 
and commands $\semof{\acmd}: \semdomain \to \semdomain$.
Semantics of function calls and concatenation are as expected.
\texttt{if}-statements branch left on non-zero evaluation and right on zero.
For the further development, we also require a special command $\cassumeof{\anexpr} \in \analph$.
Its semantics operate on a fresh variable $\assvar$ 
and are tuned to observe whether an $\cassume$ command failed.
\[
    \semof{\cassumeof{\anexpr}}(d) = \begin{cases}
        d & \semof{\anexpr}(d) = \tvar \\
        d[{\assvar}\mapsto{\tvar}] & \semof{\anexpr}(d) = \fvar
    \end{cases}    
\]

A \emph{branching} of $\aprog$ is a finite tree $\atree$ (a prefix closed subset of~$\nat^*$) with a labelling $\trlab$.
The labelling assigns nodes of $\atree$ commands $\acmd$, function calls $\afunc$, and $\cret$ statements.
The root node is labeled with ${\trlabof{\varepsilon} = \fmain}$.
Nodes $\anode$ labelled by $\trlabof{\anode} = \acmd$ or $\cret$ are leaves.
A \hbox{$\trlabof{\anode}=\afunc(n)$-labelled} node has children 
corresponding to one branch of their body $\acode$.
The set of branches $\branchesof{\acode}$ is 
\begin{gather*}
    \begin{aligned}
        \branchesof{\acmd} &= \set{\acmd} 
        & \branchesof{\afunc()} &= \set{\afunc}
        & \branchesof{\acode_1.\acode_2} &= \branchesof{\acode_1}\appl\branchesof{\acode_2}     
    \end{aligned} \\
    \branchesof{\cif{\anexpr}{\acode_1}{\acode_2}} 
    = \cassumeof{\anexpr}.\branchesof{\acode_1} \cup \cassumeof{!\anexpr}.\branchesof{\acode_2}
\end{gather*}
So if $\aword \in \branchesof{\acode}$ is the chosen branch, then $\anode$ has $\sizeof{\aword}$ children 
labelled by the corresponding $\acmd$, $\afunc$, or $\cret$ in $\aword$.

A \emph{trace} of $\aprog$ is a pair $(\trlab, \trst)$ 
where $\trlab$ is a branching with domain $\atree$
and $\trst : \atree \to \semdomain$ is a state assignment.
Every node $\anode$ is labelled by $\trstof{\anode} \in \semdomain$, 
the state of the program before execution of $\trlabof{\anode}$,
$\trstof{\varepsilon} \in \semdomain$ is the input state.
Consider node $\anode$ with $\trlabof{\anode} = \afunc$, 
so the children are one branch of its defining body.
Its leftmost child $\anode.0$ inherits the state, $\trstof{\anode.0} = \trstof{\anode}$.
The rightmost child $\anode.m$ has label $\trlabof{\anode.m} = \cret$ 
and yields back $\semof{\afunc}(\trstof{\anode}) = \trstof{\anode.m}$.
Intermediate children $\anode.i$ just carry out their semantics to the next node by
$\trstof{\anode.i+1} = \semof{\trlabof{\anode.i}}(\trstof{\anode.i})$.

A \emph{recursion structure} $\arec$ for $\aprog$ is a tree labelling $\arec : \atree \to \funcs$.
The internal of a tree $\atree$ 
is the set $\intof{\atree} = \setcond{\anode}{\anode.0 \in \atree}$.
\begin{definition}
    A branching $\trlab$ has recursion structure $\arec$ if $\intof{\domof{\trlab}} = \domof{\arec}$ and $\trlab$ and $\arec$ coincide on $\domof{\arec}$.
\end{definition}
The recursive structure thus identifies runs that differ only in $\trst$ 
and the actual commands executed,
but the function calls are exactly the same.
We say $(\trlab, \trst)$ is an $\arec$-trace when $\trlab$ has recursion structure $\arec$.
A recursion structure $\arec$ is proper, if there is an $\arec$-trace of $\aprog$.

\subsection{Checking Hyperproperties for similar Recursion Structure}

We define the yield $\yieldof{\atree}$ of a tree as usual.
The set of trace-observations 
with recursion structure $\arec$
is the set $\tracesofn{\arec}{\aprog} 
    = \setcond{\yieldof{\trlab}}{(\trlab, \trst) \text{ is an  $\arec$-trace of $\aprog$}}
    \subseteq \analph^*$.

\begin{definition}
    Let $\adfa$ be an $n$-trace hyperproperty.
    A program $\aprog$ satisfies $\adfa$, $\aprog \models \adfa$, 
    if there is a proper recursion structure $\arec$ with 
    \[
        \exists \aword_1 \in \tracesofn{\arec}{\aprog}
        \forall \aword_2 \in \tracesofn{\arec}{\aprog}
        \ldots
        Q \aword_n \in \tracesofn{\arec}{\aprog}.
        \prod_{i=1}^{n} \aword_i \in \langof{\adfa}
    \]
\end{definition}

To present our approach, we assume that programs are prefix-branching.
We call a program $\aprog$ \emph{prefix-branching} when function calls are never succeeded by commands $\aletter \in \analph$.
Formally, we assume the program code stems from the following, slightly different grammar:
\begin{align*}
    \afuncseq \quad&\Coloneqq\quad \afunc() \bnf \afuncseq\appl\afuncseq\\
    \acode \quad&\Coloneqq\quad \acmd \bnf \acode\appl\acode \\
    \pbcode \quad&\Coloneqq\quad \afuncseq \bnf \acode\appl\pbcode \bnf \cif{\anexpr}{\pbcode}{\pbcode}\\
    \afunc \quad&\Coloneqq\quad \text{\texttt{$\afunc()$ -> $\pbcode\appl\cret$}}
\end{align*}
Note, that this is not only a presentational decision.
While every program $\aprog$ has an equivalent representation in prefix-branching form, 
the translation incurs more functions 
and thus, fixing a recursion structure may fix more behavior (e.g.\ branching beyond function calls) than desired.
This means that we actually restrict our class of programs by enforcing prefix-branching.
We do not exactly know, how impactful the restriction is for practical code, 
but our example (\Cref{listing:Karatsuba}) exerts the desired structure (up to parameter passing).

\begin{theorem}
    $\aprog \models \adfa$ is decidable for prefix-branching $\aprog$.
\end{theorem}

Similar as for the finite state case, 
we will translate $\aprog$ and $\adfa$ into a term $\aterm$ and an objective $\objective$.
The intent is to first fix a recursion structure, 
and then replay the call of a function $\afunc$ with function body $\acode$
by the rewriting of a non-terminal $\afunc$ into a term $\aterm(\afunc)$.
Term $\aterm(\afunc)$ basically chooses the branch through~$\acode$.
For prefix-branching, a branch is a word from~$\analph^*\funcs^*$.
With a fixed recursive structure also the sequence of function calls in~$\acode$ is fixed.
For a fixed sequence of function calls $\afuncseq = \afunc_1 \ldots \afunc_m$ 
the set of available prefixes is captured by ${\branchesofn{\afuncseq}{\acode} \subseteq \analph^*}$ with
\[
    \branchesofn{\afuncseq}{\acode}\appl\afuncseq
    = \branchesof{\acode} \cap \analph^*\appl \afuncseq\,.
\]
We assume that all command prefixes (the $\analph^*$ part) in $\branchesofn{\afuncseq}{\acode}$ share the same length,
$\branchesofn{\afuncseq}{\acode} \subseteq \analph^{l}$ for some $l \in \nat$.
Similar to the finite state case this is a modelling issue for synchronization towards $\adfa$.
It can be achieved by introducing $\tskip$s to branches too short.
Since $\branchesof{\acode}$ is finite, 
there is only a finite set of occurring function call sequences 
$\funcseqsof{\acode} = \setcond{\afuncseq}{\branchesofn{\afuncseq}{\acode} \neq \emptyset }$.
Finally, we define $\aterm(\afunc)$ in a term game $(\funcs, \eqmap)$ with $\eqmapof{\afunc} = \aterm(\afunc)$.
\begin{align*}
    \aterm_{n+1}^{\acode, \afuncseq}(\aword_1,\ldots,\aword_n) &= \flatof{\aword_1,\ldots,\aword_n}\\
    \aterm_{i}^{\acode, \afuncseq}(\aword_1,\ldots,\aword_{i-1}) & = \begin{cases}
        \Echoice^{n-i+1}_{\aword_i \in \branchesofn{\afuncseq}{\acode}}
            {\aterm_{i+1}^{\acode, \afuncseq}(\aword_1,\ldots,\aword_i)} & \text{$i$ odd} \\
        \Achoice^{n-i+1}_{\aword_i \in \branchesofn{\afuncseq}{\acode}}
            {\aterm_{i+1}^{\acode, \afuncseq}(\aword_1,\ldots,\aword_i)} & \text{$i$ even}     
    \end{cases} \\   
    \aterm(\afunc) &= \bigEchoiceOf{n}_{\afuncseq \in \funcseqsof{\acode}}{\aterm_1^{\acode,\afuncseq}()\appl\afuncseq}
\end{align*}

It remains to model the objective $\objective$.
As before, we construct sets $\objective_i \subseteq {(\analph^n)}^*$ for $i \in [1,n]$.
\begin{align*}
    \objective &= \flatof{\langof{\adfa} \cap \bigcap_{i \mathrm{ odd}} \objective_i \cup \bigcup_{i \mathrm{ even}} \overline{\objective_i}}\\
    \objective_i &= \setcond{\aword \in {(\analph^n)}^*}{\semof{\aword_i}(\assvar) = \tvar}
\end{align*}
A command sequence $\aword \in \analph^*$ operates on the specific finite domain $\semdomain$.
These language is regular because $\semdomain$ is finite.

\begin{theorem}
    $\winsof{\fmain}{\objective}$ if and only if $\aprog \models \adfa$.
\end{theorem}

\begin{proof}[Proof Sketch]
    Notice that the set of all reachable word terms from $\aterm_1^{\acode, \afuncseq}$ 
    is exactly $\flatof{\branchesofn{\afuncseq}{\acode}^n}$.
    So by construction, playing from $\fmain$ until the term has no more non-terminals 
    yields a term $\aterm_{\arec}$ for some recursion structure $\arec$,
    where the reachable word terms 
    form precisely the set ${\flatof{\branchesofn{\arec}{\aprog}^n}}$, 
    where $\branchesofn{\arec}{\aprog} 
        = \setcond{\yieldof{\trlab}}{\trlab \text{ is an $\arec$-branching}}$.
    As before, the alternative is that there is an infinite rewriting of non-terminals, 
    in which case Eve loses.
    Thus, $\winsof{\fmain}{\objective}$ if and only if there is $\arec$ such that $\winsof{\aterm_\arec}{\objective}$.

    Next we again inspect the terms $\aterm(\aword_1, \ldots,\aword_m)$ 
    that can occur whenever urgency $n-m$ is next to be resolved (at first, we have $\aterm_{\arec} = \aterm()$ with $m=0$).
    This time, $\aword_i$ are from $\tracesofn{\arec}{\aprog}$ and for odd $i$ are chosen corresponding to the $\exists$-quantifiers.
    By construction, $\aterm(\aword_1, \ldots,\aword_m)$ has shape
    \[
        \aterm_{m+1}^{\acode^1, \afuncseq^1}(\aword^1_1,\ldots, \aword^1_m)
        \appl\;\ldots\;
        \appl \aterm_{m+1}^{\acode^l, \afuncseq^l}(\aword^l_1,\ldots, \aword^l_m)\,,
    \]
    where $\acode^1 \ldots \acode^l$ is the depth-first left to right traversal of $\arec$.

    We prove: $\winsof{\aterm(\aword_1, \ldots,\aword_m)}{\objective}$ 
    if and only if 
    $m$ is even and there is $\aword_{m+1} \in \tracesofn{\arec}{\aprog}$
    with $\winsof{\aterm(\aword_1, \ldots,\aword_m,\aword_{m+1})}{\objective}$,
    or $m$ is odd and for all $\aword_{m+1} \in \tracesofn{\arec}{\aprog}$
    we have $\winsof{\aterm(\aword_1, \ldots,\aword_m,\aword_{m+1})}{\objective}$.
    
    If $m$ is even, then Eve can choose a sequence of branches 
    $\aword^i_{m+1} \in \branchesofn{\afuncseq^i}{\acode^i}$ 
    which together form the yield of a branching ${\aword_{m+1} = \aword^1_{m+1} \ldots \aword^l_{m+1} 
        \in \branchesofn{\arec}{\aprog}}$.
    She does so by transforming $\aterm(\aword_1, \ldots,\aword_m)$ 
    into $\aterm(\aword_1, \ldots,\aword_m,\aword_{m+1})$.
    If she chooses a branching which has no trace 
    $\aword_{m+1} \notin \tracesofn{\arec}{\aprog}$ 
    she loses due to $\objective_{m+1}$.
    Thus, she wins if and only if there is 
    $\aword_{m+1} \in \tracesofn{\arec}{\aprog}$ to choose     
    and we have ${\winsof{\aterm(\aword_1, \ldots,\aword_m,\aword_{m+1})}{\objective}}$.
    The case of $m$ odd is similar.
    
    The final case, where $n=m$ holds, implies that the immediate terms are 
    ${\aterm_{m+1}^{\acode^i, \afuncseq^i}(\aword^i_1,\ldots, \aword^i_m)   
        = \flatof{\aword^i_1,\ldots,\aword^i_n}}$
    and by definition of $\winsof{}{}$ for word terms, ${\winsof{\aterm(\aword_1, \ldots,\aword_n)}{\objective}}$ 
    if and only if ${\aterm(\aword_1, \ldots,\aword_n) \in \flatof{\langof{\adfa}}}$.
\end{proof}

\section{Specialized Decision Algorithm for Weak Terms with Linear Grammars}%
\label{Section:WeakLinearGrammars}

In this section, we show that a certain fragment, weak linear grammars/terms, admits more efficient decision algorithms.
For maximal urgency $\maxurg$, the complexity goes down from $\kexptime{2\maxurg-1}$ to $(\maxurg-2)-\mathsf{EXPSPACE}$.
In the following, we first build up to the definition of weak linearity.
Then, we show how observations in linear terms with weak linear grammars can be decided in $(\maxurg-2)-\mathsf{EXPSPACE}$.

The weak fragment of urgency imposes a restriction on available choice operators per urgency.
But as \Cref{axiom:norm} suggests, choice operators can potentially combine to mimic the effect of disallowed choice operators.
For example, let $\echoicen{1}$ be disallowed, however let $\achoicen{1}$ and $\echoicen{2}$ be both allowed.
Then, we observe 
\[
    \achoicen{1} (\aterm\echoicen{2}\atermp) \axeqper{\rnormalformname} \achoicen{1}(\aterm \echoicen{1} \atermp)\axeq (\aterm\echoicen{1}\atermp)
\]
We avoid this effect by further restricting the structure of weak terms.
To this end, we define the internal urgency of a term.
Intuitively, the internal urgency $\inurgof{\aterm}$ of a term $\aterm$ is the largest urgency that appears in a choice operator.
\begin{align*}
    \inurgof{\aterm \pchoicen{\anurg} \atermp} &= \max\set{\inurgof{\aterm}, \inurgof{\atermp}, \anurg} &
    \inurgof{\aterm\appl\atermp} &= \max\set{\inurgof{\aterm}, \inurgof{\atermp}}
\end{align*}

The internal urgency of non-terminals requires a more subtle definition.
This is because for $\nonterminal\in\nonterminals$, even if the term $\eqmapof{\nonterminal}$ does not contain choice operators from a certain urgency, 
$\eqmapof{\nonterminal}$ may invoke some other $\nonterminalp$ for which $\eqmapof{\nonterminalp}$ might recreate that urgency.
To obtain a sound definition, we have to consider $\eqmapof{\nonterminal}$ for all $\nonterminal\in\nonterminals$ at once.
We define $\inurgof{\nonterminal}$ for each $\nonterminal\in\nonterminals$ to be the smallest assignment that satisfies $\inurgof{\nonterminal}=\inurgof{\eqmapof{\nonterminal}}$, along with the properties of $\inurgof{.}$ defined above.
Intuitively, this means that non-terminals with internal urgency $\anurg$ can only use choice operators of urgency $\anurg$ and non-terminals with internal urgency at most $\anurg$.
With this at hand, we define urgency-alignment.
Intuitively, urgency-alignment of $\aterm$ requires that the urgency of all terms derivable not greater than $\urgencyof{\aterm}$.
Formally, we use an inductive definition.
Terms $\aterm=\nonterminal$ or $\aterm\in\Sigma$ are always urgency-aligned.
A term $\aterm=\atermp\appl\atermpp$ is urgency-aligned if $\atermp$ and $\atermpp$ are both urgency aligned, and a term $\aterm=\atermp\pchoicen{\anurg}\atermpp$ is urgency aligned if both $\atermp$ and $\atermpp$ are urgency-aligned with $\inurgof{\atermp}, \inurgof{\atermpp}\leq\anurg$. 
A grammar is urgency-aligned if $\eqmapof{\nonterminal}$ is urgency-aligned for all $\nonterminal\in\nonterminals$.

We are now ready to define weakness and linearity.
We say that a term $\aterm$ is \emph{weak}, if it, and its grammar are urgency-aligned, with $\aterm$ and $\eqmapof{\nonterminal}$ for all $\nonterminal\in\nonterminals$ only containing choice operators from $\set{\echoicen{\anurg} \mid \anurg \text{  even}}\cup\set{\achoicen{\anurg} \mid \anurg \text{ odd}}$.
Weak terms are normalized into ($\objective$-specialized) weak normal forms $\wnfn{\anurg}$ in $(\anurg-1)-\mathsf{EXPTIME}$ \Cref{Lemma:EffectiveNormWeak}.
The base case $\wnfn{0} = \specgnfn{0}$ stays the same.
For positive urgencies, odd $\anurg$ and even $\anurgp$:
\begin{align*}
    \wnfn{\anurg}&=
    \setcond{\bigAchoiceOf{\anurgp}\atermset}{\emptyset\neq\atermset\subseteq\wnfn{\anurg-1}}
    & \wnfn{\anurgp}&=
    \setcond{\bigEchoiceOf{\anurg}\atermset}{\emptyset\neq\atermset\subseteq\wnfn{\anurg-1}}
\end{align*}
Linearity restricts the recursive behaviour of the term and the grammar.
The syntactic construction of a \emph{linear} term $l$ is given below, where $\aterm$ represents an arbitrary term that does not contain non-terminals.
We call a grammar $(\nonterminals, \eqmap)$ linear if $\eqmapof{\nonterminal}$ is linear for all $\nonterminal\in\nonterminals$.
\begin{align*}
    l \quad ::= \quad \nonterminal \mid \aterm \mid l\pchoicen{\anurg} l \mid l.\aterm \mid \aterm.l 
\end{align*}

\subsection{Weak Normal Form}\label{Appendix:WeakNormalForm}
We present a slightly modified normalization algorithm that takes advantage of the restricted expressiveness of weak grammars and weak terms.
More specifically, we modify the functions $\specnormchoiceof{.}$ and $\specnormconcof{.}$ to this end.
The modified functions return terms with urgencies no larger than the urgency of the input and preserve weakness.
The proofs remain largely the same. 
The normalization process of non-terminals remains the same in relation to $\specnormchoiceof{.}$ and $\specnormconcof{.}$.
In this case, we extract upper bounds by employing the lattice height of $\wnfn{\maxurg}$.
We prove the lemma below in this section.
Note that for $\maxurg\geq 2$, the power $.^{5}$ gets absorbed by $\bigoof{.}$.

\begin{lemma}[Modifies \Cref{Lemma:EffNormalForm}]\label{Lemma:EffectiveNormWeak}
    Given a weak term $\aterm$, a weak grammar $(\nonterminals, \eqmap)$, and a regular objective $\objective\subseteq\Sigma^{*}$ given as a DFA, we can compute $\specnormof{\aterm}\in\wnfn{\maxurg}$ with $\specnormof{\aterm}\specaxeq{\objective}\aterm$ in time $(|\aterm| + |\eqmap||\nonterminals|)\cdot\repexpof{\maxurg-1}{\bigoof{|\specgnfnof{0}{\objective}|^{5}}}$.
\end{lemma}

We first handle $\specnormchoiceof{.}$.

\begin{lemma}[Modifies
    \Cref{Lemma:ResolveChoice}]\label{Lemma:EffectiveResolveChoiceWeak}
    Let $\atermset\subseteq\wnfn{\anurg}$,
    and $\atermpp=\bigPchoiceOf{\anurg}\atermset$.
    We can construct its normalform $\specnormchoiceof{\atermpp}\in\wnfn{\anurg}$
    with ${\specnormchoiceof{\atermpp}\specaxeq{\objective}\atermpp}$ 
    in time $\bigoof{|\atermpp|}\leq\repexpof{\maxurg-1}{\bigoof{|\specgnfnof{0}{\objective}|}}$.
\end{lemma}

\begin{proof}
    We know for each $\aterm'\in\atermset$ we have $\aterm'=\bigPchoiceOf{\anurg}\atermsetp_{\aterm'}$.
    We simply construct the union $\atermpp'=\bigPchoiceOf{\anurg}\bigcup_{\aterm'\in\atermset}\atermsetp_{\aterm'}$ in time $\bigoof{|\atermpp|}\leq\repexpof{\maxurg-1}{\bigoof{|\specgnfnof{0}{\objective}|}}$.
    We observe $\atermpp\axeq\atermpp'$ by $\rlatassocname$.
    This concludes the proof.
\end{proof}

\begin{lemma}[Modifies
    \Cref{Lemma:ResolveConcat}]\label{Lemma:EffectiveResolveConcatWeak}
    For any $\aterm,\atermp\in\wnfn{\anurg}$, 
    we can find $\specnormconcof{\aterm\appl\atermp}
    \in\wnfn{\anurg}$
    with
    $\specnormconcof{\aterm\appl\atermp}
    \specaxeq{\objective}\aterm\appl\atermp$ 
    in time $\repexpof{\maxurg-1}{\bigoof{|\specgnfnof{0}{\objective}|^{4}}}$.
\end{lemma}

\begin{proof}
    The base case $\maxurg=1$ is handled similarly to the inductive case, so we only show the inductive case.
    We first use the distribution axioms on the concatenation.
    \begin{align*}
        \aterm\appl\atermp
        &\overset{\hphantom{2\times\labelcref{axiom:dist-right}}}{=}(\bigPchoiceOf{\anurg}\atermset)
        \appl(\bigPchoiceOf{\anurg}\atermsetp)\\
        &\axeqper{\labelcref{axiom:dist-right}}\bigPchoiceOf{\anurg}_{\aterm' \in \atermset}
        \aterm'\appl
        (\bigPchoiceOf{\anurg}\atermsetp)\\
        &\axeqper{\labelcref{axiom:dist-left}}
        \bigPchoiceOf{\anurg}_{\aterm' \in \atermset}
        \bigPchoiceOf{\anurg}_{\atermp' \in \atermsetp}
        \aterm'\appl\atermp'\ . 
        \intertext{The resulting term allows for an application of the induction hypothesis on the subterms with lower urgency for $\maxurg>1$.}
        &\axeqper{\mathrm{I.H.}}
        \bigPchoiceOf{\anurg}_{\aterm' \in \atermset}
        \bigPchoiceOf{\anurg}_{\atermp' \in \atermsetp}
        \specnormconcof{\aterm'\appl\atermp'}\ . 
    \end{align*}
    The distribution causes the term to jump to at most the square of its previous size.
    In any case, there are at most $\repexpof{\maxurg-1}{\bigoof{|\specgnfnof{0}{\objective}|}}^{2}$ choice operands of lower urgency.
    For $\maxurg=1$, the concatenation is resolved by monoid multiplication which takes $\bigoof{|\specgnfnof{0}{\objective}|}^{2}$ time.
    This results in an overall complexity of $\repexpof{\maxurg-1}{\bigoof{|\specgnfnof{0}{\objective}|}}^{4}$.
    For $\maxurg>1$ we employ the induction hypothesis and see that the normalization of each subterm takes $\repexpof{\maxurg-2}{\bigoof{|\specgnfnof{0}{\objective}|}}^{2}$ time.
    So in total we again need at most $\repexpof{\maxurg-1}{\bigoof{|\specgnfnof{0}{\objective}|}}^{4}$ steps.
\end{proof}

Lemma~\ref{Lemma:EffectiveNormWeak} follows as a corollary of two observations.
First, each iteration of the least fixed point calculation will require at most $|\eqmap|\cdot\repexpof{\maxurg-1}{\bigoof{|\specgnfnof{0}{\objective}|}}^{4}$ steps for normalization.
Second, we only need $|\nonterminals|\cdot|\wnfn{\maxurg-1}|$ iterations towards a fixed point, since the height of the lattice $\wnfn{\maxurg}^{|\nonterminals|}$ is bounded by this number.
Replacing the non-terminals in $\aterm$ by their normalizations, the term grows by at most a factor of $\repexpof{\maxurg-1}{\bigoof{|\specgnfnof{0}{\objective}|}}$.
This is then normalized in $|\aterm|\cdot\repexpof{\maxurg-1}{\bigoof{|\specgnfnof{0}{\objective}|}}^{5}$.
Resulting in an overall time complexity of at most $(|\aterm| + |\eqmap||\nonterminals|)\cdot\repexpof{\maxurg-1}{\bigoof{|\specgnfnof{0}{\objective}|^{5}}}$.

\subsection{Deciding the winner for weak linear terms}
The main claim of the section is given in the lemma below.
\begin{lemma}\label{Lemma:EfficientWeakLin}
    Let $(\nonterminals, \eqmap)$ be a weak linear grammar with maximal urgency $\maxurg$, $\aterm$ a weak linear term, and $\objective\subseteq\analph^{*}$ a regular objective. 
    Then $\winsof{\aterm}{\objective}$ can be decided in ${(\max\set{|\eqmap|^{2}, |\aterm|^{2}}.|\eqmap|.|\nonterminals|.\repexpof{\maxurg-2}{\bigoof{|\specgnfnof{0}{\objective}|}})}^{4}$ space.
\end{lemma}

We fix a weak linear grammar $(\nonterminals, \eqmap)$ with maximal urgency $\maxurg$, and a regular objective $\objective$ for the rest of this section.
A naive normalization of $\aterm$ into $\wnfn{\maxurg}$ gives us a complexity of $\kexptime{(\maxurg-1)}$ \Cref{Lemma:EffectiveNormWeak}.
Here, the size of the normal form terms is the main bottleneck, which can grow up to $\repexpof{\maxurg-1}{\bigoof{|\objective|}}$.
To circumvent this problem, the algorithm avoids a complete normalization.
Instead, it aims to non-deterministically guess a small-size portion $\atermp\in\wnfn{\maxurg-1}$ of the normalization, which suffices to show $\winsof{\aterm}{\objective}$.
The existence of this is guaranteed by the lemma below.
Because of it admits a direct check, we use the domination preorder $\discleqof{\objective}$ instead of $\specaxleq{\objective}$ or $\speccongleq{\objective}$.
\begin{lemma}\label{Lemma:WeakLinPortion}
    Let $\aterm$ be a weak term.
    Then, there is a $\atermp\in\wnfn{\maxurg-1}$ with 
    $\echoicen{\maxurg}\atermp\discleqof{\objective}\specnormof{\aterm}$, if $\maxurg$ is even 
    and $\achoicen{\maxurg} \atermp\discgeqof{\objective}\specnormof{\aterm}$ if $\maxurg$ is odd, where $\winsof{\aterm}{\objective}$ if and only if $\winsof{\atermp}{\objective}$.
\end{lemma}

\begin{proof}
    Assume $\maxurg$ is even.
    The case with odd $\maxurg$ is dual.
    Let $\aterm$ be a weak term.
    Then, there is a term $\aterm'\in\wnfn{\maxurg}$ with $\specnormof{\aterm}=\aterm'$.
    The definition of $\wnfn{\maxurg}$ tells us  $\aterm'=\bigEchoiceOf{\maxurg}\atermset$ for some $\atermset\subseteq\wnfn{\maxurg-1}$.
    We observe that $\winsof{\aterm'}{\objective}$ holds if and only if there is a $\atermpp\in\atermset$ with $\winsof{\atermpp}{\objective}$.
    Using the definition of $\discleqof{\objective}$, we also get $\echoicen{\maxurg}\atermpp\discleqof{\objective}\aterm'$ for any $\atermpp\in\atermset$.
    Consider the case of $\winsof{\aterm'}{\objective}$, thus there is $\aterm''\in\atermset\subseteq\wnfn{\maxurg-1}$ and $\aterm''\discleqof{\objective}\aterm'$ as well as $\winsof{\aterm''}{\objective}$ holds.
    Similarly, if $\notwinsof{\aterm'}{\objective}$, any element $\aterm''\in\atermset\subseteq\wnfn{\maxurg-1}$ has $\notwinsof{\aterm''}{\objective}$ and $\aterm''\discleqof{\objective}\aterm'$.
\end{proof}

It remains to show that this fragment can be effectively guessed.
This is stated by the lemma below.
The normalization process tells us that for $k=|\nonterminals|.|\wnfn{\maxurg-1}|$, where $|\wnfn{\maxurg-1}|\leq\repexpof{\maxurg}{|\objective|}$, we get $\replaceLFPIt{\aterm}{k}\axeq\aterm$.
Here, $\replaceLFPIt{\aterm}{i}$ is $\aterm$ with non-terminals replaced by 
the $i$'th iteration towards a fixed point as discussed in \Cref{Appendix:WeakNormalForm}, 
the syntax is borrowed from \Cref{Appendix:NTNormalization}.
Initializing $i=k$ in the lemma below yields the desired $(\maxurg-2)-\mathsf{EXPSPACE}$ decidability result.
To achieve this complexity, the algorithm avoids maintaining a call stack.

\begin{lemma}\label{Lemma:WeakLinLightRec}
    Let $\aterm$ be a weak linear term with urgency $\maxurg$, let $\atermpppp\in\wnfn{\maxurg-1}$ and let $i\in\nat$.
    If $\aterm$ does not contain non-terminals, $\atermpppp\discleqof{\objective}\specnormof{\aterm}$ can be checked in space

    \[
        {(|\aterm|^{2}.{(|\eqmap|.|\nonterminals|.\repexpof{\maxurg-2}{\bigoof{|\specgnfnof{0}{\objective}|}})}^{2})}^{2}.
    \]

    If conversely, $\aterm$ contains non-terminals, $\atermpppp\discleqof{\objective}\replaceLFPIt{\aterm}{i}$ can be checked in space 

    \[
        {(\bigoof{\log(i)}+\max\set{|\eqmap|^{2}, |\aterm|^{2}}.{(|\eqmap|.|\nonterminals|.\repexpof{\maxurg-2}{\bigoof{|\specgnfnof{0}{\objective}|}})}^{2})}^{2}.
    \]
\end{lemma}

\begin{proof}
    We only handle the case where $\maxurg$ is even, the case with an odd $\maxurg$ is dual.
    Our algorithm is a non-deterministic space bounded algorithm.
    We show that the check is possible in non-deterministic 
    $|\aterm|^{2}.\normcomplexity^{2}$ space if $\aterm$ does not contain non-terminals, and non-deterministic
    $\bigoof{\log(i)}+\max\set{|\aterm|^{2}, |\eqmap|^{2}}.\normcomplexity^{2}$ space if $\aterm$ contains non-terminals, where $\normcomplexity=|\eqmap|.|\nonterminals|.\repexpof{\maxurg-2}{\bigoof{|\objective|}}$.
    These complexities can be translated into deterministic space by Savitch's Theorem~\cite{SAVITCH}.
    This yields the deterministic complexities ${(|\aterm|^{2}.\normcomplexity^{2})}^{2}$ without non-terminals and ${(\bigoof{\log(i)}+\max\set{|\aterm|^{2}, |\eqmap|^{2}}.\normcomplexity^{2})}^{2}$ with non-terminals.
    
    We proceed by an outer induction on $i\in\nat$ and an inner structural induction on $\aterm$.
    The base case is $\aterm\in\Sigma\cup\set{\terr,\tskip}$ or $\aterm\in\nonterminals$.
    The check $\atermpppp\discleqof{\objective}\specnormof{\aterm'}$ for $\aterm'\in\Sigma\cup\set{\terr,\tskip}$ can be done in $\normcomplexity$ time, and thus in the same amount of space, by employing the usual normalization procedure followed by evaluating the relation.
    This also covers the case $\aterm\in\nonterminals$ as $\replaceLFPIt{\nonterminal}{0}=\terr$ for all $\nonterminal\in\nonterminals$.
    The inductive case for the inner induction is very similar between the outer base case $i=0$ and the outer inductive case moving from $i-1$ to $i$, so we do not handle it separately.
    
    We continue with the outer inductive case.
    For the base case of the inner induction, where $i\in\nat$, we have $\aterm\in\nonterminals$ or $\aterm\in\Sigma\cup\set{\terr,\tskip}$.
    The case $\aterm\in\Sigma\cup\set{\terr,\tskip}$ is handled exactly as in the case we initially handled.
    Let $\aterm=\nonterminal\in\nonterminals$.
    Then, $\replaceLFPIt{\aterm}{i}=\replaceLFPIt{\eqmapof{\nonterminal}}{i-1}$.
    Because $(\nonterminals, \eqmap)$ is a linear grammar, $\eqmapof{\nonterminal}$ is a linear term.
    We call the induction hypothesis to examine
    $\atermpppp\discleqof{\objective}\specnormof{\replaceLFPIt{\eqmapof{\nonterminal}}{i-1}}$ and returns its result.
    This is possible in non-deterministic $\bigoof{\log(i-1)}+\max\set{|\eqmapof{\nonterminal}|^{2}, |\eqmap|^{2}}.\normcomplexity^{2}$ space by the induction hypothesis.
    Note that we do not need to compute beyond calling the induction hypothesis, so we do not need additional space for it (i.e.\ we are a tail recursion).
    
    For the inner inductive case, let $\aterm$ be a weak linear term.
    If we have $\inurgof{\aterm}<\maxurg$, all non-terminals that appear in $\aterm$ must also have internal urgency of at most $\maxurg-1$.
    These non-terminals can only refer to non-terminals of internal urgency at most $\maxurg-1$.
    Then we can apply the normalization procedure, limited only to these non-terminals, to construct $\aterm'=\specnormof{\replaceLFPIt{\aterm}{i}}$.
    This takes $\bigoof{\log(i)}+|\aterm|.\normcomplexity$ time, i.e.\ the same amount of space.
    The former term stems from the counter used to stop at the $i$-th LFP iteration, and the latter term stems from the normalization.

    Now let $\aterm=\atermp\echoicen{\maxurg}\atermpp$.
    We have $\specnormof{\replaceLFPIt{\aterm}{i}}=\specnormof{\replaceLFPIt{\atermp}{i}}\echoicen{\maxurg}\specnormof{\replaceLFPIt{\atermpp}{i}}$.
    Even if $\aterm$ does not contain non-terminals, we still handle $\replaceLFPIt{\aterm}{i}$ instead of the equal term $\aterm$ to unify our notation.
    We write $\atermp'=\specnormof{\replaceLFPIt{\atermp}{i}}$, $\atermpp'=\specnormof{\replaceLFPIt{\atermpp}{i}}$, and also let 
    $\atermp'=\bigEchoiceOf{\maxurg}\atermsetp$ and 
    $\atermpp'=\bigEchoiceOf{\maxurg}\atermsetpp$, where $\atermsetp, \atermsetpp\subseteq\wnfn{\maxurg-1}$.
    We stress that the algorithm will not explicitly construct these terms. 
    We have 
    \[
        \specnormof{\replaceLFPIt{\aterm}{i}}=(\bigEchoiceOf{\maxurg}\atermsetp)\echoicen{\maxurg}(\bigEchoiceOf{\maxurg}\atermsetpp)\disceq
        \bigEchoiceOf{\maxurg}(\atermsetp\cup\atermsetpp).
    \]
    The final equality $\disceq$, i.e. $(\discleqof{\objective}\cap\discgeqof{\objective})$, follows from the definition of $\discleqof{\objective}$, which ignores bracketing among choices.
    By the definition of $\discleqof{\objective}$, we observe that $\atermpppp\discleqof{\objective}\replaceLFPIt{\aterm}{i}$ holds if and only if there is a $\atermppp\in\atermsetp$ or $\atermppp\in\atermsetpp$ with $\atermpppp\discleqof{\objective}\atermppp$. 
    The conditions $\atermppp\in\atermsetp$ and  $\atermppp\in\atermsetpp$ are respectively equivalent to  $\atermppp\discleqof{\objective}\atermp$ and $\atermppp\discleqof{\objective}\atermpp$.
    This means that a $\atermppp\in\set{\atermp, \atermpp}$ exists with $\atermpppp\discleqof{\objective}\atermppp$ if and only if $\atermpppp\discleqof{\objective}\replaceLFPIt{\aterm}{i}$ holds.
    To verify the existence of such terms, algorithm first guesses a term $\atermppp\in\set{\atermp, \atermpp}$.
    Then, it removes the rest of the term from the memory and checks $\atermpppp\discleqof{\objective}\replaceLFPIt{\atermppp}{i}$.
    This can be done in non-deterministic $\bigoof{\log(i)}+\max\set{|\atermp|^{2}, |\atermpp|^{2}, |\eqmap|^{2}}.\normcomplexity^{2}$ space by the induction hypothesis.
    This yields the desired complexity.
    If $\aterm$ does not contain non-terminals, we have $\replaceLFPIt{\atermppp}{i}=\atermppp$ for all $\atermppp\in\set{\atermp, \atermpp}$ and the check is possible in $\max\set{|\atermp|^{2}, |\atermpp|^{2}}.\normcomplexity^{2}$ space.

    For the remaining inductive case, we have $\aterm=\atermp.\atermpp$.
    We similarly write $\atermp'=\specnormof{\replaceLFPIt{\atermp}{i}}=\bigEchoiceOf{\maxurg}\atermsetp$ and $\atermpp'=\specnormof{\replaceLFPIt{\atermp}{i}}=\bigEchoiceOf{\maxurg}\atermsetpp$ with $\atermsetp, \atermsetpp\subseteq\wnfn{\maxurg-1}$.
    We get 
    \begin{align*}
        \specnormof{\replaceLFPIt{\aterm}{i}}=\specnormconcof{(\bigEchoiceOf{\maxurg}\atermsetp)\appl(\bigEchoiceOf{\maxurg}\atermsetpp)}
        &= \bigEchoiceOf{\maxurg}\set{\specnormof{\atermp''.\atermpp''}\mid \atermp''\in\atermsetp, \atermpp''\in\atermsetpp}.
    \end{align*}
    Note that this distribution is only possible because all terms $\atermp''\in\atermsetp$ and $\atermpp''\in\atermsetpp$ have urgency $\maxurg-1$.
    Same as before, we observe $\atermpppp\discleqof{\objective}\specnormof{\aterm}$ if and only if there are $\atermp'',\atermpp''\in\wnfn{\maxurg-1}$ with $\atermpppp\discleqof{\objective}\specnormof{\atermp''.\atermpp''}$, $\atermp''\discleqof{\objective}\specnormof{\replaceLFPIt{\atermp}{i}}$, and $\atermpp''\discleqof{\objective}\specnormof{\replaceLFPIt{\atermpp}{i}}$.
    If no such $\atermp''$ and $\atermpp''$ can be found, this implies that $\atermpppp\discleqof{\objective}\replaceLFPIt{\aterm}{i}$ does not hold.

    Assume $\aterm$ contains no non-terminals.
    Then $\replaceLFPIt{\atermp}{i}=\atermp$ and $\replaceLFPIt{\atermpp}{i}=\atermpp$.
    The algorithm proceeds as follows.
    It first guesses $\atermp'', \atermpp''\in\wnfn{\maxurg}$ and verifies $\atermpppp\discleqof{\objective}\normof{\atermp''.\atermpp''}$ by the usual normalization procedure, which takes at most $(|\atermp''|+|\atermpp''|).\normcomplexity\leq\normcomplexity^{2}$ space.
    If this fails, the algorithm returns false.
    Then it removes all information except $\atermp''$, $\atermp$, $\atermpp''$, and $\atermpp$.
    The space required for this bounded from above by $4.|\aterm|.\normcomplexity$.
    It proceeds by checking $\atermp''\discleqof{\objective}\specnormof{\atermp}$ and $\atermpp''\discleqof{\objective}\specnormof{\atermpp}$, which is possible in $\max\set{|\atermp|^{2}, |\atermpp|^{2}}.\normcomplexity^{2}$ by the induction hypothesis.
    In total, the algorithm requires 
    $4.|\aterm|.\normcomplexity + \max\set{|\atermp|^{2}, |\atermpp|^{2}}.\normcomplexity^{2}$ non-deterministic space.
    We have 
    \begin{align*}
        4.|\aterm|.\normcomplexity + \max\set{|\atermp|^{2}, |\atermpp|^{2}}.\normcomplexity^{2}
        &=
        4.|\aterm|.\normcomplexity+ 3.\max\set{|\atermp|^{2}, |\atermpp|^{2}}.\normcomplexity^{2}\\
        &\leq (4.|\aterm|+ 3.{(|\aterm|-1)}^{2}).\normcomplexity^{2}\\
        &= (4.|\aterm|+ 3.|\aterm|^{2} - 6.|\aterm| + 3).\normcomplexity^{2} \\
        &= (3.|\aterm|^{2} - 2.|\aterm| + 3).\normcomplexity^{2} \\
        &\leq 3(|\aterm|^{2}).\normcomplexity^{2}
    \end{align*}
    
    Now conversely assume that $\aterm$ contains non-terminals.
    Because $\aterm$ is linear, one of $\atermp$ or $\atermpp$ does not contain non-terminals.
    W.l.o.g.\ let $\atermp$ be this term.
    The algorithm first guesses $\atermp'',\atermpp''\in\wnfn{\maxurg-1}$ and verifies $\atermpppp\discleqof{\objective}\specnormof{\atermp''.\atermpp''}$ using the usual normalization process in $\normcomplexity^{2}$ space.
    If this fails, the algorithm returns false.
    Similarly to the case without non-terminals, it then removes all information except $\atermp''$, $\atermp$, $\atermpp''$, and $\atermpp$.
    Following this, it verifies $\atermp''\discleqof{\objective}\atermp$.
    As before, if the check fails, it returns false.
    The term $\atermp$ does not contain non-terminals, so the check takes an additional space of only $|\atermp|^{2}.\normcomplexity^{2}$ by the induction hypothesis.
    The same calculation as in the no non-terminal case yields an upper bound $|\aterm|^{2}.\normcomplexity^{2}$ on total memory used to complete this step.
    Upon completing the step, the algorithm erases $\atermp''$ and $\atermp$ from memory.
    Then it checks $\atermpp''\discleqof{\objective}\atermpp$ and returns the result of the check.
    This check is possible in $\bigoof{\log(i)}+\max\set{|\atermpp|^{2}, |\eqmap|^{2}}.\normcomplexity^{2}$ non-deterministic space by the induction hypothesis.
    Note that the algorithm does not store any additional information besides the information required to perform this check.
    This yields the desired space complexity.
\end{proof}

By slightly modifying the approach in Lemma~\ref{Lemma:DFASimulation}, we can also get a complexity bound wrt. to $|\objective|$.
Instead of letting Eve resolve the transitions in the DFA for $|\objective|$, we let Adam resolve them.
We abuse the notation and refer to this translation with the symbol $\arrof{.}$ for the rest of this section.
$$\begin{aligned}
    \arrof{\aletter}&=\bigAchoiceOf{1}{}{\setcond{(\astate, \astatep)}
        {\transitions(\astate,\aletter) = \astatep}} 
        &
        \arrof{\terr}&=\terr
        \\
        \arrof{\bigPchoiceOf{\anurg}{}{\atermset}}&=\bigPchoiceOf{\anurg}
        {}{\setcond{\arrof{\aterm}}{\aterm\in\atermset}}
        & 
        \arrof{\tskip} &= \tskip
        \\
        \arrof{\aterm\appl\atermp}&=\arrof{\aterm}\appl\arrof{\atermp}
        &
        \arrof{\nonterminal}&=\nonterminal \,,
\end{aligned}
$$
The translation for the non-terminals also remains the same $\arreqmapof{\nonterminal} =\arrof{\eqmapof{\nonterminal}}$.
The transitions of the translated objective $\arrof{\objective}$ remain the same, however it accepts $\bot$ as well as the original final states.
This is because the task of choosing the correct transition is now assigned to Adam, and a failure leads to an Eve win.
The translation only introduces $\achoicen{1}$ at the lowest urgency level and does not modify the urgencies.
Thus, a weak linear term $\aterm$ and a weak linear grammar $(\nonterminals, \eqmap)$, the translated $\arrof{\aterm}$ and $(\nonterminals, \arrof{\eqmap})$ both remain weak linear.
The Lemma~\Cref{Lemma:DFASimulation} holds for this type of translation as well.
\begin{lemma}\label{Lemma:DFASimulationAdam}
    $\winsof{\aterm}{\objective}$ if and only if $\winsof{\arrof{\aterm}}{\arrof{\objective}}$. 
\end{lemma}

\begin{proof}[Proof Sketch]
    Letting Adam simulate the translations of the automaton does not change the winner.
    Determinicity still only allows for one choice of transition per $\arrof{\aletter}$ that does not lead to his loss.
    The proof follows the same steps as those of Lemma~\Cref{Lemma:DFASimulation}.
\end{proof}

The complexity wrt. $|\objective|$ follows.

\begin{lemma}\label{Lemma:EfficientWeakLinDFA}
    Let $(\nonterminals, \eqmap)$ be a weak linear grammar with maximal urgency $\maxurg$, $\aterm$ a weak linear term, and $\objective\subseteq\analph^{*}$ a regular objective. 
    Then $\winsof{\aterm}{\objective}$ can be decided in ${(\max\set{|\eqmap|, |\aterm|}.|\eqmap|.|\nonterminals|.\repexpof{\maxurg-2}{\bigoof{|\objective|^{2}}})}^{4}$ space.
\end{lemma}
\section{Lower Bound}\label{Appendix:LowerBound}
We prove the lower bound given in \Cref{Theorem:UpperLowerInput} with a reduction from context-bounded multi-pushdown games, a concurrent programming model the complexity of which is well-understood~\cite{QadeerR05,MW20}. 
The proof of the lower bound given in \Cref{Theorem:UpperLower} can be found in \Cref{Appendix:Decidability}. 

%
\subsection{Multi-Pushdown Games}
We introduce multi-pushdown games trimmed to our needs. 
A \emph{$\ctxtbound$-context-bounded 2-stack pushdown game~($\ctxtbound$-2PDG)} is a tuple $(\states, \evestates, \astate_0, \goalstates,  \Gamma, \transitions)$
consisting of a finite set of states~$\states$, 
a set of states $\evestates\subseteq\states$ owned by Eve, 
an initial state $\astate_0$, 
a set of goal states $\goalstates\subseteq\states$, 
a stack alphabet~$\Gamma$, and 
a set of transitions $\transitions\subseteq\states\times\stackops
\times\states$. 
Transitions are annotated by a stack operation from 
$\stackops=\Gamma\times\Gamma^{\leq 2}\cup\set{\nxt}$.
With $(\asymbol, \aword)$, we pop $\asymbol$ from and 
push $\aword$ to the active stack.
With $\nxt$, we change the active stack, called a context switch. 
We assume there is a bottom of stack symbol $\stlsymbol\in\Gamma$ that is never removed. 

The semantics of a $\ctxtbound$-2PDG is a game arena 
$(\configs, \to, \own)$ with a reachability objective $\configs_{\goalstates}$ for Eve.
The positions are configuration from 
$\configs=\states\times[0, \ctxtbound]\times\Gamma^{*}\times\Gamma^*$.
A configuration $(\astate, \currcs, s_1, s_0)$ stores the current state~$\astate$, 
the number of context switches~$\currcs$ that have occurred so far, and the contents of the two stacks. 
Stack $s_{0}$ is active after an even number of context switches, stack $s_{1}$ is active when $\currcs$ is odd. 
The owner and moves are defined as expected, there are no context switches beyond $\ctxtbound$, and we assume there are no deadlocks. 
%
%
The objective is $\configs_{\goalstates}=\goalstates\times\set{\ctxtbound}\times\Gamma^{*}\times\Gamma^*$, 
meaning we reach a goal state and have exhausted the context switches.  
Plays, strategies, and winning are defined like for urgency programs.
The task is to decide whether Eve has a strategy to win
from $(\astate_0, 0, \stlsymbol, \stlsymbol)$.
\begin{theorem}{\cite{MW20}}
    $\ctxtbound$-2PDG are $\kexptime{(\ctxtbound - 2)}$-complete.
\end{theorem}
\subsection{Reduction}
The reduction is in two steps, we first reduce 2PDG to the problem of making an observation:
\begin{proposition}{Repetition of \Cref{Proposition:LowerBound}}
    Given a $(2\maxurg+1)$-2PDG $\pdg$, 
    we can compute in poly time 
    $\aterm$ over $\analph$ and  $(\nonterminals, \eqmap)$ of maximal urgency~$\maxurg$ and an objective DFA $\objective$ 
    so that Eve wins $\pdg$ if and only if $\winsof{\aterm}{\objective}$.
\end{proposition}

We now reduce the problem of making an observation to the specialized contextual equivalence. 
Indeed, $\winsof{\aterm}{\objective}$ is the same as to check $\whichcharof{\objective}{\contextvar} \speccongleq{\objective} \aterm$, where $\whichcharof{\objective}{\contextvar}$ is the characteristic term of the empty context formed for objective $\objective$. 
The problems is that the characteristic term may be exponential. 
%
%
We utilize the trick from \Cref{Section:CompactEncoding}. 
\begin{lemma}
Given $\aterm$ over $\Sigma$ and $(\nonterminals, \eqmap)$, and objective $\objective$,  
we can compute in poly time $\whichcharof{\arrof{\objective}}{\contextvar}$, $\arrof{\aterm}$, and $\arrof{\objective}$ so that $\winsof{\aterm}{\objective}$ if and only if $\charof{\contextvar} \speccongleq{\arrof{\objective}} \arrof{\aterm}$. 
\end{lemma}

We sketch the proof of \Cref{Proposition:LowerBound}.
%
We encode positions $(\astate, \currcs, \stsymboln{}\ldots\stlsymbol, 
\stsymbolpn{}\ldots\stlsymbol)$ of the 2PDG as urgency terms
\begin{align*}
    \headerurg{\anurgpp}\appl
    \underbrace{\history\appl\wrpsymboln{}{\anurg}\ldots
    \wrplsymbol{\anurg}\appl\stacksep}_{s_1}
    \appl
    \underbrace{\historyp\appl\wrpsymbolpn{}{\anurgp}\ldots
    \wrplsymbol{\anurgp}\appl\stacksep}_{s_0}. 
\end{align*}
Stack symbols $\wrpsymboln{}{}\in\Gamma$ are represented by terms $\wrpsymboln{}{\anurg}$  
of urgency~$\anurg$.
Terminal $\stacksep$ marks the end of a stack content encoding.
The terms $\history$ and $\historyp$ represent the history of the play.
Finally, $\headerurg{\anurgpp}$ implements context switches.
The construction controls $\anurgpp$, $\anurg$,
and $\anurgp$ so that the top of the active stack is leading.

The top of the active stack allows the game to proceed as
\begin{align*}
    \ldots\history\appl
    \underbrace{(
        \Echoice_{\astate\in\states}
        \Pchoice_{\atrans\in\transitions_{\astate, \stsymboln{}}}
        \encoding{\atrans}^{\anurg}
    )}_{\text{Rewritten from }\wrpsymboln{}{\anurg}}
    \ldots\;\wrplsymbol{\anurg}\appl\stacksep\ldots
    \;\gamemove\;
    \ldots\history\appl\encoding{\atrans}^{\anurg}_{\stsymboln{}}
    \ldots\;\wrplsymbol{\anurg}\appl\stacksep\ldots
\end{align*}
Eve selects the current state  $\astate\in\states$. 
%
Then the player owning this state selects the next transition.  
We use $\transitions_{\astate, \stsymboln{}} \subseteq \transitions$ to denote the set of transitions
from state $\astate$ with top of stack symbol $\stsymboln{}$. 
The set is non-empty because the 2PDG does not deadlock.  
The term~$\encoding{\atrans}^{\anurg}$ of the chosen transition contains terminals which join history $\history$ to record the state change and the urgency~$\anurg$. 
%
%
%
%
The objective~$\objective$ is a product DFA that reads the terminals for each urgency separately and enforces consistency with the 2PDA transitions. 
%
%
%
%

Push/pop operations modify the active stack encoding in the expected way. %
For context switches, the leading term must swap the stack. 
To implement this, we use a decrement process on the now no longer active stack.
We define stack symbols~$\wrpsymbol{\anurg}$ as ${\encoding{\to\sty}^\anurg\appl\atermpp
\echoicen{\anurg}\encoding{\to\nxt}^\anurg\appl \wrpsymbol{\anurg-1}}$, where  $\atermpp$ is the choice of the next transition explained above.
%
The decrement process relies on the alternative $\encoding{\to\nxt}^{\anurg}\appl \wrpsymbol{\anurg-1}$, which replaces $\wrpsymbol{\anurg}$ by~$\wrpsymbol{\anurg-1}$. 
%
A snapshot of the decrement process is 
\begin{align*}
    &{
    \begin{tikzpicture}
        \draw[->] (0,0) -> node[above] {\footnotesize Progression of the leading term} (4,0);
    \end{tikzpicture}
    }\\[-0.5em]
    \ldots\history'\appl
    \encoding{\nxt}
    &\underbrace{\ldots
    \encoding{\to\nxt}\appl
    \wrpsymboln{i-1}{\anurg-1}}
    _{\text{Urgency }\anurg-1}\appl 
    \makeleading{\wrpsymboln{i}{\anurg}}
    \ldots\wrplsymbol{\anurg}\stacksep\ldots
\end{align*}
The terminals $\encoding{\to\sty}$, $\encoding{\to\nxt}$, $\encoding{\nxt}$, and $\stacksep$ allow the objective to check the decrement process for correctness. 

Each urgency simulates two contexts. 
Since we do not need to access the odd stack before the 
first context switch, we 
only generate this stack when it is first accessed.
This allows us to simulate three contexts with the maximal urgency.
In total, the construction simulates $2\maxurg + 1$ contexts  with urgency $\maxurg$.

\section{Lower Bound Details}

\textbf{Construction, Objective:}
The set of terminal symbols $\analph$ consists of assignments 
$\setvar{x}{val}$
and assertions $\assertvar{x}{val}$ of the variables 
$f_{\anurg}\in\states\cup\set{-}$,
$s_{\anurg}\in\states\cup\set{-}$, 
and $c_{\anurg}\in\set{\mathrm{nxt}, \mathrm{sty}}$
for each $0<\anurg\leq\maxurg$, along with a 
variable $gn\in\set{+, -}$.
In the parts of the play, where the urgency 
of the term is $0<\anurg\leq\maxurg$,
variable $f_{\anurg}$
will keep track of the first MPDG state,
the variable $s_{\anurg}$ will keep track of the 
latest MPDG state, and 
$c_{\anurg}\in\set{\mathrm{nxt}, \mathrm{sty}}$
will be used to enforce the correctness of context switches.
The variable $gn\in\set{+, -}$ keeps track of 
whether the game has 
generated the second stack by 
making the first context switch.
The objective DFA $\adfa$ 
processes the updates and assertions 
on the values of these variables.
For each $0<\anurg\leq\maxurg$, 
the DFA also keeps an assertion failure flag
$err_{\anurg}\in\set{\bot, \top}$, that records 
whether there has been 
assertion failure for $f_{\anurg}$, $s_{\anurg}$,
$c_{\anurg}$, and $gn$.
If an assertion failure happens for one of these 
variables, then $err_{\anurg}$ is irrevocably set to $\bot$.
In the initial state $\astateinit$, we have 
$gn=f_{\anurg} = s_{\anurg} = -$, $c_{\anurg} = \sty $, and 
$err_{\anurg}=\top$ for all $0 < \anurg \leq \maxurg$.
The DFA $\adfa$ accepts if and only if $s_{1}\in\goalstates$,
there are no assertion errors ($err_{\anurg}=\top$ for all 
$0<\anurg\leq\maxurg$),
and the latest states are consistent with the first states 
($s_{\anurg+1}=f_{\anurg}$ for all $0<\anurg<\maxurg$).

\textbf{Construction, Assignments:}
We now move on to the construction of the defining assignments.
Each stack symbol is represented by a different term for 
each urgency.
The set of non-terminals is 
$\nonterminals=\set{\ntsymbol{\anurg}\mid 
\stsymbol\in \Gamma, 0<\anurg\leq\maxurg}$.
The representation of an individual stack symbol
for urgency $\anurg$,
wraps the corresponding non-terminal in a unary 
choice with urgency $\anurg$.
Formally the representing term is the singleton choice $\wrpsymbol{\anurg}=\bigEchoiceOf{\anurg}\ntsymbol{\anurg}$
for all $0<\anurg\leq \maxurg$.
This ensures that the $\anurg$-representation of a stack 
symbol has urgency $\anurg$ (remember that non-terminals have highest urgency).
Furthermore, the term that represent the stack symbol must be
the leftmost action.
This allows a concatenation of terms 
that represent stack symbols to act 
like one stack in the MPDG.\@
The defining assignments $\eqmap:\nonterminals\to\terms$ 
are laid out below for all 
$\ntsymbol{\anurg}\in\nonterminals$.
We use helper terms to simplify the representation.
For all $\anurg\leq\maxurg$, 
$\anurgp<\maxurg$,
$\astate,\astatep \in\states$, $\stsymbol\in\Gamma$,
and $\aword\in\stsymbol^{\leq 2}$ we have:
\begin{align*}
    \eqmapof{\ntsymbol{\anurg}}
    &=(\assertvar{c_{\anurg}}{\nxt}
    \appl \wrpsymbol{\anurg - 1})
    \echoicen{\anurg}(\assertvar{c_{\anurg}}{\sty}
    \appl \popterm{\stsymbol}{\anurg})\\
    \popterm{\stsymbol}{\anurg}&=\bigEchoiceOf{\anurg}_{\astate\in\states}
    \assertvar{s_{\anurg}}{\astate}\appl
    \encoding{\transitions_{\astate, \stsymbol}}^{\anurg}\qquad
    \encoding{\transitions_{\astate, \stsymbol}}^{\anurg}=
    \bigPchoiceOf{\anurg}_{\atrans\in\transitions_{\astate, \stsymbol}}
    \encoding{\atrans}^{\anurg}_{\stsymbol}\\
    \encoding{(\astate, \stsymbol, \aword, \astatep)}^{\anurg}
    _{\stsymbol}
    &=\setvar{s_{\anurg}}{\astatep}\appl
    \aword_{0}^{\anurg}\ldots\aword_{n}^{\anurg}\\
    \encoding{(\astate,\nxt,\astatep)}^{\maxurg}_{\stsymbol}
    &=(\assertvar{gn}{-}\appl\setvar{gn}{+}\appl
    \wrplsymbol{\maxurg}\appl\stacksep\appl\wrpsymbol{\maxurg})
    \;\echoicen{\maxurg}\\
    &\qquad(\assertvar{gn}{+}\appl
    \setvar{s_{\maxurg}}{\astatep}\appl\setvar{c_{\maxurg}}{\nxt}\appl
    \wrpsymbol{\maxurg})\\
    \encoding{(\astate,\nxt, \astatep)}^{\anurgp}_{\stsymbol}
    &=\setvar{s_{\anurgp}}{\astatep}\appl\setvar{c_{\anurgp}}{\nxt}\appl
    \wrpsymbol{\anurgp}\\
\end{align*}
\begin{align*}
    \stacksep&=\setvar{c_{1}}{\sty}\ldots\setvar{c_{\maxurg}}{\sty}\\
    \headerurg{\anurg}&=\headerurg{\anurg - 1}\appl
    \bigEchoiceOf{\anurg}_{\astate\in\states}\setvar{f_{\anurg}}{\astate}
    \appl\setvar{s_{\anurg}}{\astate}\\
\end{align*}

The initial term for the game is simply 
$\headerurg{\maxurg-1}\appl
\makeleading{\wrplsymbol{\maxurg}}\appl\stacksep$.
The terminals $\encoding{\to\sty}^\anurg$ and 
$\encoding{\to\nxt}^\anurg$ used in the main paper 
refer to the assertions $\assertvar{c_{\anurg}}{\sty}$ 
and $\assertvar{c_{\anurg}}{\nxt}$.
At context switches in urgency $\maxurg$,
Eve also needs to ``guess'' whether the second stack has been generated.
In the case where it has not yet been generated, the correct choice
generates it.
In the case where it has already been generated, the correct choice 
triggers a context switch in the usual way.
\subsection{Denotational Semantics}\label{Appendix:Denotational}
We show how to define a denotational semantics based on our axiomatization.
What we find interesting is that, with the axiomatization at hand, the denotational semantics is a derived construct: the semantic domain and the interpretation of function symbols are induced by the axiomatization, yet the semantics is guaranteed to be fully abstract wrt. contextual equivalence resp. its specialized variant.
The creativity that is saved in the definition of the semantics of course had to be invested up front when coming up with the axiomatization.  
We found it easier to study an axiomatization than a denotational semantics, because the problem is narrowed down to understanding the interplay between operators as opposed to coming up with a freely chosen semantic domain. 
We recall the basics of denotational semantics before turning to the details.

A \emph{denotational semantics} for our programming language is a pair $((\semdom, \subseteq), \seminter)$ consisting of a complete partial order~$(\semdom, \subseteq)$ of semantic elements and an interpretation $\seminter:\setfunsymb\to \semdom^\omega\to \semdom$ that assigns to each function symbol $\afunsymb \in \setfunsymb$ in our language a monotonic function $f^{\seminter}:\semdom^{\mathit{ar}(\afunsymb)}\to \semdom$ of the expected arity. 
The function symbols $\setfunsymb$  are $\analph$, $\set{\tskip, \terr, \appl}$, and choices of arbitrary arity with urgency $1$ to $\maxurg$. 
We lift the interpretation to all terms~$\aterm$ and assign them an element $\densemof{\aterm}\in\semdom$, called the denotational semantics of the term. 
For recursion-free terms, the lifting is purely compositional: 
\begin{align*}
\densemof{\aletter}=\aletter^{\seminter}\qquad \densemof{\aterm\appl \atermp}=\densemof{\aterm}\appl^{\seminter}\densemof{\atermp},
\end{align*} and similar for the other function symbols. 
For the  non-terminals $(\nonterminals, \eqmap)$, this allows us to understand the defining equations as a monotonic function $$\eqmap^{\seminter}:(\nonterminals\to \semdom)\to \nonterminals\to \semdom.$$ 
The least solution of this function is the denotational semantics of the non-terminals: $\densemof{\nonterminal}=[\mathit{lfp}.\eqmap^{\seminter}](\nonterminal)$ for every $\nonterminal\in\nonterminals$. 
This is the missing case to define the semantics of arbitrary program terms again in a compositional way. 

We focus on the denotational semantics induced by the axiomatic congruence. 
The development for the $\objective$-specialized axiomatic congruence with $\objective$ right-separating is the same.
If~$\objective$ is not right-separating, we cannot give a guarantee that the resulting semantics will be fully abstract.   
The \emph{denotational semantics induced by $\axeq$} is $((\semdom_{\axeq}, \axleq), \seminter_{\axeq})$.  
The set of semantic elements is $\semdom_{\axeq} = \factorize{\terms}{\axeq}$, we factorize the set of terms along the axiomatic congruence. 
The complete partial order on these congruence classes is the one given by the axiomatic precongruence. 
It is guaranteed to be well-defined due to the precongruence.
It is guaranteed to stabilize in an ordinal by the fact that chains are well-ordered sets. 
The interpretation of the function symbols is as expected: 
\begin{align*}
\aletter^{\seminter_{\axeq}}\ =\ \axclassof{\aletter}\qquad \axclassof{\aterm}\ \appl^{\seminter_{\axeq}}\ \axclassof{\atermp}\ =\ \axclassof{\aterm\appl\atermp}.
\end{align*} 
Well-definedness holds because $\axeq$ is a congruence, monotonicity holds because $\axleq$ is a precongruence. 
The semantics is fully abstract wrt. contextual equivalence, 
$\densemof{\aterm}=\densemof{\atermp}$ iff $\aterm\congleq\atermp$, which is merely a reformulation of \Cref{Theorem:FullAbstractionContextual}. 
We can define other fully abstract semantics by introducing representative systems on the congruence classes, for example based on normal forms. 


\end{document}